\documentclass[12pt]{article}
\usepackage[top=1.25in,bottom=1.25in,left=1.25in,right=1.25in]{geometry}
\usepackage[dvips]{graphicx}

\usepackage{color}
\usepackage[x11names,rgb]{xcolor}

\usepackage{tikz}
\usetikzlibrary{arrows,decorations,backgrounds}
\usepackage{setspace}
\usepackage{amsfonts}
\usepackage{amsmath}
\usepackage{amssymb}
\usepackage{amsthm}
\usepackage{enumerate}
\usepackage{graphicx}
\usepackage{float}
\usepackage[font=small,labelfont=bf]{caption}
\usepackage{authblk}
\usepackage{subcaption}
\usepackage{float}
\usepackage[normalem]{ulem}
\usepackage[pdftex]{hyperref}
\usepackage[authoryear]{natbib}
\bibpunct{(}{)}{;}{a}{,}{,}
\hypersetup{colorlinks,%
	citecolor=blue,%
	filecolor=blue,%
	linkcolor=blue,%
	urlcolor=blue,%
}

\renewcommand{\vec}[1]{\mathbf{#1}}
\newtheorem{proposition}{{\bf \sc Proposition}}

\newtheorem{lemma}{{\bf \sc Lemma}}

\newtheorem{corollary}{{\bf \sc Corollary}}
\newtheorem{definition}{{\bf \sc Definition}}

\theoremstyle{remark} \newtheorem{example}{{\bf \sc Example}}
\def\eproof{\qed}

\usepackage[final]{changes}

	{\end{color}}

\begin{document}
	
	\title{{\Large Dynamic opinion updating with endogenous networks}\thanks{A previous version of this paper has circulated with the title ``Polarization when people choose their peers''. 
 We thank Yann Bramoull\'e,
			Benjamin Golub, Matthew Jackson, Alessandro Pavan, Roberto Rozzi, Alessandro Stringhi, Yves Zenou, and the
			seminar participants at the Adansonia Conference in Bocconi, at Asset 2019, at CTN 2018, at the 2nd Padua
			Workshop on Economic Design and Institutions, at UTS Sydney, at the Vanderbilt Fourth Annual Conference on Network Science and Economics, and at the Workshop on Dynamic Models of Interaction in Paris
			for their useful comments.
			Paolo Pin gratefully acknowledges funding from the
			Italian Ministry of Education Progetti di Rilevante Interesse Nazionale
			(PRIN) grants 2017ELHNNJ and P20228SXNF.}}
	
	\author[1]{Ugo Bolletta}
	\author[2]{Paolo Pin}
\affil[1]{{\small%
RITM Laboratoire de Recherches en \'Economie et Gestion, Universit\'e Paris Saclay, France. \url{ugo.bolletta@universite-paris-saclay.fr}}}
		\affil[2]{{\small%
 Department of Economics and Statistics, Universit\`a di Siena, Italy
		\ \& \ 
		BIDSA, Universit\`a Bocconi, Italy.  \
		\url{paolo.pin@unisi.it}}}
  	\date{May 2024}
	\maketitle
	\begin{abstract}
		Polarization is a well-documented phenomenon across a wide range of social issues. However, prevailing theories often compartmentalize the examination of herding behavior and opinion convergence within different contexts. In this study, we delve into the micro-foundations of how individuals strategically select reference groups, offering insight into a dynamic process where both individual opinions and the network evolve simultaneously. 
   We base our model on two parameters: people's direct benefit from connections and their adaptability in adjusting their opinions.
   Our research highlights which conditions impede the network from achieving complete connectivity, resulting in enduring polarization. Notably, our model also reveals that polarization can transiently emerge during the transition towards consensus. We explore the connection between these scenarios and a critical network metric: the initial diameter, under specific conditions related to the initial distribution of opinions.
	\end{abstract}
 
	
	
	\bigskip


%
 
    \section{Introduction}

    Political polarization has increased in several countries in the past decades and continues to do so \citep{draca2020polarized,boxell2022cross}. Because of the important implication of political polarization on economic decisions, researchers have investigated the issue and proposed an increasing number of potential explanations for the rise in polarization in the past few years. Among the causes of polarization, \cite{gentzkow2016polarization}, among others, has highlighted the impact of the digitalization of media outlets. Many authors argue that the arrival of the internet may have favored \textit{echo chambers} by providing more fragmented and less quality-oriented information.\footnote{See \cite{gentzkow2016polarization} for an extensive discussion of research papers that suggest the internet as an important driver of polarization} However, \cite{boxell2017greater} show evidence that polarization has increased more among individuals with little internet use, as opposed to more avid consumers of the internet. Recently, increasing attention has been devoted to the role of \textit{affective polarization}, which measures the feeling of negativity toward political parties different than the own. Yet, it is not clear whether affective polarization increases polarization, or the other way around.\footnote{The paper by \cite{boxell2022cross} documents that while affective polarization has increased substantially in the US, other countries have experienced different trends.} 
    In general, polarization seems to arise naturally in the presence of a disconnection in the patterns of relationships. If two groups do not interact, the opinions can evolve independently within each group, with no compromise between each other's opinions.
    
    We develop a model of network formation and opinion updating that allows us to analyze the co-evolution of the societal structure and professed opinions, focusing on the conditions for the society to disagree in the long run.
    We examine two crucial parameters: the benefit derived from connections and the adaptability of agents in adjusting their opinions. We establish precise conditions for network disconnection, which results in sustained disagreement. Furthermore, we pinpoint parameter ranges that drive the system towards 
    \added{\emph{temporary disagreement} (or \emph{self--correcting polarization}, when we want to stress the case of two communities)} and consensus. Lastly, for specific distributional assumptions, we characterize the entire dynamic system through a singular statistic: the network's diameter in the initial period.
    
    We depart from the literature on social learning because we assume that the object of debate, upon which individuals have subjective opinions, does not depend on a true state of the world. Individuals do not obtain utility by aggregating information and discovering the true state of the world but rather choose to profess their opinion in their community for the sake of debating.
    Examples of topics inherently independent from a true state of the world can be, but are not limited to, moral issues, political opinions, and the like. As members of a society, we constantly deliberate with others through social interactions. By doing so, we contribute to the determination of social norms which shape the overall functioning of our society. Since social norms are the outcome of a collective effort, it is only natural to consider how the tension between subjective opinions (the individual component) and professed opinions (the collective component) can affect both the structure of the society and the overall distribution of opinions over time.
    
    In our model, we assume that individual preferences depend on three main components: i) a conformist term, in the form of a quadratic cost for any deviation of an individual's endogenous professed opinion, and the professed opinion of her friends. In other words, individuals dislike disagreement among friends; ii) a taste for internal consistency, in the form of a quadratic cost for the deviation of one's professed opinion from her true opinion. Thus, individuals dislike professing an opinion different from their subjective opinion; finally, iii) individuals derive utility from direct friendships in the form of a linear benefit, so they enjoy having friends.\footnote{The linear benefit is equivalent to the term in the connections model found in the seminal paper of \cite{jackson1996strategic}. We assume there are no indirect benefits for connections at a distance greater than one (friends-of-friends), which is a special case of what is assumed in \cite{jackson1996strategic}.} Such preferences lead to best reply functions such that the individual's professed opinion is a convex combination of the average (professed) opinions in her group of friends and her true opinion.

    Within this framework, we study a discrete dynamical model in which, at every period, individuals interiorize their professed opinion, which becomes their own opinion in the next time step.
   \cite{golub2012homophily} study the \cite{degroot1974reaching} model, on which we base our step by step opinion updating. They provide actually two interpretation of the model: a coordination game and an opinion updating rule. Our interpretation lies in--between those two, and is consistent with the literature on how individuals update their preferences to align with past actions, and in general with the theory of cognitive dissonance, first introduced in psychology by \cite{festinger1957theory}. Such theory states that individuals seek to reduce the discrepancies between their subjective values and their actual behavior (in our case, professed opinions). See \cite{akerlof1982economic}, among others, for an economic approach to cognitive dissonance.
    
    In our model, the network also changes across time, as an effect of individual choices, and since opinions change, individuals will change the network accordingly. This allows us to study the co-evolution of networks and opinions over time. In terms of results, we find that the society can persistently agree or disagree, depending on whether individuals (endogenously) end up forming an integrated society or a segregated society with multiple groups. The basic relation between network structure and consensus (or disagreement) is well known,\footnote{The paper by \cite{golub2010naive} provides two necessary and sufficient conditions on the (exogenous) network structure for consensus of opinions to arise in the long run. Consensus is reached eventually if and only if the network is strongly connected. Thus, if we obtain that the network is strongly connected at any point in time, consensus will be reached in our model, too. On the contrary, disagreement will persist if the society is disconnected, i.e., two or more separated groups are present.} but we provide sufficient conditions on the distribution of opinions for segregation to arise at some point in time. Intuitively, if opinions are distant enough at some point in time, the network would disconnect. Moreover, we identify situations where individuals do not become more ``extreme'' in perpetuity but become more extreme only temporarily. \added{This is what we define temporary disagreement or self--correcting polarization.} In such a case, the overall distribution of opinions becomes polarized for some periods, although consensus still arises in the long run. In other words, polarization can be self-correcting. We provide additional results to show that, when this happens, the society spends a non-negligible amount of time in a polarized state. In fact, the speed of convergence to consensus reaches its minimum during the polarized state.

    To summarize, our analysis defines three main scenarios: long run consensus, when the society \added{ends up being} integrated,  \added{includes (i) monotonic convergence and (ii) temporary disagreement, while (iii) long run disagreement happens when}  the society \added{ends up being} segregated into two or more communities.
    \added{We corroborate our findings by studying analytically a special case of our model. We assume that the initial distribution of opinions is uniformly distributed and that the number of agents tends to infinity. 
    Under this assumption, we show analytically that the diameter of the network in the first steps of the dynamics, defined as the largest geodesic distance between any two individuals in $t_1$, is a strong predictor of the long-run distribution of opinions. In particular, even if the diameter is relatively low (i.e., at least 5) and so the society is fully integrated, but relations are sparse, then disagreement arises in the long run. On the contrary, if the society is strongly integrated and relations are dense enough since the beginning (i.e., the diameter is at most equal to 3), consensus will be reached (quickly). The intermediate value of  $4$ of the diameter can lead to the two scenarios mentioned above, but also  to temporary polarization of opinions. So, we determine specific values of the diameter of the network that allow us to classify the topology of long-run stable outcomes, which could be fruitful for further empirical research.}

    
    The rest of the paper is organized as follows. Section \ref{sec:literature} discusses the related literature. Sections \ref{sec:model} and \ref{sec:results} present the model and the general results of the paper, respectively. In section \ref{sec:uniform}, we provide further analytical results by making explicit assumptions on the initial distribution of opinions. Section \ref{sec:conclusion} concludes.
    \added{\ref{app:proofs} contains all our formal proofs.}
    In \ref{app:rational} we include an extension of the model with increased rationality from our agents. \ref{app:simulations} provide support through numerical simulations, while \ref{app:polarization_normative} provides further contextualization of polarization in our framework.
    
    \section{Related literature}\label{sec:literature}

    Since the paper is related to several branches of literature, we conceptually organize the review in sections. 
\subsection{Games on endogenous networks}
The current paper is not the first to consider a network game with an endogenous network. Among others, \cite{galeotti2010law} and \cite{kinateder2017public} consider strategic network formation in the context of public goods, \cite{boucher2016conformism} and \cite{badev2021} in the context of peer effects. None of these models consider the co-evolution of networks and behaviors. The paper by \cite{konig2014nestedness}, while considering a dynamic approach, focuses specifically on nestedness, defined as the situation where a neighborhood of a node is ``nested" into neighborhoods of nodes with a higher degree. Moreover, while strategic decisions are still part of the network formation process, the dynamics are driven stochastically. For these reasons, we believe their paper differs substantially from ours. We contribute to the literature by analyzing the impact of fully strategic linking decisions on opinion formation and its dynamics, providing a model that can explain the motives behind the polarization of opinions.
Recently, \cite{sadler2021games} contributed toward a general understanding of games on endogenous networks. They study games where agents form the network and create spillovers for their neighbors. Results show that we can obtain strong regularity conditions on the stable network structures by categorizing the nature of the spillovers (positive or negative) and whether the links are complements or substitutes of players' actions. The analysis is carried through considering both notions of stability of Nash equilibrium and pairwise stability. Notably, our analysis does not fit directly in their setup because, in our model, link choices depend on the average behavior of other agents. 
\subsection{Networks and polarization}

Several papers have studied the topic of networks and polarization. Among others, \cite{genicot2022tolerance} has used a similar model to ours to show the effect of tolerance on agents’ ability to compromise. She found that under specific distributions of tolerance levels, polarization may occur. While we keep the parameter defined as tolerance in \cite{genicot2022tolerance} as fixed, we focus on a fully-fledged dynamic model to analyze polarization. Thus, we find that polarization can increase regardless of the initial distribution of opinions, although we use the uniform distribution for some of our analytical results. In general, our model produces naturally cognitive dissonance. In fact, in the presence of social interactions, individuals' professed opinions differ from the subjective opinion because while there is an inherent taste for internal consistency, individuals dislike disagreeing with their friends and are thus brought to mediate their professed opinions. This mechanism is similar to what is referred to as \textit{compromise} in \cite{genicot2022tolerance}. 
The tension between internal consistency and disagreement with friends shapes the strategies for network formation. Naturally, individuals form connections with similar-minded individuals because, by doing so, they minimize cognitive dissonance and disagreement simultaneously. In our model, we show that there is an explicit trade-off between cognitive dissonance and variance of opinions in an individual group of friends. This trade-off drives most of the novel aspects of our model.  

Using a different approach, \cite{hegselmann2002opinion} studied the evolution of the network and opinions in the standard updating framework \'a la \cite{degroot1974reaching}. In the paper, the authors assume that individuals have \textit{bounded confidence}, and aggregate opinions until a certain threshold. Our model lets the same rule arise endogenously, providing a rationale for the exogenous rule imposed by \cite{hegselmann2002opinion}.  Rationalization aside, our model differs substantially from \cite{hegselmann2002opinion} because, being strategic, the individuals can, at any time, sever connections with similar opinions to join a group with more cohesive opinions, which cannot happen in \cite{hegselmann2002opinion}.\footnote{In Appendix \ref{app:HK} we explicitly compare the model in \cite{hegselmann2002opinion} with ours using numerical simulations to clarify the differences in the evolution of the network, driven by our choice to let the network arise endogenously.}

Related to this paper, \cite{anufriev2021dissonance} explicitly model \textit{dissonance}. In their paper, agents face a trade-off because they obtain utility by remaining true to themselves and aligning with their beliefs but simultaneously obtain utility by minimizing the distance between their and other agents' beliefs. In a one-shot model, agents may sever connections when the network is endogenous, which can lead to either consensus or polarization, depending on how costly it is to change the network. In this respect, our model follows a similar approach. While we do not model beliefs but rather professed opinions, we propose a canonical setup to study network formation, and analyze the one-shot, limit, and transitional properties of the dynamic system. 
    
Our model contributes to the literature on networks and polarization by showing a natural tendency for societies to become more polarized, even without external interference in the form of media and political elections, which we do not model. In particular, our model shows how the mass of ideologies in the center of the distribution can decrease, in line with the empirical findings by \cite{draca2020polarized}. In our framework, this phenomenon takes place because of the combined effect of the condensation of opinions far from the center and the capability of individuals to sever connections with like-minded individuals to join a cohesive group. They do so to minimize disagreement with their friends.

\added{Finally, a paper related to ours is \cite{gallo2020social}. In that work, a dynamical system is presented wherein agents autonomously update a weighted network, with links strengthened with a rule reminiscent of \cite{krause2000discrete}. Additionally, our paper investigates the asymptotic outcomes and the convergence speed.}

    \section{The model}\label{sec:model}

In our model, there are $n \geq 2$ agents, \added{indexed from $1$ to $n$}.
\added{We study a discrete dynamical system. 
At each time step, the agents are players in a subgame.
In this section, we present the dynamics and characterize the unique equilibrium of this subgame.}

At the beginning of each time step $t$, each agent $i$ holds a \textit{subjective opinion} $x_{i,t-1}\in \mathbb{R}$.
During each time step, each agent changes her connections in the directed social network and updates her \textit{professed opinion}  $x_{i,t}\in \mathbb{R}$ (which in general differs from their subjective opinion).
At the end of each time step, each agent's professed opinion becomes their new subjective opinion for the next time step.

The connections are recorded in an $n \times n$ adjacency matrix $G_t$ at every time step $t$. 
If at time $t$ an agent $i$ decides to connect to an agent $j$, we have that $G_{ij,t}=1$ (we fix $g_{ii,t}=0$ for every $i$ and $t$). 
\added{The linking choice $G_{i,t}$ of agent $i$ is a row of the adjacency matrix $G_t$, where $1$'s identify a subset of all the other players.}
We denote with $N_{i,t}=\{j|G_{ij,t}=1\}$ the out--neighborhood of agent $i$ at time $t$. 
The cardinality of the set of neighbors defines the degree of an agent, and we denote it by $d_{i,t}=|N_{i,t}|$. 

We assume agents are homogeneous, across time, with respect to a parameter $f\in(0,1)$ (for \emph{flexibility}), which measures the relative weight in their preference for the social component. Agents also assign a linear benefit $V\in\mathbb{R}^+$ to every connection they choose to form. 

Finally, at the beginning of the process, each agent $i$ is characterized by an idiosyncratic term $x_{i,0}\in[0,1]$ that defines their initial subjective opinion on a given topic of discussion, which are the initial conditions of the system.
\added{These initial opinions respect the indeces of players, so that if $i < j$, then $x_{i,0} \leq x_{j,0}$.}\footnote{We assume a uni-dimensional space of opinions. While this is a restriction, several topics of interest can be reduced to one single dimension, such as politics, taxation, moral issues etc. See \cite{demarzo2003persuasion} for an extensive discussion of uni-dimensional opinions and conditions to reduce a multi-dimensional space to a left-right spectrum.} 
So, $x_{i,0}$ is exogenously given \added{as initial condition of the process}, but $x_{i,t}$ will update over time. 
Throughout the paper, we use bold notation to indicate matrices and vectors, while normal notation indicates scalar values.

	    \subsection{Opinion updating within period}
	   \label{section:one period}

     Let's see what happens once the network is formed at each time step.
     Agents are \emph{temporally myopic}, so they focus only on the present time step.
     They care about making friends, while balancing the tradeoff between not disagreeing with their friends and maintaining internal consistency between their subjective opinion and their professed opinion.
      Formally, we assume that preferences take the following form: 
	    \begin{eqnarray}
	\label{eq:payoff}
	\pi_{i,t} (x_{i,t}, \added{x_{i,t-1}}, \vec{x}_{t}) & = &  \sum_{j \in N_{i,t}} \left( V - f (x_{i,t}-x_{j,t})^2 - (1-f) (x_{i,t} - x_{i,t-1})^2 \right) \nonumber \\
	& = &   d_{i,t} \left( V - (1-f) (x_{i,t} - x_{i,t-1})^2  \right) -  f \sum_{j \in N_{i,t}} (x_{i,t} - x_{j,t})^2 \ \ ,
\end{eqnarray}
\added{where $x_{i,t-1}$ is the subjective opinion, and $\vec{x}_{t}$ is the endogenous profile of expressed opinions.}
The payoff structure we use is well-known and leads to linear best-response functions. The preferences capture three main aspects: i) an interest in creating friendships, ii) a conformistic term, which represents the cost of disagreeing with friends and iii) a term for internal consistency, in the form of a cost for professing an opinion that differs from the subjective one. The terms i) and ii) are also the main terms through which the network becomes relevant.
To summarize, individuals dislike interactions that lead to some level of disagreement in professed opinions. If agents compromise, i.e. their professed opinion differs from the subjective opinion, their utility is reduced, too. However, having friends is intrinsically beneficial because every friend provides a gross benefit equal to $V$.

We call $\mu_{i,t} = \frac{\sum_{j \in N_{i,t}} x_{j,t}}{d_{i,t}}$ the mean of the professed opinions of $i$'s neighbors if there are any (that is if $d_{i,t}>0$). 
Otherwise, if $i$ has no out--links, $\mu_{i,t}=x_{i,t}$.
We call $\boldsymbol{\mu}_{t}$ the profile of all $\mu_{i,t}$, as we have defined them.
The unique best response for agent $i$ is the following convex combination determined by the level of flexibility, \added{expressing a system of linear equations, one for each agent $i$, where $\vec{x}_{t-1}$ are parameters and $\vec{x}_t$ are unknown}:
\begin{eqnarray}
\label{eq:bestresponse}
x_{i,t}^* (x_{t-1}, \vec{x}_{t}) &=& f \mu_{i,t} + (1-f) x_{i,t-1}. 
\end{eqnarray}
To begin with, we analyze one period only as if the network was exogenously fixed.
Then, we discuss the network's endogeneity and the model's dynamics.

If we call  
	$D_{ij,t} = \left\{
	\begin{array}{ccc}
		\frac{1}{d_{i,t}} & \mbox{ if $j \in N_{i,t}$,} \\
		0 & \mbox{ otherwise},
	\end{array}
	\right.$
\added{(the \emph{adjusted adjacency matrix})} then a compact way to write (\ref{eq:bestresponse}) is as follows:
\begin{eqnarray*}
	(I - fD_t) \vec{x}_t = (1-f) \vec{x}_{t-1} \ \ .
\end{eqnarray*}

In Lemma  \ref{lemma_unique},  in \ref{app:proofs},\footnote{%
Also the proofs of all the following results are in  \ref{app:proofs}. 
Note that we let $f$ be homogeneous, though the model could account for heterogeneity in this dimension. Lemma  {\color{red} \ref{lemma_unique}} actually proves the result in this more general case.} 
we show that the above equation has a unique solution $\vec{x}_t \in [0,1]^n$.
This means that the one-shot sub-game we base our analysis on has a unique Nash equilibrium.
This is crucial, allowing us to study the network formation process with an approach based on backward induction. Prior to this, we focus on the payoff structure and derive a formula for the \textit{payoff in equilibrium}, which is key to analyzing agents' strategies. 
\begin{lemma} \label{lemma_mean_variance}
	The payoff in equilibrium is
	\begin{eqnarray}
		\label{eq:payoff_equilibrium}
		\pi_{i,t} & =  d_{i,t} \left( V - f (1-f) (\mu_{i,t} - x_{i,t-1})^2  - f \sigma^2_{i,t} \right)   \ \ 
	\end{eqnarray}
	where we call $\sigma^2_{i,t} = \frac{\sum_{j \in N_{i,t}} (x_{j,t} - \mu_i)^2}{d_{i,t}}$ the variance of the professed opinions of $i$'s neighbors, at time $t$.
\end{lemma}
Given the quadratic payoff structure, it is  natural that the second moments of the distribution of behaviors appear in the analysis. 
Interestingly, we observe that fully flexible agents $(f \rightarrow 1)$ show a preference for homogeneous groups,\footnote{Formally, we constrain $f$ never to reach exactly 1. If $f=1$, agents would coordinate on a single opinion, although each opinion in the continuous space $[0,1]$ is potentially an equilibrium.} regardless of which opinion the group holds.
Therefore, they will form groups with agents with the most similar opinions. Conversely, if agents tend not to change their opinion $(f \rightarrow 0)$, others' opinions are completely irrelevant, so connections are formed at cost 0. At the intermediate values of flexibility, agents want opinions within the group to be both homogeneous and similar to theirs.

\subsection{Network formation within period}

As discussed above, agents choose whom to form directed links with at each time step before they update their opinion.
That is, they choose the other agents that they want to choose as peers.

In doing so, we assume that our agents are
\emph{strategically myopic}, in the sense that they do not consider that the subjective opinions of others will change with respect to their professed opinion.
In practice, if player $i$ evaluates at time $t$ whether or not to connect to player $j$, she takes into account $x_{j,t-1}$, without considering that what will actually affect her payoff at time $t$ is $x_{j,t}$.
As we show in \ref{app:rational}, \added{where we analyze the model with fully rational agents,}
 this helps solving problems of multiplicity of equilibria in the linking strategies.

Then, player $i$ chooses the subset of other players that maximizes equation \eqref{eq:payoff}, assuming (wrongly) that $\vec{x}_{j,t}=\vec{x}_{j,t-1}$, and computing (correctly) her best response 
$x^*_{i,t}$ to that, from equation \eqref{eq:bestresponse}.
From Lemma \ref{lemma_mean_variance}, 
she will choose the subset that maximizes the quadratic expression \eqref{eq:payoff_equilibrium} with respect to mean and variance.
 We define therefore the payoff that the players think they are going to obtain from a specific set of links. The realized payoff will in general differ because their realization depends on the updated opinions $\vec{x}_{t}$, while the linking choices are made on $\vec{x}_{t-1}$. Players form the network by maximizing the following formula:
\begin{equation}
    \Tilde{\pi}_{i,t}=d_{i,t}\left( V-f(1-f)(\Tilde{\mu}_{i,t}-x_{i,t-1})^2-f\Tilde{\sigma}_i^2 \right)
\end{equation}
where $\Tilde{\mu}_{i,t}=\frac{1}{d_{i,t}}\sum_{j\in N_{i,t}}x_{j,t-1}$, and similarly $\Tilde{\sigma}_{i,t}^2=\frac{\sum_{j\in N_{i,t}}(x_{j,t-1}-\Tilde{\mu}_{i,t})^2}{d_{i,t}}$.
This choice will imply a mean ${\mu}_{i,t}$ of the actions of her chosen peers, and then she will update 
\begin{equation}
	\textbf{x}_{i,t}=f {\mu}_{i,t}+(1-f)\textbf{x}_{i,t-1}.
\end{equation}

Formally, the linking strategy of each player $i$ at each time step $t$ is a subset $G_{i,t}$ of all the other players \added{(or equivalently, as discussed above, a row of the adjacency matrix $G_t$)}: those that player $i$ decides to link to.
\added{We call \emph{configuration} a state $(\vec{x}_t, G_t)$ of our dynamical system at time $t$.} 

\added{In order to make our dynamical system deterministic, we impose a tie--breaking rule for the generic possibility of indifference.
If two or more subsets of other players can be chosen, then a player will adopt lexicographic order based on the indices of the subset members.
This specific rule, with respect to any other, does not change any of our following results.
}

\subsection{Timing of each time step}

Summing up, the full model that we analyze is in discrete time. At every time step, agents start with some opinions, decide how to establish links so that a network realizes, and then they play the sub--game of opinion updating.
We assume that agents are myopic, not only in the temporal sense that they maximize only the payoff one step ahead but also in the strategic sense, because they do not anticipate that the other players will change their opinion when adapting in the present time step.
Formally, the dynamic best reply functions are now given by the following compact form expression: 
\begin{equation}\label{eq:degroot}
	\textbf{x}_t=f\boldsymbol{\mu}_{t}+(1-f)\textbf{x}_{t-1}
\end{equation}
So, starting conditions are the initial opinions $\vec{x}_0$.

The timing of our process, starting at $t=1$, is as follows:	
\begin{definition}[Timing]\label{def:timing2}
	In each time step $t$ of the dynamic model \added{we have three sub--steps}:
	\begin{itemize}
		\item \textbf{Time} $t^1$: Opinions $\textbf{x}_{t-1} $ are given from previous time step to agents;
		\item \textbf{Time} $t^2$: Agents simultaneously and independently form the directed network, which becomes common knowledge -- they decide as if opinions were fixed;
		\item \textbf{Time} $t^3$: $\textbf{x}_{t} $ is formed according to the best replies with respect to opinions and the network.
	\end{itemize}
\end{definition}

\added{We are now ready to study how our model evolves in time.}

    \section{Results}\label{sec:results}

	Lemma \ref{lemma_mean_variance} allows us to highlight the main trade-off agents face when choosing their links. On one hand, agents want to build a group such that the group's average opinion is as close as possible to their ideology because it would minimize the social influence. On the other hand, groups where ideologies are dispersed cause disutility. Thus, there are situations where agents might be willing to drop connections that minimize the social influence to reduce the group's variance. 

\subsection{Ordered configurations}

 	\added{To begin}, we provide a  definition and a result that allow us to  characterize \added{the states of our dynamical system.}
  
	\begin{definition}\label{def:ordered}
		\added{At time $t$, a configuration of network $G_t$ and opinions $\vec{x}_t$} is \textbf{ordered} if for two nodes $i<j$, \added{we have $x_{i,t} < x_{j,t}$, and,} if they both have non--empty set of peers $G_{i,t}$ and $G_{j,t}$, we have that
		\begin{itemize}
			\item for any $\ell \ne i$, $\ell \notin G_{i,t} $, either $x_{\ell,t} < \min \{x_{h,t} : h \in G_{i,t} \} $ or $x_{\ell,t} > \max \{ x_{h,t} : h \in G_{i,t} \}  $;
			\item for any $\ell \ne j$, $\ell \notin G_{j,t} $, either $x_{\ell,t} < \min \{ x_{h,t} : h \in G_{j,t} \}  $ or $x_{\ell,t} > \max \{ x_h : h \in G_{j,t}  $;
			\item $ \min \{ x_{h,t} : h \in G_{i,t} \} \leq \min \{ x_{h,t} : h \in G_{j,t} \} $ and $ \max \{ x_{h,t} : h \in G_{i,t} \} \leq \max \{ x_{h,t} : h \in G_{j,t} \} $. 
		\end{itemize}
	\end{definition}
	
	In the above definition, we consider a \added{configuration} \emph{ordered} if agents are matched only with a set of neighbors that, in terms of \added{opinions},  has no \emph{holes}, meaning that all players whose \added{opinions} can be obtained from convex combinations of the opinions of those chosen, should also be included.
	Moreover, the set of peers chosen are shifted according to the order of initial opinions.
	Here below we show that \added{a configuration} of the \added{dynamics} is always ordered, \added{as a result of the fact that} the ordering of ex-ante opinions $\vec{x}_t$ is preserved over ex-post opinions ${x}_{t+1}$. 
	
	\begin{lemma}\label{res:ordered}
		\added{In our dynamical model, at any time $t$, the configuration of network and opinions} is ordered. 
	\end{lemma}
	
	To get some intuition on this result, recall equation \eqref{eq:payoff_equilibrium} interpreting the payoff deriving from a set of links as follows. First, agents want the average behavior in their neighborhood to be consistent with their initial opinion. In other words, agents are \textit{homophilous}. Therefore agents with similar opinions will tend to group together. In addition, consider the second term of equation \eqref{eq:payoff_equilibrium}, which shows agents' preferences regarding diversity of opinions inside their neighborhood. It is clear that agents will tend to form groups with the agents most similar to them, and the reference groups are well defined since the variance term ensures sharp bounds.

    \subsection{Long--run consensus or disagreement}\label{sec:sufficient}
    
    We know  from \cite{golub2010naive} that, generically, on an exogenously fixed network, \added{opinions can either converge to a unique opinion  or to more. In their framework, this depends solely on the network and whether it is connected or not. In our model, the network is endogenous, but there are still properties of the network that once obtained become irreversible. 
We remind that in a directed network \emph{components} are defined as a partition of the nodes such that two nodes  are in the same component if they are connected, even indirectly, and no two nodes in different components are connected.
The next result is based on the formation and evolution of components in our dynamical network.}

\begin{proposition}\label{res:empty_myopic}
\added{Given initial conditions $V \in \mathbb{R}_+$, $f \in (0,1)$ and $\vec{x}_0 \in \mathbb{R}^n$. The dynamical system defined by our model always converges asymptotically to some $\vec{x}_\infty \in \mathbb{R}^n$. At the steady state, any connected component is characterized by all agents having the same opinion and forming a complete sub--network between them.}
	If the first network, at time $t=1$, is empty, then \added{$\textbf{x}_{\infty}=\textbf{x}_{t}= \textbf{x}_0$} for any $t$.
 \end{proposition}

\added{The intuition behind this characterization lies in the fact that in our model, at any time step, the network is divided in components (if the network is connected the unique component is the whole set of players).
In the next time steps, players from two different components cannot connect together again because the condition that made players decide not to connect cannot improve.
So, a component can only disconnect, and eventually, this partition of the players can only (weakly) increase the number of these components.
If a component does not disconnect, its players will converge to the same opinion and, as they become closer, find it all profitable to connect together.}

\added{In the following, we will talk about \emph{consensus} when all agents converge to the same opinion and to \emph{persistent disagreement} when the system converges to $2$ or more separate components with different opinions.
From Proposition \ref{res:empty_myopic} this distinction covers all the possibilities.}


We are interested in characterizing situations where the network becomes endogenously disconnected. In the literature of opinion dynamics, where the network is often exogenously given, disagreement in the society is trivially a consequence of a disconnected network. In our model, the players choose the network strategically, and thus we can focus on the initial conditions that lead eventually to a disconnected network. The model of strategic network formation is entirely specified for a given initial distribution of opinions, the value of the parameter that determines the value of connections $V$, and the parameter $f$ that weights the relative importance of social influence and own opinion. 
We start with the specification of a local sufficient condition.

\begin{proposition}\label{cor:suff_cond}
	There is a threshold $\xi$, depending on $f$ and $V$, such that if there exist two adjacent agents $i$ and $i+1$, and a time $t$, such that $|x_{i,t}-x_{i+1,t}|>\xi $, then all networks  are disconnected for any $t' > t$.\\
	This threshold is $\xi=\left( \frac{V}{f(1-f)}\right)^{1/2}$.	
\end{proposition}

This result characterizes local \added{sufficient} conditions for the network to become disconnected \added{at some} time $t$.          
\added{While the network can be connected initially, it can still disconnect if there exists a time $t$ such that the conditions are satisfied for two agents $i$ and $i+1$.}
\added{If this happens,} for any time after $t$, the network remains disconnected. 
Our model generally predicts that opinions within each component of the network will converge to the average opinion of the same component. When agents are connected, either through direct or indirect connections, their opinions will regress to the mean opinion of the group. This is a natural consequence of the best response functions, which entail averaging opinions within a group of individuals. Hence, in order to understand polarization, it becomes crucial to understand when the population becomes divided. 

By studying conditions that lead to a fragmented society, we can characterize the situations where polarization can increase over time until it reaches a point of stabilization where it no longer changes. \added{If we imagine any discrete distribution of opinions, and we identify that, for one couple of agents, the condition of Proposition \ref{cor:suff_cond} is satisfied, then the network at the limit is composed of two or more components. If so, the opinions (at the limit) will be the same within each component but far apart between the  components, a typical case of polarization (see \ref{app:polarization} for the  formal definition)}. Notably, Proposition \ref{cor:suff_cond} does not require that the opinions of two agents are distant enough as an initial condition. This can happen dynamically in our model, even though the network has been connected for some time. 
    
The polarization of opinions is a natural outcome in our model. There are two main reasons for this: i) extreme agents can only update their opinion by getting closer to the center mass of opinions, and ii) agents can trade off connections in order to join a cohesive group of individuals whose opinion is further from the center. The first point is not novel. In the paper by \cite{hegselmann2002opinion}, the authors identify that agents with \textit{extreme opinions} are under \textit{one sided influence}, which leads to a \textit{condensation} of opinions in the extremes of the distribution.\footnote{See Pag. 12, Figure 4 in \cite{hegselmann2002opinion}. We discuss more in depth the relations between our model and that model in Appendix \ref{app:HK}.}
The overall range of opinions \textit{shrinks}, but the condensation attracts agents from the center of the distribution, causing polarization. We have the same effect in our model. The main difference is that condensation is the consequence of agents' decisions. Therefore, the condensation arises endogenously, corroborated by the strategic link formation, which leads us to discuss point ii). We recall the result from Lemma \ref{lemma_mean_variance}. In our model, players can trade off connections toward the center of the distribution of opinions to join the condensed, cohesive group close to the extreme. In other words, players may choose to profess an opinion further from their true opinion, increasing their cognitive dissonance to join a cohesive group. This reinforces the results by \cite{hegselmann2002opinion} introducing a novel theoretical aspect to the origins of polarization. In the next section, we analyze this aspect by characterizing situations where the polarization of opinions may not be a persistent outcome but could be self-correcting.

	
	
    \subsection{Temporary disagreement}\label{sec:nonmono}

        When looking at the recent evidence about polarization in the real world, the causes remain unclear. Many have explored the role of media and the advent of the internet in explaining the polarization of opinions.\footnote{Among others, \cite{boxell2017greater} show that internet penetration is not associated with increased polarization. On the other hand, the paper by \cite{allcott2020welfare} shows that reduced access to social media significantly decreases several measures of polarization among individuals. } Yet, there is no conclusive evidence. Our model allows us to identify a novel explanation for the rise in polarization through endogenous network formation. In the previous section, we have identified a local condition on the distribution of opinions for polarization to occur in a society. In particular, Proposition \ref{cor:suff_cond} tells us under which conditions the network can disconnect. When that happens, the society remains perpetually disconnected, and their opinions are characterized as \textit{persistent disagreement}. 
        However, this result \added{describes} the convergence properties of the model \added{and disregards the evolution of opinions during the transition to the steady state, while} polarization can occur temporarily in a population that converges to consensus. Therefore, it is important to characterize the evolution of societal polarization further. In particular, it can be interesting to gather insight into whether polarization is a persisting outcome or can be transitory and self-correct over time. Our model predicts both of these outcomes via endogenous network formation. 

    As discussed in the previous Section, polarization can arise in our model due to the condensation of opinions toward the extremes of the distribution. This phenomenon attracts agents from the middle of the distribution, causing an increase in polarization. From Proposition \ref{cor:suff_cond}, we know that if two adjacent agents disconnect at some point, the whole network becomes disconnected, and polarization will strengthen over time (up to a limit) and remain persistent. However, this does not necessarily happen. We can observe an increase in polarization due to the condensation of opinions in the tails of the distribution of opinions. If the network never becomes disconnected, then there will be a time $t$ where consensus will be achieved, as we know from \cite{golub2010naive}. In this instance, we can observe polarization of opinions, which self-corrects over time. Also, some agents necessarily exhibit non-monotonic updating of opinions.\footnote{%
    \added{%
    The fact that opinions can diverge before they reach a consensus is not novel. It is observed already by \cite{golub2010naive} in the DeGroot model, and it is obtained, in the context of cultural transmission, in the model analyzed by  \cite{panebianco2014socialization}, \cite{prummer2017community}  and \cite{panebianco2018convergence}.}} 
    From the data, it can be difficult to explain how some individuals can change their opinions in different directions over time. Our theory provides a rationalization of this phenomenon because this case remains fully consistent within our strategic framework. In Figure \ref{fig:non-monotone}, we show numerical results to corroborate the discussion using a uniform distribution. 
    \added{In particular, panel $(b)$ of Figure \ref{fig:non-monotone} shows a case of \emph{temporary} polarization, that lies in--between segregation of the society and monotone convergence to consensus.
    This is a case in which an external observer could assess polarization, while the existence of few individuals who connect separate communities will eventually lead the process to consensus. 
    In this sense, our model suggests that, in order to predict the evolution of opinions in a population, it may not be enough to look at the distribution of opinions, but also at the network of connections through which opinions evolve.
   In Section \ref{sec:uniform}, we provide  conditions for this type of behavior to happen.
    }

    We provide extensive numerical results with different starting distributions in Appendix \ref{distribution}. Moreover, to accompany Figure \ref{fig:non-monotone}, we provide visualization of the network for all time steps in Appendix \ref{app:networks}, with consistent parameters for the numerical results of Figure \ref{fig:non-monotone}, and the associated measures of polarization in \ref{app:polarization}.

    \begin{figure}[ht]
		\centering

		\begin{subfigure}[t]{0.485\textwidth}
			\centering
			\includegraphics[width=\textwidth]{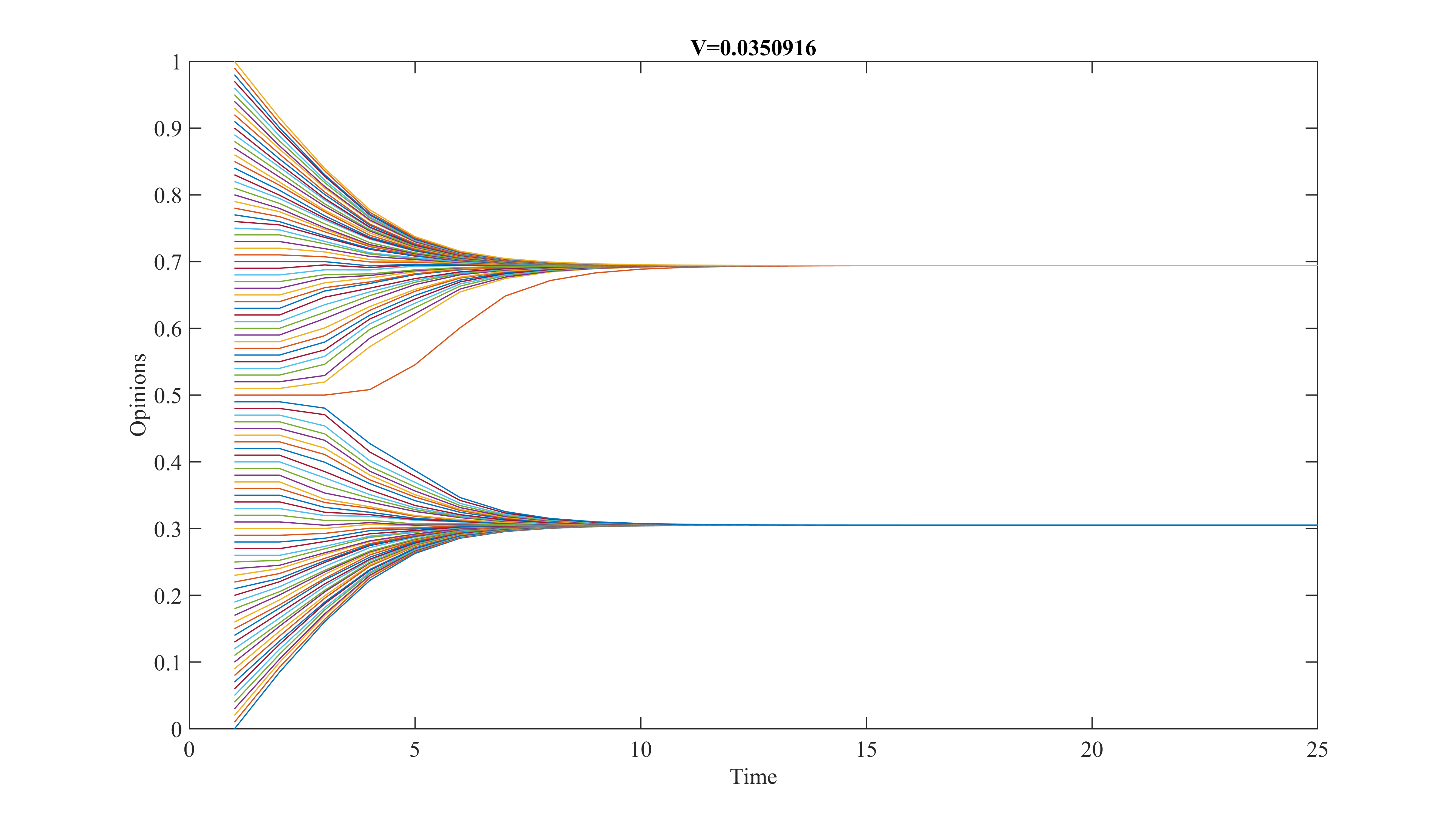}
			\caption{$f=0.5$ and $V=0.0350916$: The network ends up disconnected, and polarization persists in the society.}
		\end{subfigure}
    		~
		\begin{subfigure}[t]{0.485\textwidth}
			\includegraphics[width=\textwidth]{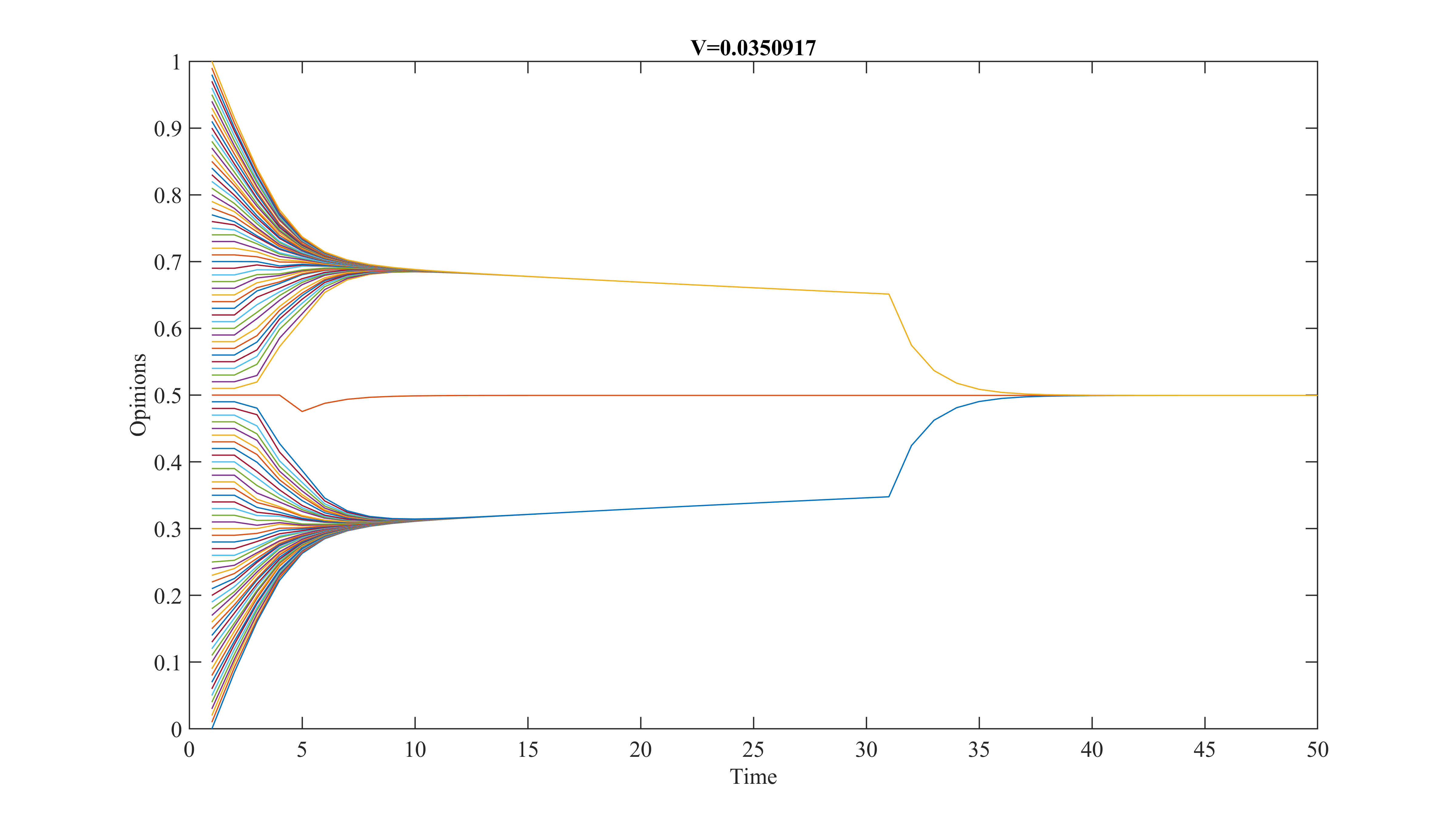}
			\caption{$f=0.5$ and $V=0.0350917$: Passed a threshold value, the opinions reach a consensus in the long run. Some agents display non-monotonic updating of opinions.}
		\end{subfigure}

		\caption{Plot of the dynamics starting with a uniform distribution, $n=101$ and intermediate values of $V$.}\label{fig:non-monotone}
	\end{figure}

    Obtaining analytical results on non-monotonic updating is challenging. It entails keeping track of all players and their decisions over time and the distribution of opinion evolves accordingly. Nonetheless, we are interested in investigating two further issues: i) the speed of convergence in our model and ii) the parameter space for convergence, self-correcting polarization, and persistent disagreement to occur. We analyze point i) in the next Section. We compare our model's convergence speed to the benchmark given by the case of the exogenous network studied in \cite{golub2012homophily}. Also, we investigate specifically the case of self-correcting polarization to understand the duration of a polarized distribution of opinions. Point ii) is discussed extensively in Section \ref{sec:uniform}, where we further investigate the relationship between the model's initial conditions and long-term outcomes.

    
    \subsection{Speed of convergence}
    \label{sec:speed}

    We have established conditions for different long-term outcomes of the dynamic process with strategic network formation. While it is relevant to understand how opinions will be distributed in a society in the long run, we believe it is also important to study what happens during the transition. In this section, we study the speed of convergence, focusing on the situations when consensus happens. To compare our model to a benchmark, we use the \cite{golub2010naive,golub2012homophily} model as the reference point.\footnote{While the model in \cite{golub2010naive,golub2012homophily} substantially differs from ours because we assume that there is not a true state of the world, we refer specifically to the results pertaining to the speed of convergence, which do not depend on this stark difference between the two models.} 
    
    

    In this section, we analyze the speed of convergence of this process, taking as a benchmark the case in which the network remains exogenously fixed during the whole process.
We will consider two dynamics, one in which the network remains fixed as in the first step of the dynamics of our model for all the future steps, and one in which the network evolves endogenously as in our model. 
We show that, if the endogenous process converges to consensus, the endogenous process is asymptotically faster than the exogenous one.
For this, we need to recall that $D_0$  is the adjusted adjacency matrix of a network.
We call $W_0=(1-f)I + f D_0$ the \emph{hearing matrix} of the process, the stochastic matrix used by our agents to update their opinions.
When we consider the endogenous case we need a time index: at time $t$, $D_t$ will be the adjusted adjacency matrix and $W_t$ the hearing matrix.
In the following, $|| \cdot||$ is any norm in $\mathbb{R}^n$.

\begin{proposition}
\label{prop:time}
Suppose that $D_0$ is not the complete network, and that the endogenous process converges to some consensus ${\bf x}_{\infty}$. Then, there is $n$ large enough and  $T>0$, such that, for all $t \geq T$:
$ ||  W_0^{t} \textbf{x}_0 - W_0^{\infty} \textbf{x}_0 || > ||   \textbf{x}_t - \textbf{x}_{\infty} ||$.
\end{proposition}

In the proof of Proposition \ref{prop:time}, we do not make any assumption on the initial distribution of links. 
If we assume that the initial distribution of opinions is symmetric around $\frac{1}{2}$, then by symmetry, we will have that $ W_0^{\infty} \textbf{x}_0 = {\bf x}_{\infty}$, meaning that the two processes converge to the same common opinion.
The proof is based on the fact that the size, in absolute value, of the second largest eigenvalue of the hearing matrix $W_t$ is inversely correlated with the speed of convergence.
Our endogenous process reaches, at some point, the complete network.
The complete network is the network for which this size is the smallest possible, so the speed of convergence reached once the network is complete will be the maximum possible.

Bearing this in mind, it is interesting to analyze the dynamics from panel (b) in Figure \ref{fig:non-monotone}.
Figure \ref{fig:evolution_uni_connected1} in Appendix \ref{app:networks} reports all the network configurations of this dynamic process.
Figure \ref{fig:secondlargesteigen} reports the evolution of the second largest eigenvalue of $W_t$ along this path, showing clearly that this evolution is not monotone.
\added{During the first steps of the dynamic process, the speed of convergence is slow (and gets slower)}, and then suddenly, when the network becomes complete, as is used in the proof of Proposition \ref{prop:time}, it reaches the maximal possible speed of convergence.
Even more surprisingly, along all the steps that seem polarized (from step 5 to step 30 in that simulation), the network is made out of two separated components that are complete and bridged only by a few central nodes. \cite{powers1988graph,powers1989bounds} shows that this is exactly the network configuration where the second largest eigenvalue of the hearing matrix is the largest, meaning that this is exactly the situation in which, among all the connected networks, the speed of convergence will be slowest.
Nevertheless, as proven by Proposition \ref{prop:time}, in the long run, this process will reach a complete network and converge asymptotically faster than if the network was exogenously fixed and not complete.

\begin{figure}[ht]
	\centering
	\includegraphics[width=\textwidth]{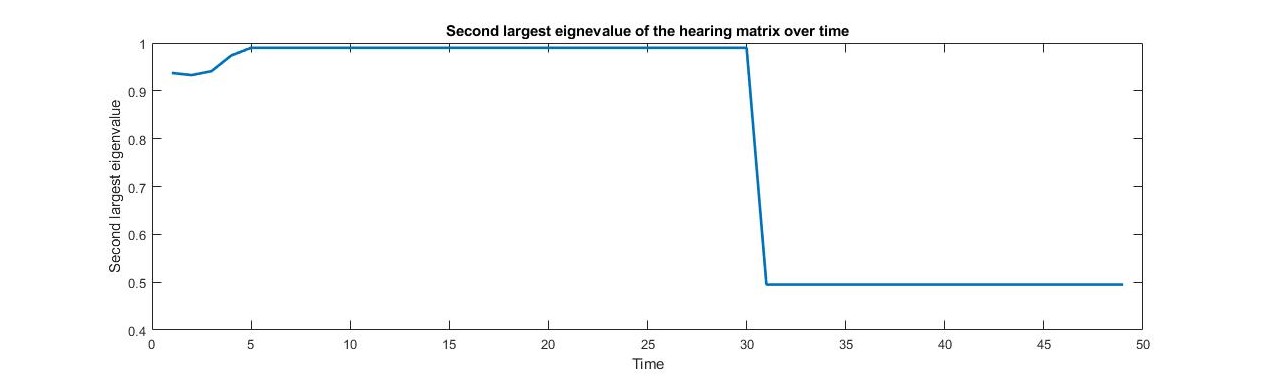}
	\caption{This figure shows the evolution of the second largest eigenvalue of the hearing matrix $W=(1-f)I+fD$ using the parameters of Figure \ref{fig:non-monotone}b. The curve shows that the evolution of this value is non-monotone. Importantly, during the periods where polarization is at its highest (time 5-30), the second largest eigenvalue reaches its theoretical maximum, so the speed of convergence is at its minimum. Polarization persists for a relatively long period of time.\label{fig:secondlargesteigen}}
\end{figure}


    \section{Predicting long-run outcomes from initial conditions: the case of the uniform distribution}\label{sec:uniform}

    Summarizing the findings of the previous section, we observe overall that for low benefit values, the network is likely to be disconnected,
exhibiting multiple components. \added{Consequently, the opinions exhibit }disagreement over time. If we increase the benefit, there is a region where the network is connected, so opinions will reach consensus. However, an interesting pattern emerges: some agents’ opinions evolve non-monotonically over time. Finally, if we increase the benefit even more, all opinions will evolve monotonically, and the speed of convergence will be consistently higher. The intuition is simple: by increasing the benefit, we are increasing the density of links, and this largely, but not solely, explains the mechanism behind our findings.
\added{To explore this issue further, we focus on a special case that relies on two strong assumptions: the agents are uniformly distributed, and the number of agents tends to infinity. So, agents in $t_0$ are equally spaced from each other. This choice has two motivations.}

\added{First, in Section \ref{sec:sufficient}, we have provided a sufficient condition for the network to become disconnected. Our result relied entirely on a local condition that can manifest at any point in the dynamics. In Section \ref{sec:nonmono}, we have also discussed that there can be a fine line between the case of persistent disagreement and consensus and that non-monotonic updating of opinions can occur for some agents. In light of these findings, by studying the uniform distribution, we can investigate the thresholds in parameters that distinguish between those two important cases while imposing the least restrictions on the initial distribution. In other words, the uniform distribution provides a conservative benchmark to study the evolution of the network and the distribution of opinions.}



Second, studying the uniform distribution has an undeniable technical advantage as we can obtain sharp results on the initial conditions and how they contribute to long-term outcomes. Moreover, we can do so while studying a neutral benchmark, as opposed to more general distributions that naturally imply the presence of one or more modes, which could favor (e.g. a bimodal distribution) or hinder (a uni-modal distribution, like the normal distribution) polarization and disagreement.\footnote{We provide in Appendix \ref{distribution} numerical simulations for these two cases. The parameter space for persistent disagreement increases substantially if we impose a bimodal distribution. On the other hand, starting from a normal distribution, the model cannot exhibit a perfectly bi-polarized persistent disagreement because a mass of agents will always persist in the middle of the distribution of opinions. However, the society can still become fragmented, exhibiting three poles.}

The approach here is thus to characterize the payoff for an extreme agent as a function of the number of links that the agent could make in the first step of the dynamics (because of this, in this section we omit the subscript of time). Let us call that agent $e$. Since the distribution is uniform, we just plug in the statistic for that distribution. 
If we let the extreme agent link with $k$ agents, since the optimal network will be ordered, agent $e$ will be linked from the closest agent to the $k^{th}$. 
However, we know that $e$ has an initial opinion of $0$, and the second and closest agent has an opinion of $1/n$, whereas the $k^{th}$ agent has an opinion of $k/n$. 
Therefore agent $e$'s payoff will follow from equation \eqref{eq:payoff_equilibrium}, which depends on the average and the variance. Such statistics in the uniform distribution will be respectively  $\mu_e=\frac{k-1}{2n}$ and $\sigma^2_e=\frac{(k-1)^2}{12}$.\footnote{We approximate the statistics with a continuous distribution. However, that choice is justified by the fact that our simulations are based on a large number of agents, most of the figures reported here showing results for $n = 101$ agents.} 
After plugging in the local statistics for the uniform distribution, we are able to express payoff for agent \textit{e}, abusing some notation, as a function of $ V$,$f$,$k$ and $n$: 
\begin{equation}\label{eq:uniform_extreme}
\pi_e(V,f,k,n)=k \left(V-\frac{1}{4} (1-f) f \left(\frac{k-1}{n}\right)^2-\frac{1}{12} f \left(\frac{k-1}{n}\right)^2\right)
\end{equation}

A similar exercise can be performed with any agent that will form a \textit{symmetric group}. 
By symmetric group we mean that an agent will form an equal number of connections on her left and on her right. Again, considering that agents are all equally spaced in the very first period, we can also determine the optimal $k$ for a central agent called $c$. 
Provided this equation, we can now maximize for $k$ and thus find the size of the groups (at left and right) for a central agent $c$:
\begin{equation}\label{eq:uniform_central}
\pi_c(V,f,k,n)=k \left(V-\frac{f \left((2 k+1)^2-1\right)}{12 n^2}\right)
\end{equation} 


One aspect worth investigating is the relation between the initial density of links and the long-run outcome of the model. Thanks to the use of a uniform distribution, we can calculate conditions on parameters that affect this diameter in the initial step.

\begin{proposition}\label{res:eqdiameter}
	There exist an $n$ sufficiently large such that, for every $z \in \{0,1,2,3,\dots \}$, if
	\begin{equation}
	\label{eq_diameter}
	\sqrt{ \frac{V}{f} } \left(  z + \frac{ 2 }{ \sqrt{4-3f} } \right) \leq 1 \ \ ,
	\end{equation}
	 then in period 1 the network has directed diameter at least $D_1=z+1$
\end{proposition}

The benchmarks given by condition \eqref{eq_diameter}, for $z \in \{1,2,3,4,5 \}$, are reported in Figure \ref{fig:diameter} together with numerical results obtained through a set of simulations. 
Although the simulations are performed with a relatively low number of agents ($ n= 80$), while the analytical curves are calculated using an approximation with a continuous uniform distribution, we observe that the two match extremely well.

\begin{figure}[ht]
	\centering
	\includegraphics[width=\textwidth]{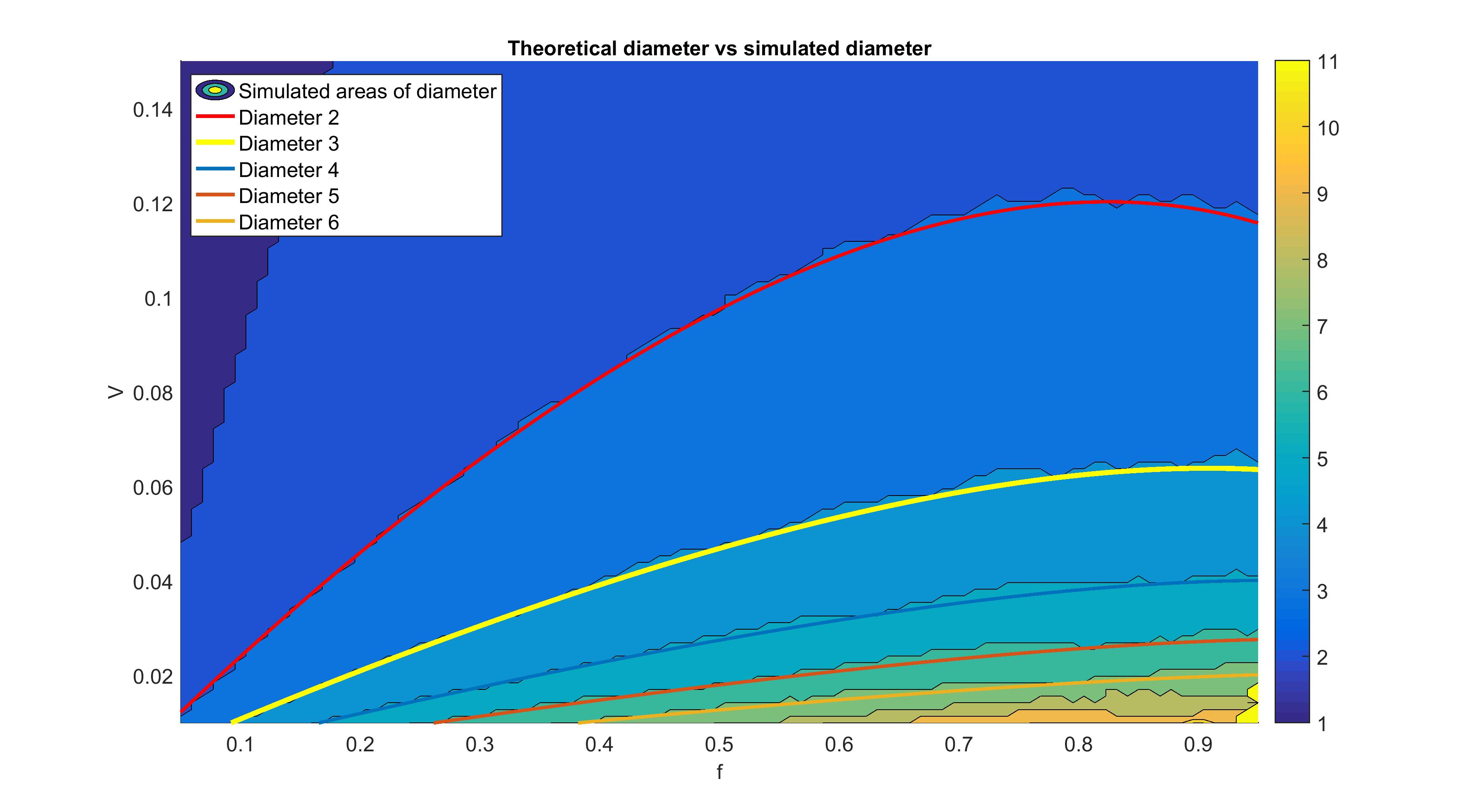}
	\caption{Full simulation of the model for   a $100 \times 100$ grid of parameters for  $V$ and $f$, for $n=80$ agents, and with 100 periods. On the horizontal axis we have the parameter $f$, while on the vertical axis we have, instead the parameter $V$. The different colored areas represent the diameter of the network in the first period as the result of a given simulation. The smooth lines  are the lines obtained analytically from the benchmark given by inequality \eqref{eq_diameter}.}\label{fig:diameter}
\end{figure}

\bigskip

Starting with a uniform distribution, using equations \eqref{eq:uniform_extreme} and \eqref{eq:uniform_central}, we can also estimate the size of the peer groups for nodes that are initially extremal (i.e.~with initial opinion close to $0$ or $1$) or central (i.e.~with initial opinion close to $1/2$).
We have that: 
\begin{enumerate}[i)]
\item for a central node the optimal most distant link is at distance 
\begin{equation} \label{eq:delta}
\delta=\frac{-f+\sqrt{f(f+9n^2V)}}{3fn} \ \ ; 
\end{equation}
\item for an extremal node the optimal most distant link is at distance 
\begin{equation} \label{theta}
\vartheta= \frac{6-6f+\frac{\sqrt{3f(4+3(-2+f)f)+36(4-3f)n^2V}}{\sqrt{f}}}{(12-9f)n} \ \ ; 
\end{equation}
\item at the limit $n\rightarrow \infty$ we have that $\delta\rightarrow \sqrt{\frac{V}{f}}$ and $\vartheta\rightarrow \frac{2}{\sqrt{4-3f}}\sqrt{\frac{V}{f}}$; 
\item at the limit $n\rightarrow \infty$, $4\delta<1$ is equivalent to $V<\frac{1}{16}f$;
\item at the limit $n\rightarrow \infty$, $2\vartheta>1$ is equivalent to $V>\frac{f}{16}(4-3f)$. 
\end{enumerate}

\added{The thresholds highlighted in points iv) and v) are of particular importance. When $4\delta<1$, it means that the central agent will form links such that the furthest agent is located in the right neighborhood of 1/4 (and the left neighborhood of 3/4). This means that agents that are in the interval $[0,\delta)$ and $(1-\delta,1]$, in the initial distribution $\vec{x}_0$, will not connect to those that are in $[2\delta,1-2\delta]$.}\footnote{\added {We use a similar reasoning in the proof of Proposition \ref{res:divergence}, where we prove that the opinions will not reach consensus if $4\delta<1$. All the details can be found in the Appendix.  In other words, an extreme agent will not connect to an agent that is in the middle of the distribution of opinions in the first step of the dynamic process. This implies that the diameter is at least 4 in all region of parameters such that $V<\frac{1}{16}f$.}}

\added{Similarly, the threshold identified in point v), identifies the region of parameters $(V,f)$ such that the diameter of the network in the first step is at most 3. Intuitively, if the extreme left agent has their furthest link located beyond 1/2, it means that they could reach the extreme right agent passing through at most 2 other nodes, i.e. the diameter is equal to 3. In the next two Propositions, we use these threshold values on parameters (that imply specific values of the diameter of the network in the first step) to predict the long run behavior of the distribution of opinions. In particular, we find that when the diameter in the first step is at most 3, the opinions always reach a consensus. Moreover, when the diameter in the first step is at least 5, the network is always disconnected (two or more components), and so opinions persistently differ in the long run. The formal results follow.}

\begin{proposition}[Long--run consensus] \label{res:convergence}
	 There is an $n$ large enough, such that if $V>\frac{f}{16}(4-3f)$, then the process converges to $x_i\rightarrow \frac{1}{2}$ for all $i$. 
\end{proposition}
 
This proposition states a sufficient condition on the values of $V$ and $f$ for the process to converge to consensus. It is intuitive that the value of $V$ should be high enough in order for the density of links in the first period to be high, too. 
When that is the case, the contraction of opinions is strong enough for the system to move strongly toward the complete network, which in turn leads to opinion consensus.

Probably less intuitive is the following result on a different case, where the network begins as a single connected component but then splits up over time. Technically, to address the proof, we consider the limiting case of large $n$ where the distribution of agents is continuous in the interval $[0,1]$.
It is easy to extend equations  \eqref{eq:degroot} and \eqref{eq:payoff_equilibrium}, that define the dynamics, to this limiting case.\footnote{%
We choose to move to a continuous description of the problem only to simplify the proof. We have not been able to write a rigorous and reasonably concise proof that takes into account all the caveats on discreteness for the discrete case with large enough $n$.
The simulations shown in Figure \ref{fig:region} show that this result holds also for \emph{large} $n$ in the discrete case.}

\begin{proposition}[Long--run disagreement]
\label{res:divergence}
\added{There exist an $n$ sufficiently large such that,}
if $V<\frac{1}{16} f$, then the network \added{becomes disconnected}, and there are at least two separate components in the limiting network. 
\end{proposition}

All the results in this section are performed approximating the formulas of the discrete uniform distribution with those of very large $n$. The results of a simulation performed on a 100x100 grid of parameters $V$ and $f$, for $n=80$ agents, reported in Figures \ref{fig:diameter} and \ref{fig:region}, show that our results holds also for limited $n$ and a relatively short number of iterations. 
	Figure \ref{fig:region} summarizes most of our findings and additionally shows that the theoretical curves of Propositions \ref{res:convergence} and \ref{res:divergence} are sufficiently close to the realized thresholds, when the population is limited.

\begin{figure}[ht!]
	\centering
	\includegraphics[width=\textwidth]{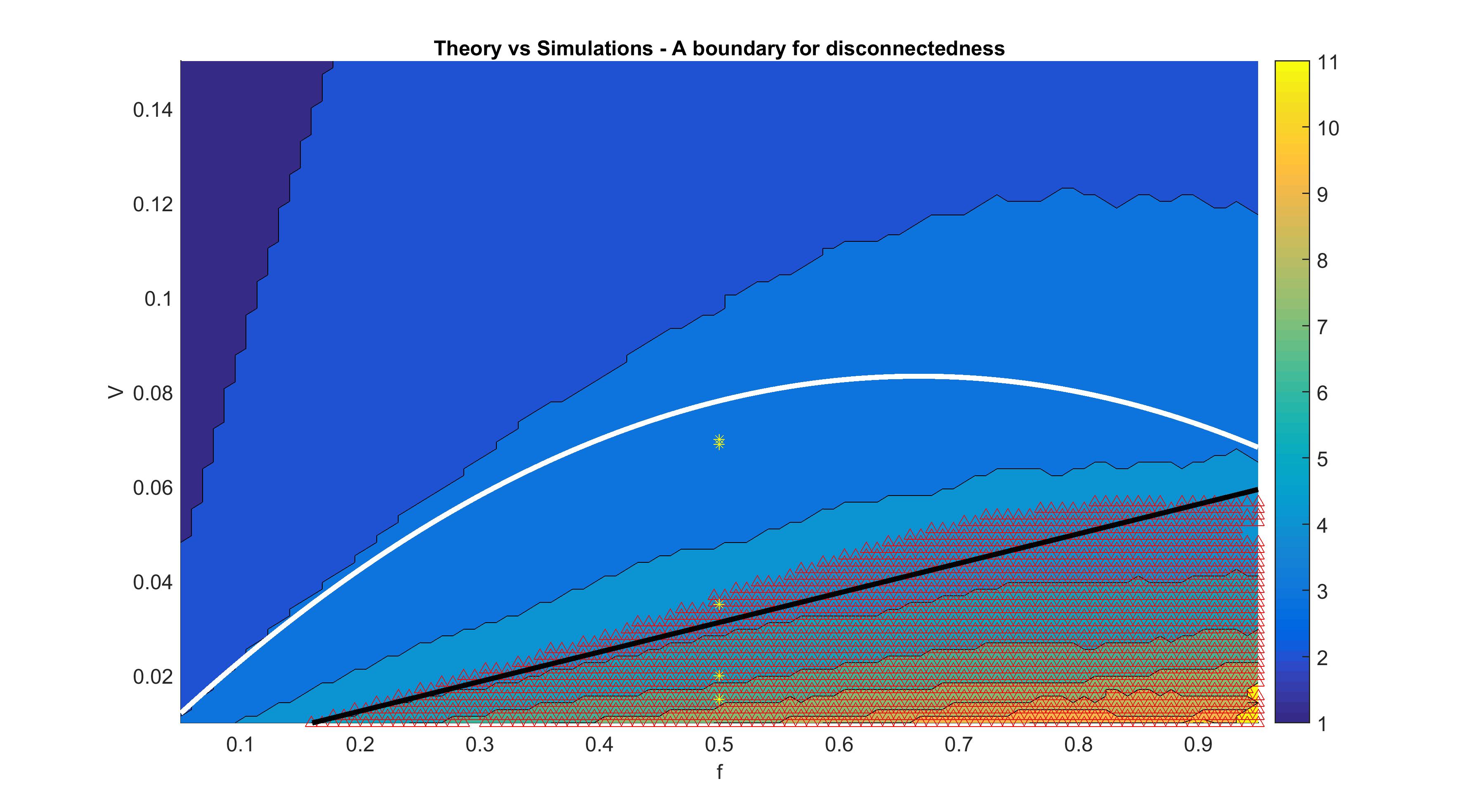}
	\caption{This figure summarizes most of the results of Section \ref{sec:uniform}. The grid corresponds to a full simulation of the model for a given combination of parameters, for $n=81$ agents, and with $T=20$ periods. On the horizontal axis we have the parameter $f$. On the vertical axis we have, instead the parameter $V$. 
	The white curve is given by equation $V=\frac{f}{16}(4-3f)$, which we refer to in Proposition \ref{res:convergence}.
	The black line is given by equation $V=\frac{1}{16} f$, which we refer to in Proposition \ref{res:divergence}.
	Overall, the figure summarizes the results of 100x100 simulations. We highlight the region of parameters for which the results of the simulation exhibit a diameter in period 1 equal to, respectively from the top left to bottom right, diameter equal to 1, 2, 3, 4, 5... which are represented by the heatmap in at the right of the figure. Finally, we have in the bottom right part of the figure an area highlighted by red triangles. These triangles show up in the figure whenever the long run outcome of the simulation exhibits multiple components.  It is clear that most interesting aspects of the model happen to be contained in a region of parameters such that the diameter of the network in the first period is around 4, consistently with the analytical results. \label{fig:region} }
\end{figure}

\bigskip

Notice that, combining Propositions \ref{res:eqdiameter}, \ref{res:convergence} and \ref{res:divergence},  a simple statistic of the network, the diameter, is a strong predictor of where the system is headed. 
In fact, the threshold condition from Proposition \ref{res:convergence} happens when the diameter is 3, while the threshold condition from Proposition \ref{res:divergence}  happens when the diameter is 4 (this is an analytical result from Proposition  \ref{res:eqdiameter} that is perfectly confirmed by the simulations).
Thus, if at step $1$ the initial diameter of the network is  2 or less, then opinions will always converge.
If instead at step $1$ the initial diameter of the network is 5 or more, then opinions will always diverge \added{and the diameter will never decrease along the dynamics}.

Looking back at the discussion in Sections \ref{sec:nonmono} and \ref{sec:speed}, self--correcting polarization seems to be the behavior of our dynamical process for the parameter values $V$ and $f$  at the boundary between convergence and polarization: that is, the upper bound of the red region, which lies a bit above the black benchmark line provided by Proposition \ref{res:divergence}.

\added{Moreover, simulations indicate that the actual threshold between reaching consensus or disagreement always lies in the region of parameters $(V,f)$ such that the diameter at step $1$ is 4. 
Thus, 
an initial diameter of 3 will always predict convergence to a unique opinion. Actually, since in a network converging to polarization the diameter is weakly increasing, while in a  network converging to unanimity it is weakly decreasing, simulations show that if we start from a uniform distribution a diameter of 3 observed at any point in the dynamics predicts consensus.
Also, consider non–-monotonic cases of temporary disagreement, like that in  Figure \ref{fig:non-monotone},  where opinions reach consensus after a phase of polarization.
From our analysis, using  Propositions \ref{res:convergence} and \ref{res:divergence}, we indirectly show that, by exclusion, this interesting region lies between the two thresholds (the transition region is the boundary of the red shaded area, lying between the white and black lines in Figure \ref{fig:region}).
Simulations show that such cases are only possible when the  diameter of the network is $4$.}

\bigskip

We believe that all these insights from observing the diameter of the network are relevant, especially when we turn to real-world cases. The so-called \emph{small world} phenomenon is evocative of the fact that average distance and diameter are small in most real-world social networks. 
\cite{backstrom2012four} show that the average distance and the average diameter have decreased over the years in the Facebook network, and the limit appears to be between 4 and 5.

It is worth making a small digression to consider the phenomenon of so-called \emph{echo chambers}, defined as the reinforcement of opinions in a closed system composed of like-minded individuals. There is extensive evidence that, through social network platforms, individuals today tend to interact in closed communities. The reasons are partly exogenous, such as the feed-refresh algorithms used in Facebook that expose individuals to suggestions heavily correlated with their activities on the platform, aiming to maximize user attention. However, individuals can also endogenously start deleting connections with different-minded individuals, in moments of pronounced divergence of opinions.\footnote{%
One example is found in the Washington Post in an article entitled\textit{\href{https://www.washingtonpost.com/news/the-intersect/wp/2015/12/16/you-might-think-trump-is-bad-but-unfriending-his-supporters-on-facebook-is-worse/?noredirect=on&utm_term=.b309899b61c3}{``Unfriending'' Trump supporters is just another example of how we isolate ourselves online}}. The article reveals that many websites provide support for massively deleting Facebook friends.} 
This kind of behavior is highlighted in the proof of Proposition \ref{res:divergence}. We also corroborate this discussion with visual support in Appendix \ref{app:networks}, where the evolution of the network structure is shown under different selections of parameters. 
We believe this is a useful insight into the natural formation of echo chambers that can help in the implementation of interventions aimed at reducing opinion polarization.

Furthermore, our study reveals a counter-intuitive observation: heightened sensitivity to social influence among agents (that is, high $f$, when the white line in Figure \ref{fig:region}, obtained in Proposition \ref{res:convergence}, starts decreasing) imposes stricter parameter requirements on the value of connections (that is, lower values of parameter $V$), leading to increased disagreement and pronounced polarization during the transition phase. This outcome, driven by agents' short-term focus on coordination, hinders consensus formation, despite the potential for higher payoffs. This insight underscores the challenges in achieving consensus in the presence of diverse ideological groups.

    \section{Conclusion}\label{sec:conclusion}

In this study, we extend the classical \cite{degroot1974reaching} model of na\"ive \added{opinion updating} by introducing endogenous network formation among agents. Our research yields valuable insights into the dynamics of opinion polarization within societies.
\added{The model is based only on two parameters characterizing the agents: a benefit from connections and a flexibility in adapting their opinions.}

Our analysis partitions the \added{dynamics} into three cases. \added{In one case, the network ends up being disconnected, with  persistent disagreement between the different separated components.}
As the benefit of interactions increases, we enter \added{the two cases} where a long-run consensus is achieved.
\added{In one of them, where the benefit of interactions is not too high, we observe temporary polarization and}
opinion dynamics exhibit non-monotonic behavior. \added{In the last one}, at the highest level of this benefit from interactions, consensus is reached, and notably, it occurs at a rate faster than what would be observed in an exogenously given network.
\added{For the case of uniform distribution of opinions to begin with}, we have found that the diameter of the network in the first updating period serves as an accurate predictor of the parameter region, explaining both transitional and long-run opinion behaviors within a society.

\added{When the initial starting condition of the opinions approximates a uniform distribution, we also characterize regions of the two parameters for these three regions,.
This characterization is partial, so we specify the transition regions with simulations.
Doing so, we show that the diameter of the network in the first step of the process and in future ones is a predictor of the type of the dynamics.
A diameter of $5$ or higher at any step of the simulation predicts long--run disagreement.
If it is $3$ or less in any step of the simulation, the diameter will decrease during the process and  we will surely end up to consensus.
This happens smoothly, if the diameter is $3$ in the very first step.
Finally, when it is $4$ in the first step of the simulation we are around the transition phase, and all the three cases are possible, including temporary disagreement.}


While our work significantly enhances the standard naive learning approach by considering endogenous network connections and referral dynamics, we acknowledge certain limitations. Our model focuses solely on interactions between equally important individuals, disregarding prominent agents, influencers, media, and politicians who may influence opinion distribution. Future research should explore the impact of such influential entities on polarization patterns, including reinforcement and cycles between consensus and polarization. Additionally, introducing network formation frictions, akin to social media algorithms, may offer insights into consensus and polarization dynamics in the digital age.

In summary, our work not only extends the \cite{degroot1974reaching} model in a parsimonious manner but also lays the groundwork for understanding the underlying mechanisms of opinion dynamics. Future extensions that consider prominent agents and network frictions promise to provide a more comprehensive view of polarization in contemporary societies.

\clearpage

	\appendix
	\global\long\def\thesection{Appendix \Alph{section}}
	\global\long\def\thesubsection{\Alph{section}.\arabic{subsection}}
	\setcounter{proposition}{0} \global\long\def\theproposition{\Alph{proposition}}
	\setcounter{lemma}{0} \global\long\def\thelemma{\Alph{lemma}}
	\setcounter{definition}{0} \global\long\def\thedefinition{\Alph{definition}}
	\setcounter{figure}{0} \global\long\def\thefigure{\Alph{figure}}
	\setcounter{example}{0} \global\long\def\theexample{\Alph{example}}
 \setcounter{proposition}{0} \global\long\def\theproposition{\Alph{proposition}}

\section{Proofs} \label{app:proofs}
	
We start with a simple Lemma in linear algebra.

\begin{lemma} \label{lemma_unique}
Consider a directed network between $n$ nodes, where $N_i$ is the set of node $i$'s neighbors and $d_i$ is the cardinality of $d_i$.
Consider a diagonal matrix 	$F = \left(
	\begin{array}{ccc}
		f_1 &  & 0 \\
		& \ddots & \\
		0 & & f_n
	\end{array}
	\right) \in [0,1]^{n \times n}$, a vector $\vec{\theta}  \in \mathbb{R}^n$, and an  \emph{adjusted adjacency matrix} $D$ such that 
	$D_{ij} = \left\{
	\begin{array}{ccc}
		\frac{1}{d_i} & \mbox{ if $j \in N_i$,} \\
		0 & \mbox{ otherwise}.
	\end{array} \right.$\\
 
	The following equation:
\begin{eqnarray}
	\label{eq:bestresponse_matrix}
	(I - FD) \vec{x} = (I-F) \vec{\theta} \ \ 
\end{eqnarray}
 has a unique solution $\vec{x} \in [0,1]^n$. \eproof
\end{lemma}

	\subsection*{Proof of Lemma \ref{lemma_unique} }
	
	There is clearly a unique solution to the unconstrained equation (\ref{eq:bestresponse_matrix}), because $(I-F)$ and $(I-FD_t)$ are always full rank matrices. \\
	Now suppose that the maximum element of $\vec{x}_t$, call it $x_{m,t}$ is such that $x_{m,t}>1$.
	Then, by (\ref{eq:bestresponse}), it is a convex combination of $\frac{\sum_{j \in N_{m,t}} x_{j,t}}{d_{m,t}}$ and $x_{m,t-1}$.
	But since $x_{m,t-1}\leq1$, it follows  that $\frac{\sum_{j \in N_{m,t}} x_{j,t}}{d_{m,t}} > x_{m,t}$, which contradicts the initial assumption. \\
	In the same way it is impossible for the minimum element of $\vec{x}_t$ to be less than $0$. \eproof
	
	\subsection*{Proof of Lemma \ref{lemma_mean_variance} (page \pageref{lemma_mean_variance})}

	From (\ref{eq:payoff}) and (\ref{eq:bestresponse}), the payoff in equilibrium is 
	\begin{eqnarray}
		\pi_{i,t} & = &  d_{i,t} \left( V - f^2 (1-f)  (\mu_{i,t} - x_{i,t-1})^2  \right) -  f \sum_{j \in N_{i,t}} (x_{i ,t}- x_{j,t})^2 \nonumber \\
		& = &  d_{i,t} \left( V - f (1-f)  (\mu_{i,t} - x_{i,t-1})^2  \right)   - f \sum_{j \in N_{i,t}} \left[ (x_{i ,t}- x_{j,t})^2 + (1-f)^2  (\mu_{i,t} - x_{i,t-1})^2 \right] \nonumber \\
		& = & d_{i,t} \left( V - f (1-f)  (\mu_{i,t} - x_{i,t-1})^2  \right)   - f \sum_{j \in N_{i,t}} (x_{j,t} - \mu_{i,t})^2 \nonumber \\
		& = & d_{i,t} \left( V - f (1-f) (\mu_{i,t} - x_{i,t-1})^2  - f \sigma^2_{i,t} \right)   \ \ ,
	\end{eqnarray}
	where we have called $\sigma^2_{i,t} = \frac{\sum_{j \in N_{i,t}} (x_{j,t} - \mu_{i,t})^2}{d_{i,t}}$ the variance of the actions of $i$'s neighbors. \eproof

	\subsection*{Proof of Lemma \ref{res:ordered} (page \pageref{res:ordered})}
 
	\begin{proof}
  \added{We prove our result by induction, starting with the consideration that $\vec{x}_{0}$ respects, by definition of players' labels, the order labels, and that, as no network is defined at time $0$, before the dynamics starts, the configuration respects trivially also the conditions on the network.}
  
  \added{Suppose that $\vec{x}_{t-1}$ respect the order of labels.}
Take a player $i$ at time $t$, and let her choose optimally her set of peers $G_{i,t}$, maximizing the payoff in \eqref{eq:payoff}, which we know that in equilibrium is given by \eqref{eq:payoff_equilibrium}.
		It is easy to see that the optimal set of peers, given \eqref{eq:payoff_equilibrium}, must be such that their actions $x_{j,t-1}$'s are ordered: meaning that for any $k \in G_{i,t}$, $\min \{ x_{j,t-1} : j \in G_{i,t} \} \leq x_{k,t-1} \leq \max \{ x_{j,t-1} : j \in G_{i,t} \} $ (the condition holds trivially if $G_{i,t}$ is empty). \\
		This result characterize the best response with respect to the actions of the other players.
		
		\bigskip
		
		Now, we need to show that the \added{evolution of opinions $\vec{x}_t$} is well ordered with respect to the initial opinions  $\vec{x}_{t-1} $.
		Suppose it is not, then there are two nodes $i$ and $j$, with $x_{i,t-1} < x_{jt-1}$, such that at least one of the following two conditions hold:
		\begin{enumerate}
			\item $\min \{ x_{k,t} : k \in G_{i,t} \} > \min \{ x_{k,t} : k \in G_{j,t} \} $ ;
			\item $\max \{ x_{k,t} : k \in G_{i,t} \} > \max \{ x_{k,t} : k \in G_{j,t} \} $ .
		\end{enumerate}
		However, both $i$ and $j$ are maximizing  \eqref{eq:payoff_equilibrium}, with the same parameters $V$ and $f$, so that if at least one of the two conditions above holds, then $G_{j,t}$ (assuming it is optimal for $j$) would dominate $G_{i,t} $ for node $i$, and $G_{i,t}$ (assuming it is optimal for $i$) would dominate $G_{j,t} $ for node $j$. This results in a contradiction.
		
		\bigskip

  \added{This shows that the network at time $t$ respects the bullet points of Definition \ref{def:ordered}.}
		
		\bigskip
		
		It follows directly that the best response of nodes $i$ and $j$, given  their ex--ante opinions and given the fact that their reference sets are shifted accordingly, are such that $x_{i,t} < x_{j,t}$.
		
		\bigskip
		
		To finish the proof, consider the situation in which there is a node $\ell$ with empty reference set. This node will have $x_{\ell,t}= x_{\ell,t-1}$. 
		With an argument similar to the one used above, since $\ell$ is also maximizing \eqref{eq:payoff_equilibrium}, no node at the left (respectively, right) of $\ell$ will include any node in her reference set  at the right (respectively, left) of $\ell$.
		So, if a node $h$ is at the left (respectively, right) of $\ell$ it will have $x_{h,t} < x_{\ell,t}$ (respectively, $x_{h,t} > x_{\ell,t}$). \\
		This concludes, \added{by induction,} the proof.
	\end{proof}

	\subsection*{Proof of Proposition \ref{res:empty_myopic} (page \pageref{res:empty_myopic})}
	
	\begin{proof}
  \added{We first show that if at some point $t$ we have two  components of the network (i.e.~a subset of players connected between each others and not connected to others), players in those two components will not connect together at any time $t'>t$.
  From Lemma \ref{res:ordered}, at any time step an equilibrium is ordered, according to Definition 
		\ref{def:ordered}.
  So, if we have two components $C_t$ and $C'_t$, and for some $i \in C_t$ and $j \in C'_t$, we have $x_{i,t}<x_{j,t}$, then for all $k \in C_t$ and $h \in C'_t$ we will have $x_{k,t}<x_{h,t}$.}
  
  \added{Call $C_t$ the set of players in a component  at time $t$. Call 
  $k = \min_{C_t} $, and $h = \max_{C_t} $.
  Because of the labelling order of players, and because of Lemma \ref{res:ordered}, $x_{k,t}$ is the minimal opinion in the component, and $x_{h,t}$ is the maximal.
  By definition of components, player $k$ has chosen only peers with an ex--ante  opinion higher than her own, so she will update to some $x_{k,t+1} \geq x_{k,t}$.
  With the same reasoning, $x_{h,t+1} \leq x_{h,t}$.
  This proves that, once two components are disconnected, they will never connect together again, because the extrema (i.e.~the players in that component with extreme opinions) will converge, so that if they did not find it profitable to connect earlier, they will neither in the future.
  If the extrema of components does not find it profitable to connect, then neither interior players will.}

\added{The above result avoid cycles in the dynamics, because a component will either converge asymptotically to a unique opinion (because of averaging of opinions), and along this process all its nodes will connect reciprocally to each others (because the closer they get the more profitable it is for them to connect), or it will disconnect, and the process of splitting is irreversible.
  }

  \added{To conclude the proof, if the first network is empty, because no player found it profitable to connect to any other node, then $\vec{x}_1 = \vec{x}_0 $.
  However, from now on, every singleton is a component, and will never connect to others, so that $\textbf{x}_{\infty}=\textbf{x}_{t}= \textbf{x}_0$ for any $t$.}
	\end{proof}

\subsection*{Proof of Proposition \ref{cor:suff_cond} (page \pageref{cor:suff_cond})}

	\begin{proof}
To prove the result, we first characterize the threshold of the proposition. The threshold is obtained by studying analytically the case of $n=2$, and then we show that the result applies for $n>2$. 

We assume that $n=2$. We can calculate the payoff for players $(i,j)$ as follows:
\begin{equation}\label{eq:threshold}
    \Tilde{\pi}_{i,t}=\Tilde{\pi}_{j,t}=V-f(1-f)\left(x_{j,t-1}-x_{i,t-1}\right)^2
\end{equation}
From equation \eqref{eq:threshold} we can verify that a link between $i,j$ is not profitable as long as $\left( \frac{V}{f(1-f)}\right)^{1/2}<|x_{j,t-1}-x_{i,t-1}|$. Notice that this condition implies that if the agent is indifferent, she will form the link. We need to prove that this lower bound on $V$ still applies for $n>2$. To do so, we first consider $n=3$. 

We call the $n=3$ agents $i,j,\ell$, respectively. Without loss of generality we set $x_{\ell,t-1}<x_{i,t-1}<x_{j,t-1}$ and assume that $V=f(1-f)\left(x_{j,t-1}-x_{i,t-1}\right)^2$. If the network is characterized by only one link between $i$ and $\ell$. Since this link exists by construction, it implies that $|x_{j-t-1}-x_{i,t-1}|>|x_{\ell-t-1}-x_{i,t-1}|$. Given these conditions, agent $i$ considers whether to form a link with $j$ comparing the payoff of maintaining the connection with $\ell$ only and the payoff from having both links with $\ell, j$. Calling $g$ the network where agent $i$ is linked with $\ell$, the payoff are
\begin{equation}
    \Tilde{\pi}_{i,t,g}=f(1-f)\left(x_{j,t-1}-x_{i,t-1}\right)^2-f(1-f)(x_{\ell, t-1}-x_{i,t-1})^2
\end{equation}
Hence, we call $g'$ the network where $i$ is linked both with $\ell$ and $j$. The payoff in this case are
\begin{equation}
    \begin{aligned}
        \Tilde{\pi}'_{i,t,g'}&=2f(1-f)\left(x_{j,t-1}-x_{i,t-1}\right)^2-2f(1-f) \left(\frac{x_{j,t-1}+x_{\ell,t-1}}{2}-x_{i,t-1}\right)^2+\\
        &-f\left(\left(x_{j,t-1}- \frac{x_{j,t-1}+x_{\ell,t-1}}{2}\right)^2+\left(x_{\ell,t-1}- \frac{x_{j,t-1}+x_{\ell,t-1}}{2}\right)^2\right)
    \end{aligned}
\end{equation}
If $\Tilde{\pi}'_{i,t,g'}>\Tilde{\pi}_{i,t,g}$ we would obtain a contradiction to our statement. With some simple algebra we can reduce the problem to the following inequality
\begin{equation}
\begin{aligned}
        &(1-f)\left(x_{j,t-1}-x_{i,t-1}\right)^2-2(1-f) \left(\frac{x_{j,t-1}+x_{\ell,t-1}}{2}-x_{i,t-1}\right)^2+\\
        &-\left(\left(x_{j,t-1}- \frac{x_{j,t-1}+x_{\ell,t-1}}{2}\right)^2+\left(x_{\ell,t-1}- \frac{x_{j,t-1}+x_{\ell,t-1}}{2}\right)^2\right)>-(1-f)(x_{\ell, t-1}-x_{i,t-1})^2
\end{aligned}
\end{equation}
From this inequality, we consider two extreme cases by setting $f=1$ and $f=0$, respectively. If $f=1$ it is trivial that the inequality cannot be satisfied. We then proceed to study the case of $f=0$. The latest equation reduces to the following
\begin{equation}
\begin{aligned}
    & \left(x_{j,t-1}-x_{i,t-1}\right)^2+(x_{\ell, t-1}-x_{i,t-1})^2-2 \left(\frac{x_{j,t-1}+x_{\ell,t-1}}{2}-x_{i,t-1}\right)^2+\\
   & -\left(\left(x_{j,t-1}- \frac{x_{j,t-1}+x_{\ell,t-1}}{2}\right)^2+\left(x_{\ell,t-1}- \frac{x_{j,t-1}+x_{\ell,t-1}}{2}\right)^2\right)>0
\end{aligned}
\end{equation}
Solving all squares and re-arranging, we see that the left-hand side is always equal to $0$. This contradicts the initial statement, so we conclude that if $V=f(1-f)\left(x_{j,t-1}-x_{i,t-1}\right)^2$, a link between $i$ and $j$ is not profitable, even if player $i$ has a connection with another player. Clearly, the argument would not change if player $i$ would have multiple connections other than $j$. Moreover, by symmetry, the argument holds for player $j$ as well, so if $\left( \frac{V}{f(1-f)}\right)^{1/2}<|x_{j,t-1}-x_{i,t-1}|$, the adjacent agents $i$ and $j$ will not link together. 

	\end{proof}

\subsection*{Proof of Proposition \ref{prop:time} (page \pageref{prop:time})}
\begin{proof}
	If  there are no nodes with out-degree $0$, then by definition, $W_t$ is row--stochastic, so it has a maximum eigenvalue equal to $1$.
	Let us call $\lambda_2 ( W_t )$ the second largest eigenvalue of the hearing matrix. \\
	Given a network $G$, with nodes $N$, we can call $A\subseteq N$ any subset of her nodes, and $\partial A$ the links between the nodes in $A$  and those outside $A$.
	The  \emph{Cheeger constant} \citep{chung1996laplacians} for a network  is defined as:\footnote{%
		The Cheeger constant is \emph{large}  if any possible division of the vertex set into two subsets has \emph{many} links between those two subsets.
		Intuitively, a large Cheegerr constant is present when there are \emph{bottlenecks} in the graph that slow down the diffusion of a process in the network.}
	\[
	h(G) = \min \left\{ \frac{ | \partial A |}{|A|} : A \subseteq  N,  \    0 < |A| \leq |N|/2  \right\}
	\]  
	The \emph{Laplacian matrix} $L(G)$ of a network $G$ of $n$ nodes  is defined as an $n \times n$ matrix where
	\[
	L_{ij} = \left\{
	\begin{array}{ccl}
		1 & \mbox{if} & i=j \\
		- \frac{1}{\sqrt{k_i k_j}} & \mbox{if} & \mbox{$i$ is linked to $j$} \\
		0 & & otherwise
	\end{array}
	\right.
	\]
	where $k_i$ and $k_j$ are the degree of $i$ and $j$. 
	If the network is regular we have that $L(G) = I - D$, and so $W= I - fL$. \\
	\emph{Cheeger inequality} \citep{chung1996laplacians} states that:\footnote{%
		\cite{chung2005laplacians} proves that this inequality holds also if $G$ is not symmetric.
	}
	$\lambda_2 (L (G) ) \leq h(G)$. From this, we have that $\lambda_2 (W) = 1 - f \lambda_2 (L (G) ) \geq 1 - f h(G)$. \\
	The initial configuration $D_0$ of our dynamics is a regular ring lattice where the Cheeger constant is obtained when we take all the nodes in half of the spectrum (see e.g., \citealt{gu2010consensus}).
	We know  from \eqref{eq:delta} that the size of the links from this subset of nodes to all the other nodes is $\delta= \sqrt{ \frac{V}{f} }$, so that $h(W_0)= 2\delta = 2\sqrt{ \frac{V}{f} }$.
	From this we obtain that
	\[
	\lambda_2 (W_0) = 1 - f \lambda_2 (L (G) ) \geq 1 - 2  \sqrt{ f V }
	\]
	From Proposition \ref{eq_diameter} we know that if the network has diameter at least $2$ (that is, it is not the complete network), then $1 - 2  \sqrt{ f V }$ is strictly positive (implying also that $2 \sqrt{V} > \sqrt{f}$, which we use below).
	If the stochastic process converges, \emph{consensus time } is defined as: 
	\[
	CT(\epsilon, W) = \sup_{ \textbf{x}  \in [0,1]^n} \min \left\{  T: ||  W^{T} \textbf{x} - W^{\infty} \textbf{x} ||  < \epsilon  \right\} \ \ ,
	\]
	where $|| \textbf{x} - \textbf{y} ||$ is a distance defined by any specified norm. 
	Consensus time measures how many steps we would need to be $\epsilon$ close to consensus, if the network were fixed to remain $D_t$ for all the future steps, and if we started from the worst possible distribution of opinions $\textbf{x}$. \\
	In their Lemma 2, \cite{golub2012homophily}  prove that, if the stochastic process defined by $W$  converges, then, for any positive $\epsilon$ we have
	\begin{equation}
		\left\lfloor \frac{\log(2 \epsilon)  - K(W) }{\log \left( | \lambda_2 ( W  ) | \right) } \right\rfloor
		\leq 
		CT (\epsilon, W)
		\leq
		\left\lceil    \frac{ \log(\epsilon) }{ \log \left(  | \lambda_2 ( W  ) | \right) }  \right\rceil \ \ ,
		\label{GJbounds}
	\end{equation}
	where $K(W)$ is a constant that depends on the degree distribution of the underlying network.
	Now, we can focus on our two dynamics.
	The exogenous time will have a consensus time, for a certain precision $\epsilon$, that is longer than 
	\begin{equation}
		\left\lfloor \frac{\log(2 \epsilon)  - K (W  )  }{\log \left( | \lambda_2 ( W  ) | \right) } \right\rfloor \geq 
		\left\lfloor \frac{\log(2 \epsilon)  - K (W  )  }{\log \left( 1 - 2  \sqrt{ f V } \right) } \right\rfloor 
		\label{cond1}
	\end{equation}
	We know that the endogenous process will converge, at some time $t'$ to the complete network, for which the second eigenvalue of the Laplacian matrix is $1-1/n$ \citep{powers1988graph,powers1989bounds}, and hence the second eigenvalue for our hearing process is $1-f+\frac{f}{n}$.
	This means that the consensus time will be shorter than 
	\begin{equation}
		\left\lceil    \frac{ \log(\epsilon) }{ \log \left(  1-f +  \frac{f}{n} \right) }  \right\rceil  \ \ .
		\label{cond2}
	\end{equation}
	As $2 \sqrt{V} > \sqrt{f}$, we can always choose an $\epsilon$ small enough and an $n$ large enough, so that the lower bound in \eqref{cond1} for the exogenous process is greater than the upper bound in \eqref{cond2} for the complete network, and that this difference is greater than $t'$.
	Hence, for that $\epsilon$ the resulting upper bound computed in \eqref{cond2}, and then summed to $t'$, provides the desired $T$. 
\end{proof}

	\subsection*{Proof of Proposition \ref{res:eqdiameter} (page \pageref{res:eqdiameter})}
	\begin{proof}
		Imagine odd nodes at \added{the initial conditions $\vec{x}_0$} evenly spaced. Then for the central node we have that the optimal $k$ is given by $\frac{-f+\sqrt{f(f+9n^2V}}{3f}$. For the extreme left node, instead, the optimal $k$ is given by $\frac{6-6f+\frac{\sqrt{3f(4+3(-2+f)f)+36(4-3f)n^2V}}{\sqrt{f}}}{(12-9f)n}$. Therefore the network has directed diameter $D=z+1$ if 
		\begin{equation}
			\lim\limits_{n\rightarrow \infty} z\left(\frac{-f+\sqrt{f(f+9n^2V}}{3fn}\right)+\frac{6-6f+\frac{\sqrt{3f(4+3(-2+f)f)+36(4-3f)n^2V}}{\sqrt{f}}}{(12-9f)n}=1
		\end{equation}
		The limit is equal to $\frac{-2\sqrt{f}\sqrt{(4-3f)V}-4\sqrt{fV}z+3f\sqrt{fV}Z}{f(-4+3f)}$, which is equivalent to \eqref{eq_diameter}.
	\end{proof}

	\subsection*{Proof of Proposition \ref{res:convergence} (page \pageref{res:convergence})}
	\begin{proof}
		
		To prove this result we look at the extreme player, called for the purposes of this proof, agent 0. We show that this agent will keep adding links until he or she has connected with the whole network. By symmetry this is true for the other agent that sits in the opposite extreme of the distribution. Therefore, these agents will converge towards the center. Thus because of the Proposition on ordered equilibrium, we conclude that all other agents will converge there, too.
		
		We first consider this with a value of $V=\frac{f}{16}(4-3f)$. With such a value we know that agent $0$ will send links until to the agent that is located exactly at $\frac{1}{2}$. With $n$ sufficiently large, $\mu_{0,t_1}=\frac{1}{4}$. Thus $x_{0,t_1}=f\frac{1}{4}$. Because of symmetry we observe that in $t_2$, $x_i\in[f\frac{1}{4},1-f\frac{1}{4}]$. Moreover, we claim that the connections for agent 0 will increase in period 2. To prove this assume that agent 0 maintains the same connections in period 2. It is easy to observe that $\sigma_{0,t_2}^2<\sigma_{0,t_1}^2$. Thus, if $\left(\mu_{0,t_2}-f\frac{1}{4}\right)^2\leq\frac{1}{16}$, then we see that the  payoff for agent 0 has gone up. To show this we rewrite the inequality above as follows, denoting with $k_0$ the cardinality of $\{j\in N_0\}$.
		
		\begin{equation}
			\left(\frac{1}{k_0}\sum_{j\in N_0}x_{j,t_2}-x_{0,t2}\right)\leq\left(\frac{1}{k_0}\sum_{j\in N_0}x_{j,t_1}-x_{0,t1}\right)
		\end{equation} 
		
		which can be rewritten as 
		
		\begin{equation}
			\frac{1}{k_0}\sum_{j\in N_0}x_{j,t_2}-\frac{1}{k_0}\sum_{j\in N_0}x_{j,t_1}+k_0\left(\frac{1}{k_0}x_{0,t_1}-\frac{1}{k_0}x_{0,t_2}\right)<0
		\end{equation}
		Then the observation that $x_{0,t_2}-x_{0,t_1}>x_{j,t_2}-x_{j,t_1}$ for all $j\in N_0$ completes the argument. Note that the inequality $\left(\mu_{0,t_2}-f\frac{1}{4}\right)^2<\frac{1}{16}$ is strict.  
		
		Thus we have shown that if agent 0 keeps exactly the same connections in the second period as in the first one, her payoff will unambiguously go up. Agent 0 can therefore increase her utility by adding further links as far as $\sigma_{0,t_2}^2=\sigma_{0,t_1}^2$ and $\left(\mu_{0,t_2}-f\frac{1}{4}\right)^2=\frac{1}{16}$. If she does so, agent 0 will bear the same costs as in period 1, but will have added links and therefore will receive a greater benefit. 
		
		The argument proposed above can be iterated, up to the point where agent 0 will be connected with the whole population. By symmetry the opposite agent on the spectrum of opinions will do the same. Thereafter, the process can only converge to the average opinion in the population, which is $x_i\rightarrow \frac{1}{2}$ for all $i$.
	\end{proof}

	\subsection*{Proof of Proposition \ref{res:divergence} (page \pageref{res:divergence})}
	\begin{proof}
		Let us consider variable $\delta$ (from equation \eqref{eq:delta}), which defines the radius of an interior agent at time $1$, i.e.~an agent who chooses an interior reference group symmetrically distributed around her.
		It still holds in the continuum approximation, as in the  limit for $n\rightarrow \infty$, that $V<\frac{1}{16}$ (our hypothesis) is equivalent to $4\delta<1$.
		Let us consider what happens to our dynamics in the first $3$ steps, which will provide us with the starting hypothesis for an induction argument.
		
		\bigskip
		
		{\bf Step $0$}: Since $\delta<1/4$, only the agents in the intervals $[0,\delta)$ and $(1-\delta,1]$ will form an asymmetric reference group, all those in the interval $[\delta, 1-\delta]$ will form a symmetric reference group. \\
		Another implication is that, {\bf (a) those agents that are in  $[0,\delta)$ and $(1-\delta,1]$ at time $1$ will not connect to those that are in $[2\delta, 1-2 \delta]$ at time $1$.} \\
		It results from these choices of reference groups that those  agents in the intervals $[0,\delta)$ and $(1-\delta,1]$ will move their $x$'s  towards the center, while all the others will maintain their $x$'s.
		
		\bigskip
		
		{\bf Step1}: now, the agents in the intervals $[\delta,2\delta)$ and $(1- 2 \delta,1-\delta]$ (they have not changed $x$ since step $0$ -- note also that since $4 \delta<1$ these intervals are separated) have more concentration of peers at the extrema than towards the center.\\
		As a consequence,  {\bf (b) those agents that are in  $[\delta,2\delta)$ and $(1- 2 \delta,1-\delta]$  at time $1$ will decrease, at least weakly, their connection towards the center of the distribution, and increase their connections towards the extrema of the distribution.} \\
		It results from these choices of reference groups that those  agents that were originally in the intervals $[0,\delta)$ and $(1-\delta,1]$ will keep moving their $x$'s  towards the center, 
		those agents that have been in  $[\delta,2\delta)$ and $(1- 2 \delta,1-\delta]$ since step $0$ will move towards the extrema,
		while all the others will maintain their $x$'s.
		
		\bigskip
		
		{\bf Step2}: now, at least some agents in the interval $[2 \delta, 1-2 \delta]$ (only those in the extrema of this interval, if $6 \delta < 1$, all of them otherwise) have less concentration  at the extrema than they have towards the center.\\
		As a consequence,  {\bf (c) those agents that are in  $[2 \delta, 1/2]$  at time $1$ will weakly decrease their connection towards those agents that are in $[0 , 2 \delta)$ at time $1$; those agents that are in  $[1/2, 1-2 \delta]$  at time $1$ will weakly decrease their connection towards those agents that are in $(1-2 \delta,1]$ at time $1$.} \\
		It results from these choices of reference groups that those  agents that have been in the intervals   $[\delta,2\delta)$ and $(1- 2 \delta,1-\delta]$ since step $0$ will move towards the extrema, while all those in the interval $[2 \delta, 1-2 \delta]$ will either move to the center or maintain their $x$'s.
		
		\bigskip
		
		From the three considerations in bold above we can construct a proof based on numerical induction. The statement that we want to prove is that those  agents that were originally in the intervals   $[\delta,2\delta)$ and $(1- 2 \delta,1-\delta]$, at time $1$, will (weakly) move away from those agents that were originally in the interval $[2 \delta, 1-2 \delta]$ at time $1$. This is true in the first 3 steps, where conditions {\bf (a)}, {\bf (b)} and {\bf (c)} above holds.
		
		\bigskip
		
		{\bf Induction step:} if conditions {\bf (a)}, {\bf (b)} and {\bf (c)} on the way agents choose their reference groups hold, then those agents that were originally in  $[\delta,2\delta)$ and $(1- 2 \delta,1-\delta]$ will move to the extrema with respect to benchmarks $2 \delta$ and $1-2 \delta$, respectively; 
		on the other hand, those that were originally in the interval  $[2 \delta, 1-2 \delta]$ will not move outside this interval, and those at the extrema of this interval will move towards the center.
		As a consequence, statements {\bf (a)}, {\bf (b)} and {\bf (c)} will hold at any step of the process.
		
		\bigskip
		
		Finally, to conclude the proof, we need to show that the network ends up being disconnected.
		In fact, if this were not the case, all agents would end up having the same $x$.
		This goes against what we have proved, namely that those agents that were originally in the intervals   $[\delta,2\delta)$ and $(1- 2 \delta,1-\delta]$, at time $1$, will (weakly) move away from those agents that were originally in the interval $[2 \delta, 1-2 \delta]$ at time $1$.
	\end{proof}

\section{Analysis of a fully rational model} \label{app:rational}
	
	In this appendix we consider a model with fully rational agents and a single step.
	Initial distribution of opinions is given, players choose whom they want to link to, and then how to update their opinions, as in Section \ref{section:one period}.
 \added{The difference here is that players anticipate, using backward induction, that also other players will update their opinions.}

\added{We provide this extension for two main reasons. 
First, we can prove results for the fully rational model, analogous to those in Proposition \ref{res:empty_myopic}, with the same logic of those that we use for our main model.
Second, we show that the case with fully rational players has multiple equilibria, that would make the description of a model with many steps undefined. 
}

\added{The issue of multiplicity is well-known in the literature on endogenous networks and is particularly relevant in the empirical estimation of models with endogenous networks. The empirical literature has so far treated this issue by imposing restrictions on the rationality of agents to select equilibria and reduce the multiplicity. In some instances, it has been assumed that agents meet according to a dynamic meeting protocol; alternatively, other papers have assumed a setting of incomplete information or limited the influence that passes through the network to a certain distance.}\footnote{\added{See \cite{comola2022estimating} for a recent example of application of incomplete information to solve the issue of multiplicity in estimating a model with endogenous directed network and externalities. We also remand to \cite{comola2022estimating} for their extensive literature review on the issue of multiplicity in the empirical literature on networks.}} 
\added{The model we study in the main body of the paper falls in the same category because we can solve a multiplicity of equilibria by restricting individuals' rationality, even though with a different flavor to what has been done in the empirical literature.}

To avoid subscripts, in this appendix we call the initial profile $\vec{x}_0$ of opinions with $\vec{\theta}$ and the update profile of opinions $\vec{x}_1$ with $\vec{x}$.
	The timing of the one period game (that can be compared with the multi--step process of Definition \ref{def:timing2}) is as follows.
	\begin{definition}[Timing with rational players]\label{def:timing1} 
	In the two stage game with rational players:
	\begin{itemize}
			\item \textbf{Time} $0$: Opinions $\vec{\theta}$ are exogenously assigned to agents;
			\item \textbf{Time} $1$: agents simultaneously and independently form the directed network, which becomes common knowledge;
			\item \textbf{Time} $2$: $\vec{x}$ is formed, according to best replies with respect to opinions and the network.
		\end{itemize}
	\end{definition}
	
	Formally, the strategy space $S_i$ of player $i$, is $S_i= \{0,1\}^{n-1}\times \left( x_i:G \rightarrow \mathbb{R}^+ \right)$, where $x_i$ is a function. The agents choose first whether to form links with all the other agents, we call $G_i \in \{0,1\}^{n-1} $ this choice for player $i$, and then, after the observed realization $G$ of the network, they update their opinion $x_i$. 
	So, $G$ is an adjacency matrix obtained from the rows $\{ G_1, \dots , G_n \}$.  Note that $G_i$ can be seen also as a  subset of the other players, and this is how we will interpret it when referring to it as a set. From the timing and the definition of strategies, we introduce the equilibrium concept. Formally we see this as a sequential game that can therefore be solved by backward induction. We now focus on the solution of the system of best responses, when the network is given, and therefore $\mu_i$ is uniquely defined for all $i$. Then in the next section, we move back to the previous stage letting agents form the network, and finally move back to the very initial stage on the distribution of opinions to characterize the possible equilibria.

	Endogenizing the network, the first natural step is to consider the equilibrium concept of sub-game perfect Nash equilibrium. Since the network is directed, the solution is fully non-cooperative, worked out in two stages via backward induction. In fact, the agents will first choose their neighbors, and then will update their behaviors according to equation \eqref{eq:bestresponse}.  \added{Lemma \ref{lemma_unique} establishes} uniqueness of equilibrium for the sub--game given fixing a network and we can now focus on the network formation stage.

	Multiplicity of equilibria is still an issue. Intuitively there can be an agent that is a priori indifferent whether forming the link to her right or to her left. However, after choosing one of the two, it is no longer beneficial to form the link in the other direction, because the opinion would update according to the link that has been formed. Therefore the choice will determine the shape of the equilibrium. The following example shows a simple case where this multiplicity arises  and is non--generic. 
	
	\begin{example}[Three nodes]
		\label{ex:threenodes}
		Consider three agents with $\theta_1=0$, $\theta_2=\frac{1}{2}$ and $\theta_3=1$. Assume $f=\frac{1}{3}$ and $1/32 < V < 3/64$. \\
		Let us show that one equilibrium has nodes $1$ and $2$ connected together and node $3$ disconnected.
  \added{In this case,} $x_1=\frac{1}{8}$, $x_2=\frac{3}{8}$, and since $V>\frac{1}{32} =  \frac{f(1-f)}{(1+f)^2} (\theta_2-\theta_1)^2$, they keep their connection. \\
		If node 2 connects to node 3, then $x_3=1$ and, from equation (\ref{eq:bestresponse}),
		\[
		\left\{
		\begin{array}{ccl}
			x_1 & = & \frac{1}{3} x_2 \\
			x_2 & = & \frac{1}{3} + \frac{1}{3} (1+x_1)
		\end{array}
		\right.
		\Rightarrow
		\left\{
		\begin{array}{ccl}
			x_1 & = & \frac{1}{4}  \\
			x_2 & = & \frac{3}{4} 
		\end{array}
		\right. \ \ .
		\]
		The new payoff for node 2 is,  from equation (\ref{eq:payoff_equilibrium}),
		\[
		\pi'_2 = 2 \left(  V- \frac{2}{9} \frac{1}{8^2} - \frac{1}{3} \frac{3^2}{8^2} \right) < 2 \left(V - \frac{3}{64} \right) < 0 \ \  .
		\]
		So, node $2$ will not connect to node $3$. \\
		Similarly, node $3$, connecting to node $2$, would play $x_3 = 2/3+1/8 = 19/24$, and get a payoff of 
		\[
		\pi'_3= V- \frac{2}{9} \left(  \frac{19}{24} - \frac{3}{8}  \right)^2 = V - \frac{25}{72} < V - \frac{3}{64} < 0 \ \ .
		\]
		Finally, it is easy to see that by symmetry there is also an equilibrium where nodes $2$ and $3$ connect together, and node $1$ remains isolated. \eproof
		
	\end{example}
	
	Therefore we determined that if an equilibrium exists, meaning that the set of equilibria is non-empty, then this set is entirely composed by ordered equilibria. However, so far this characterization has been insufficient to answer the main question of the paper: ``when does higher polarization occur?''. To get there, we need the network to be disconnected. We dig, then, into the set of equilibria described above to search for the subset of equilibria such that the network is indeed disconnected. We provide a condition on the initial distribution of types, guaranteeing that we will observe the phenomenon we are looking for. 
	
	Intuitively the condition states that two agents should be ``distant'' enough in the first period. Then in the future period, the distance in their opinion cannot decrease. In fact, if they remain isolated then the distance between opinions remains constant, and if a link is formed by any of these two agents, the distance actually increases. This characterizes a local condition on the initial distribution of opinions.
	
	\begin{proposition}\label{res:suff_cond}
		Assume that ex--ante opinion $\{ \theta_1 , \dots \theta_n \}$ are  in increasing order.
		There is a threshold $\phi$, depending on $f$ and $V$, such that if there exist two agents $i$ and $i+1$ such that  $|\theta_i-\theta_{i+1}|>\phi $, then all equilibria are disconnected.\\
		This threshold is $\phi=(1+f)\left(\frac{V}{f(1-f)}\right)^{\frac{1}{2}}$, and is convex at $f$, with a minimum at $\frac{1}{3}$.\\
	\end{proposition}	
		
	\begin{proof}
	This proof follows the steps of the proof of Proposition \ref{cor:suff_cond}. We report the passages because the value of the threshold differs from the myopic model studied in the main body of the paper, as stated in Proposition \ref{res:suff_cond}. 
 
 First, let us assume that $n=2$. Then we can calculate the payoff for players $\{i,j\}\in N$ as follows:
		\begin{equation}
			\pi_i=\pi_j=V-f(1-f)\left(\frac{\theta_j-\theta_i}{1+f}\right)^2
		\end{equation}
  {\color{red}  }
		From this we obtain that a link between $i$ and $j$ is never profitable as long as $\phi=\left(\frac{V(1+f)^2}{f(1-f)}\right)^{\frac{1}{2}}<|\theta_j-\theta_i|$, which is the condition identified in the Proposition.	We now have to prove that the link between $i,j$ keeps remaining not profitable regardless of any other existing link that $i$ and $j$ might have. To do so we now assume that $n>2$. For simplicity we further assume that $i$ has only one existing link with player $\ell$. If so, we know that such link must be profitable and thus from equation \ref{eq:payoff_equilibrium} we obtain
		\begin{equation}
			\pi_i=V-f(1-f)\left(x_{\ell}-\theta_i\right)^2
		\end{equation}
		Setting $V=f(1-f)\left(\frac{\theta_j-\theta_i}{1+f}\right)^2$ then we get 
		\begin{equation}
			\pi_i=f(1-f)\left(\frac{\theta_j-\theta_i}{1+f}\right)^2-f(1-f)\left(x_{\ell}-\theta_i\right)^2>0
		\end{equation}
		Now we compare this value to the payoff that $i$ would obtain from being simultaneously linked with $j$ and $\ell$. We call this payoff $\pi_i'$ and it is equal to the following
		\begin{equation}
			\begin{aligned}
				\pi_i'&= 2(V-f(1-f)(\mu_i-\theta_i)^2-f\sigma^2_i)\\
				&=2(V-f(1-f)\left(\frac{x_j+x_{\ell}}{2}-\theta_i\right)^2-f\frac{1}{2}\left((x_j-\frac{x_j+x_{\ell}}{2})^2+(x_{\ell}-\frac{x_j+x_{\ell}}{2})^2\right))
			\end{aligned}
		\end{equation}
		We than set again $V=f(1-f)\left(\frac{\theta_j-\theta_i}{1+f}\right)^2$ and build the inequality $\pi_i'>\pi_i$, supposing that our statement is not true. We then obtain, after some simple algebra
		\begin{equation}
			\begin{aligned}
				(1-f)\left(\frac{\theta_j-\theta_i}{1+f}\right)^2-2(1-f)&\left(\frac{x_j+x_{\ell}}{2}-\theta_i\right)^2-\left((x_j-\frac{x_j+x_{\ell}}{2})^2+(x_{\ell}-\frac{x_j+x_{\ell}}{2})^2\right)>\\
				&>-(1-f)(x_{\ell}-\theta_i)^2
			\end{aligned}
		\end{equation}
		From this we take two extreme cases with $f=0$ and $f=1$. The latter case is evident how the inequality is never satisfied, so it remains to show that the contradiction stand when $f=0$. The above inequality becomes
		\begin{equation}
			(\theta_i-\theta_j)^2+(x_{\ell}-\theta_i)^2- \left( x_j-\frac{x_j+x_{\ell}}{2} \right)^2-\left( x_{\ell}-\frac{x_j+x_{\ell}}{2} \right)^2-2\left(\frac{x_j+x_{\ell}}{2}-\theta_i\right)^2>0
		\end{equation}
		Solving all the squares and rearranging we obtain
		\begin{equation}
			\theta_i \left( \frac{x_j+x_{\ell}}{2}-\theta_j-x_{\ell} \right)-\frac{x_j+x_{\ell}}{2} \left( 2\frac{x_j+x_{\ell}}{2}-x_j-x_{\ell} \right)>0
		\end{equation}
		from which we obtain a contradiction since the first term is always negative, and the second is null. Thus, if $V\leq f(1-f)\left(\frac{\theta_j-\theta_i}{1+f}\right)^2$, $i$ does not find it profitable to link with $j$. The same inequalities apply even if more links are at $i$'s right. 
		\end{proof}

Comparing the threshold obtained in Proposition \ref{res:suff_cond}, we notice that it is more stringent than the one in Proposition \ref{cor:suff_cond}. The threshold is in fact increased by the factor $(1+f)$. The reason for this difference lies in the fact that players take into account that connections affect others' opinions. In particular, their opinion would become ``closer'', and thus the distance required for the sufficient condition on disconnected networks is greater. 
This difference is higher the higher the flexibility $f$ of agents. Furthermore, we can formalize the following Corollary, which identifies a lower bound this point simply states that there will be a number of components $C$ at least equal to the number of pairs of agents that do not satisfy the sufficient condition for the network to be connected, meaning that it will have a single component. It can be seen easily that if there is a single couple  $ij$, then $C=2$.
		This is a lower bound, since we cannot control the components that arise because of the rule of order chosen to select one of the equilibria.
		\begin{corollary}
		    Equilibria exhibit a number of components $C$ that obeys the inequality:
		\[ C\geq | \{i \in \{1, \dots, n-1 \} : \phi<|\theta_i-\theta_{i+1}|\} |+1 \ \ . \]
		\end{corollary}

	From the above result, there exists a threshold on distance of types such that for all rules of order in the network formation stage, the network will be disconnected. This is therefore a sufficient condition. 
	
	The last line in the statement says that the number of gaps in the initial distribution of opinions imposes a minimum threshold on the number of components in equilibrium, because nodes across those gaps would never have an incentive to connect together.

	Importantly, when assuming that the initial distribution of opinions is uniformly distributed, meaning that all types are equispaced on the spectrum of opinions, we obtain a generic case of the previous result, because all adjacent players have opinions the same distance apart before any link is formed.
	Thus, apart from intermediate cases, as in Example  \ref{ex:threenodes}, where multiple equilibria arise, the network in the first period is either empty or connected. To see that, we simply need to check among all couple of adjacent agents, and determine the couple characterized by the maximum distance. Then we can check that there is always a $V,f$ pair such that this distribution can show disconnectedness.
	
	Finally we note that the threshold is convex in $f$, with a minimum for $f=1/3$. The intuition behind this result is that agents have to give enough weight to the variance term. However, there is a non monotonic effect for larger $f$ since agents will also be more influenced by neighbors, leading to lower overall variance. This also implies that in the first place, agents prefer homogeneous neighborhoods, which is why in equilibrium they are likely to occur locally, leading to segregation. 
	

\section{Simulations}\label{app:simulations}

In this appendix we use simulations to analyze aspects of our dynamical model.

 \subsection{Relations with Hegselmann and Krause (2002)} \label{app:HK}
 
  What is  the relation between our model and the mechanical model from \cite{hegselmann2002opinion}? We provide some illustrative examples.
\begin{example}
 To show the main differences between the model in our paper and the model by \cite{hegselmann2002opinion} (from now on, HK model), there two ways to proceed: i) we target the parameters to match the network in the first step, and ii) we target the long-term dynamics. Hence, we create a benchmark to compare the two models. As we show below, the HK model requires much less connectivity in the first step of the dynamics for convergence to occur in the long run, with respect to our model. The reason for this is not due to sensitivity to parameters, but to the fundamental contribution of strategic choice of links. The following exercise has a number of conditions: the parameter $f$ is set to 0.5 in both models; there are $n=101$ agents; we vary the confidence level denoted by $\varepsilon_i$ in the HK model, and we vary the parameter $V$ in our model. 	
	In the first step of our comparison of the two models, we set $\varepsilon_i=0.2$ and $V=0.02$. The latter is the parameter value used in Figure 1b in the paper. 
	\begin{figure}
		\centering
		\begin{subfigure}{0.48\textwidth}
\includegraphics[width=\textwidth]{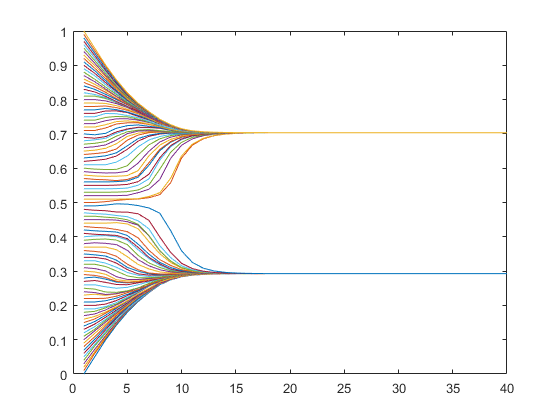}
\caption{}
		\end{subfigure}
		\begin{subfigure}{0.48\textwidth}
	\includegraphics[width=\textwidth]{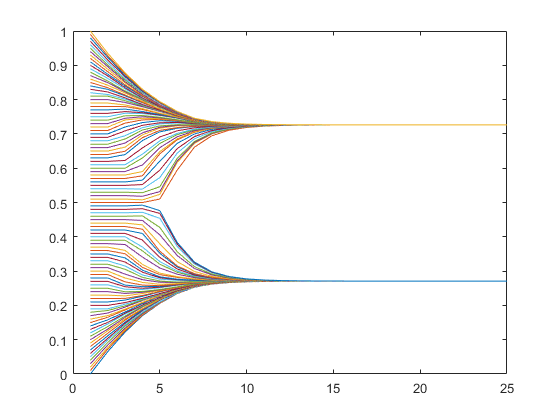}
	\caption{}
\end{subfigure}\\
\begin{subfigure}{0.48\textwidth}
\includegraphics[width=\textwidth]{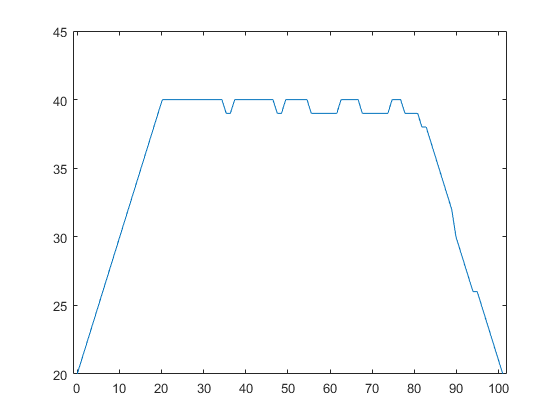}
\caption{}
\end{subfigure}
		\begin{subfigure}{0.48\textwidth}
	\includegraphics[width=\textwidth]{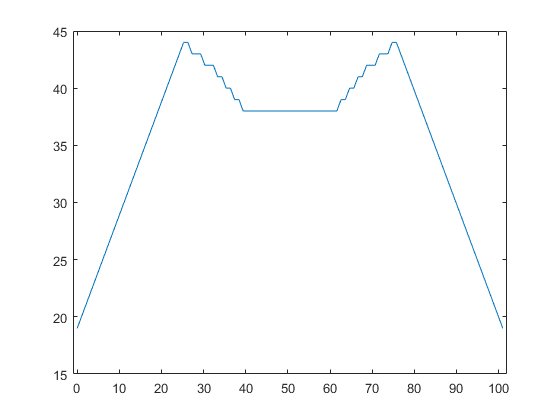}
	\caption{}
\end{subfigure}
\caption{Panel (a) shows the evolution of the HK model using $\varepsilon_i=0.02$, for all $i\in N$. The figure below, panel (c) reports the degree distribution of the network in the first step that results from the same process. Comparing it to panels (b) and (d) on the right, which correspond, respectively, to the full dynamics and the degree distribution in the first step for the model in our paper, we observe that the two models reproduce similar evolution of opinions and similar network in the first step. }
\label{fig:disconnect}
	\end{figure}

	\begin{figure}
	\centering
	\begin{subfigure}{0.48\textwidth}
		\includegraphics[width=\textwidth]{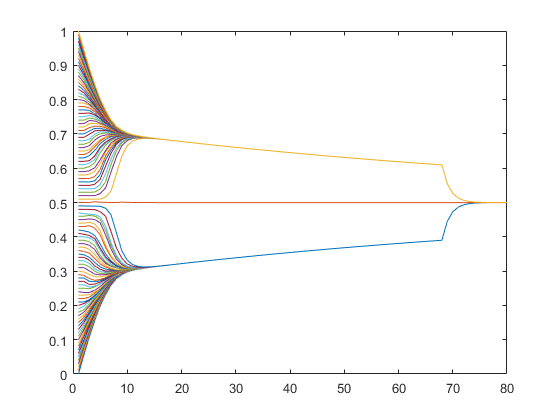}
		\caption{}
	\end{subfigure}
	\begin{subfigure}{0.48\textwidth}
		\includegraphics[width=\textwidth]{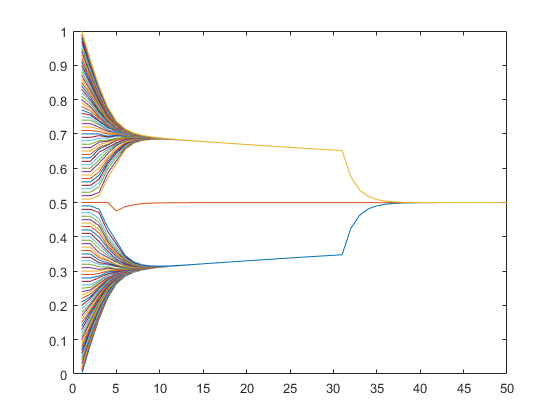}
		\caption{}
	\end{subfigure}\\
	\begin{subfigure}{0.48\textwidth}
		\includegraphics[width=\textwidth]{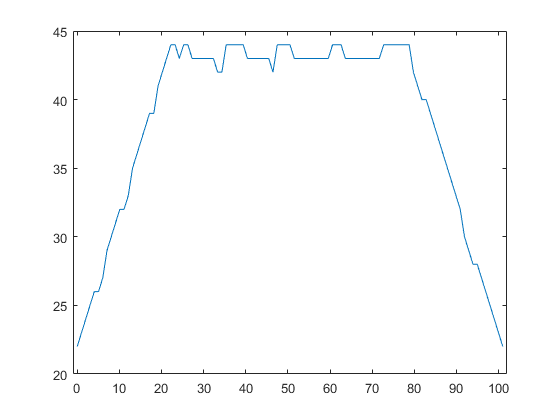}
		\caption{}
	\end{subfigure}
	\begin{subfigure}{0.48\textwidth}
		\includegraphics[width=\textwidth]{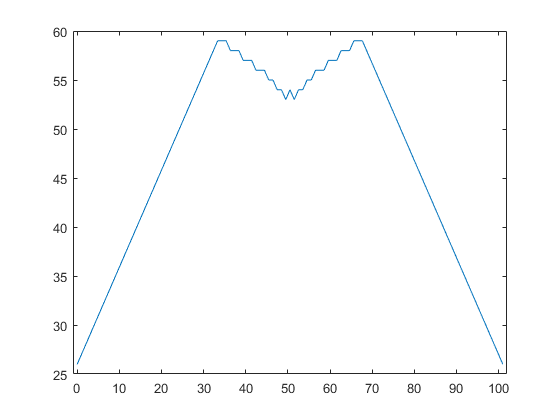}
		\caption{}
	\end{subfigure}
	\caption{Panel (a) shows the evolution of the HK model using $\varepsilon_i=0.022$, for all $i\in N$. The figure below, panel (c) reports the degree distribution of the network in the first step that results from the same process. Comparing it to panels (b) and (d) on the right, which correspond, respectively, to the full dynamics and the degree distribution in the first step for the model in our paper, we observe that the two models reproduce similar evolution of opinions, but the network in the first step is consistently more connected in our model. Since the parameter value is at the threshold for having connected network we acknowledge that our model requires more connectedness for the network to remain connected.  }
	\label{fig:connect}
\end{figure}

Looking at the top panels of figures \ref{fig:disconnect} and \ref{fig:connect}, we observe that the dynamics are very similar for the selected parameters used in the simulation. In particular, looking at figure \ref{fig:disconnect} we remark only a slightly higher polarization in the case of strategic network formation with respect to the HK model. The degree distribution is also similar: at the edges of the distribution of types, the degree is 20 in the HK model, 19 in the strategic model. The maximum degree is 40 in the HK model, 44 in the strategic model. Looking at this figure alone, we conclude that the strategic model is similar to the HK model, in the sense that both model produce similar dynamics, long run outcome and the first step degree distribution is somewhat similar. 

In Figure \ref{fig:connect} we use $\varepsilon_i=0.22$ and $V=0.03501$ for the strategic model.\footnote{The value of $V$ is the one corresponding to the simulation reported in the paper in figure \ref{fig:non-monotone}, panel \textit{b}.} We chose this values because they are close to the threshold level of the parameters before which the network is disconnected (in the long run), and after which it is connected (in the long run). \eproof
\end{example}

Firstly, we remark that the speed of convergence is much larger in the strategic model than in HK. More importantly, looking at the degree distributions we observe that the strategic model requires much more density in the first step of the dynamics. This happens for a specific reason: when agents are strategic there is a trade-off between the distance from the mean and the variance of opinions in an agent's reference group (see Equation \ref{eq:payoff_equilibrium} in the paper, referred to as \textit{payoff in equilibrium}). Deleting connections and extremization of opinions is a documented behavior in real life, and our model captures that. 

To conclude, while HK is able to produce non-monotonic behavior, this is due purely to the mechanical condensation of opinions around the extremes of the distribution of opinions. In our model, in addition to that, there is a preference-driven force that pushes agents to delete links making it more likely to disconnect the network.

\subsection{The role of the initial distribution of opinions}
	\label{distribution}
	
	Proposition \ref{res:suff_cond}  states, even for the fully rational model of network formation, that there exists a threshold for the distance between two adjacent agents' opinions above which they would not form a link in the initial period. This identifies a local condition in a general initial discrete distribution of opinions, that allows us to determine the situations when the opinions will not converge in the long-run. So, in general, the shape of the initial distribution of opinions matters. It is possible, though, to carefully design a distribution that would induce either long-run convergence or divergence under a wide range of parameters. For instance, a bimodal distribution has a natural inertia towards polarization in the long-run. In fact, when the distribution is already polarized \added{in the initial distribution}, then the conditions to reach convergence are more strict. Therefore it seems natural to think that with a Gaussian distribution we should expect the opposite.
	However, this is not necessarily true, because even under a Gaussian distribution we may have conditions under which the tails will disconnect at a certain point, forming a group with a relatively extreme opinion, and disagreeing with the rest of the population. In the example below we report simulations for respectively a normal and a bimodal distribution, varying the values of $V$. The purpose is to show that all the cases we have discussed when using a uniform distribution will reproduce changing the initial distribution. 
	\begin{example}
		In this example we explore few initial distributions with the same parameters to shed light on the impact of all the initial conditions of the model. Here we fix the number of agents to $n=101$, the value for $f=0.5$, and we change the initial distribution of opinions showing two values of $V$ for each example, $0.01$ and $0.02$ respectively. We show results for a Normal distribution, Bimodal and Uniform. 
		\bigskip
		
		\noindent	{\bf Normal distribution}: Because of the discreteness of our set of agents, we build this initial distribution by combining several uniform distributions, with higher variance towards the tails and lowest variance in the middle segment. The overall variance of the vector of opinions chosen here is $\sigma^{2}_{t_1, Normal}=0.0546$. For reference, the uniform distribution with $n=101$ agents has a variance of $\sigma^2_{t_1,Uniform}=0.0858$. Figure \ref{fig:gaussian} shows the result for two different values of $V$, keeping $f=0.5$
		(Figure \ref{fig:evolution_normal} in \ref{app:networks} shows some figures of the network during the evolution of the case in the left panel).
		
		\begin{figure}[ht]
			\centering
			\begin{subfigure}[b]{0.48\textwidth}
				\includegraphics[width=\textwidth]{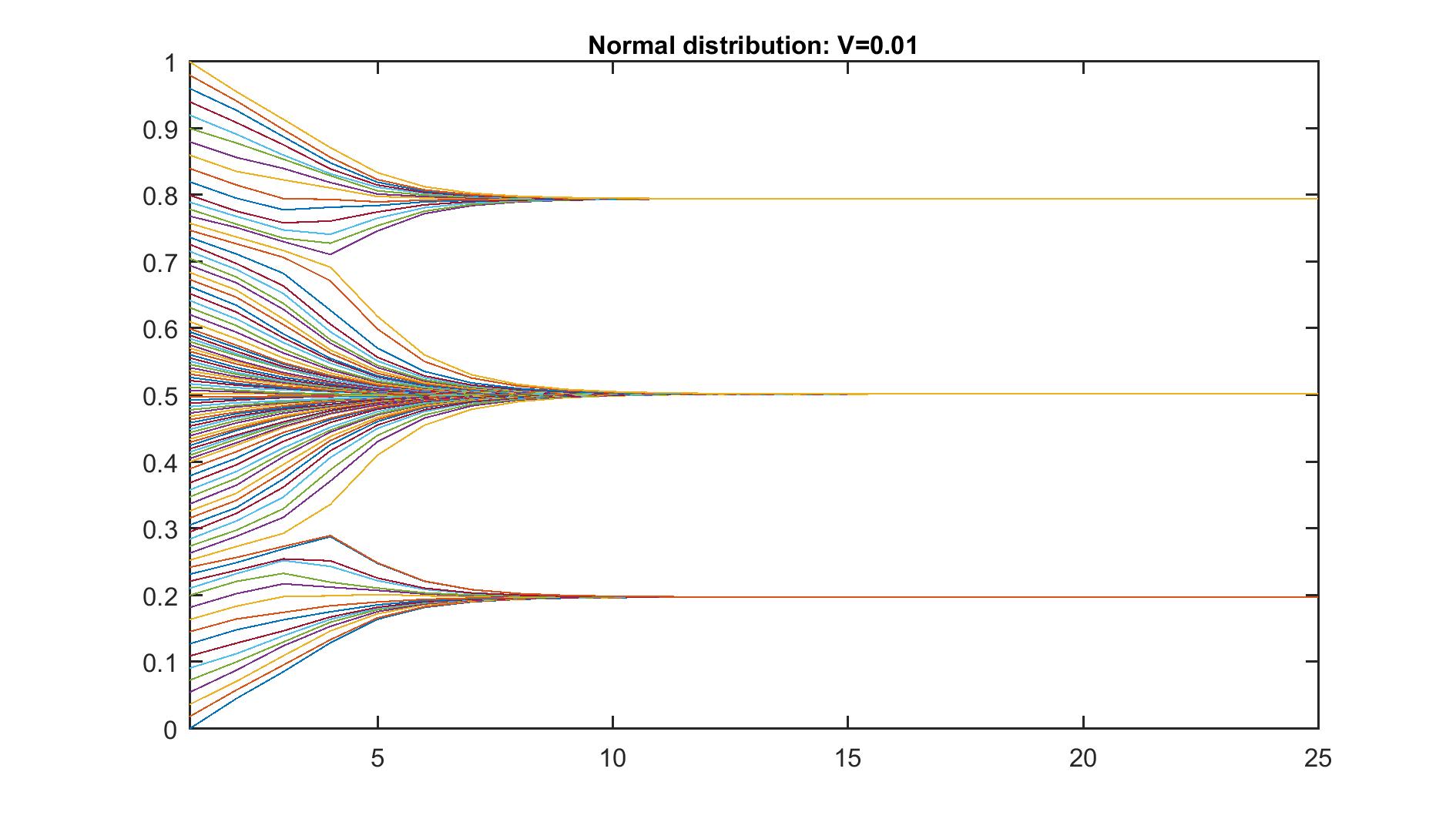}
				\caption{Low value of V ($V=0.01$). The tails separate from the large middle group, and will converge to an extreme opinion in the long-run. }
			\end{subfigure}
			~
			\begin{subfigure}[b]{0.48\textwidth}
				\centering
				\includegraphics[width=\textwidth]{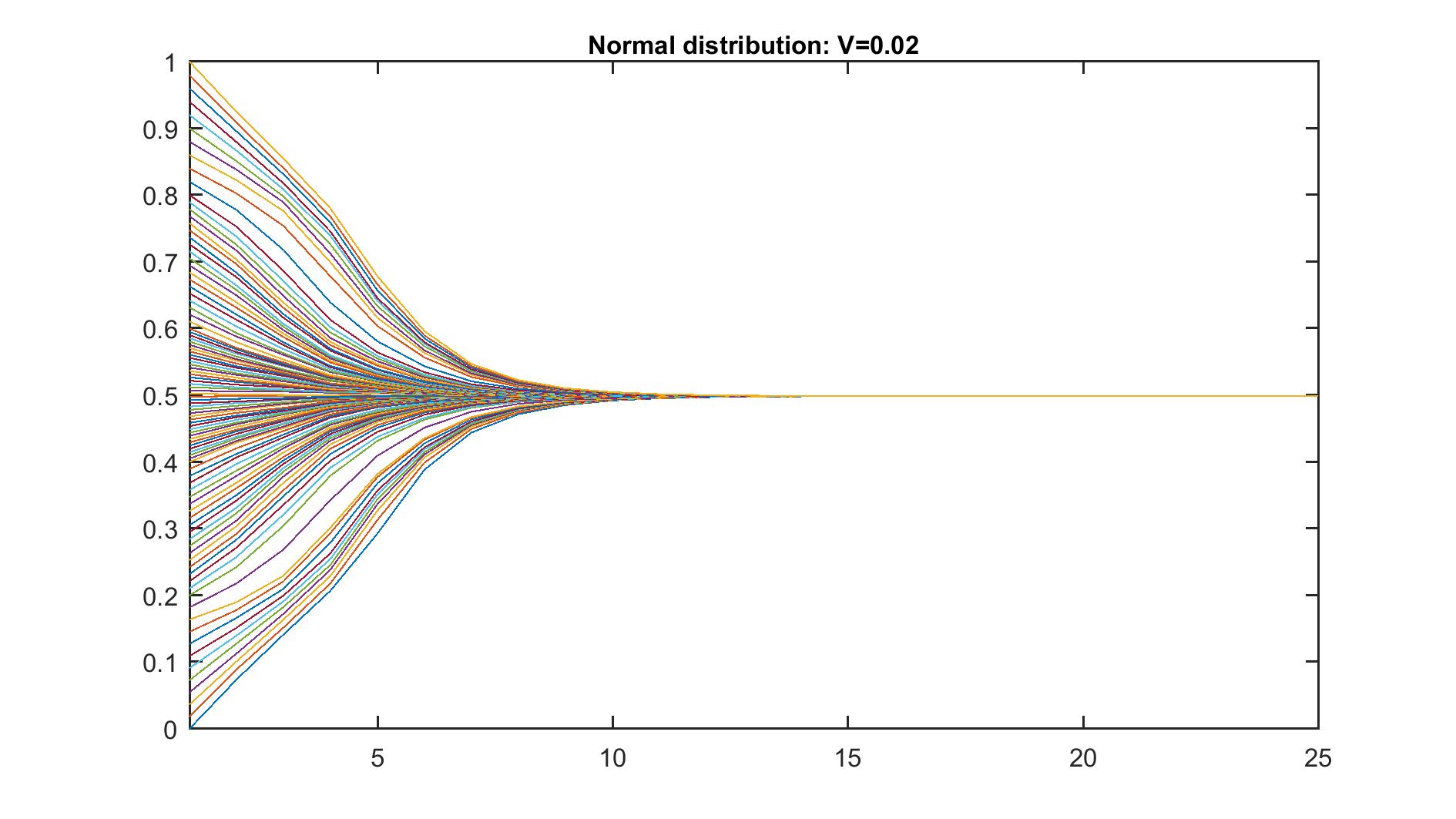}
				\caption{High value of V ($V=0.02$). The tails trade-off their initial position to access a larger number of links, favored by a high value of V. }
			\end{subfigure}
			\caption{Analysis of the dynamics starting with a Gaussian distribution.}
			\label{fig:gaussian}
		\end{figure}
		\bigskip
		\noindent {\bf Bimodal distribution}: The bimodal distribution requires more strict conditions in order to reach consensus in the long run. Intuitively, an already polarized distribution tends towards long-run polarization, and thus the region of parameters that induces consensus must be smaller with respect to other distributions. 
		Figure \ref{fig:bimodal} shows the result for two different values of $V$, keeping $f=0.5$.
		\qed
		
		\begin{figure}[ht]
			\centering
			\begin{subfigure}[t]{0.48\textwidth}
				\includegraphics[width=\textwidth]{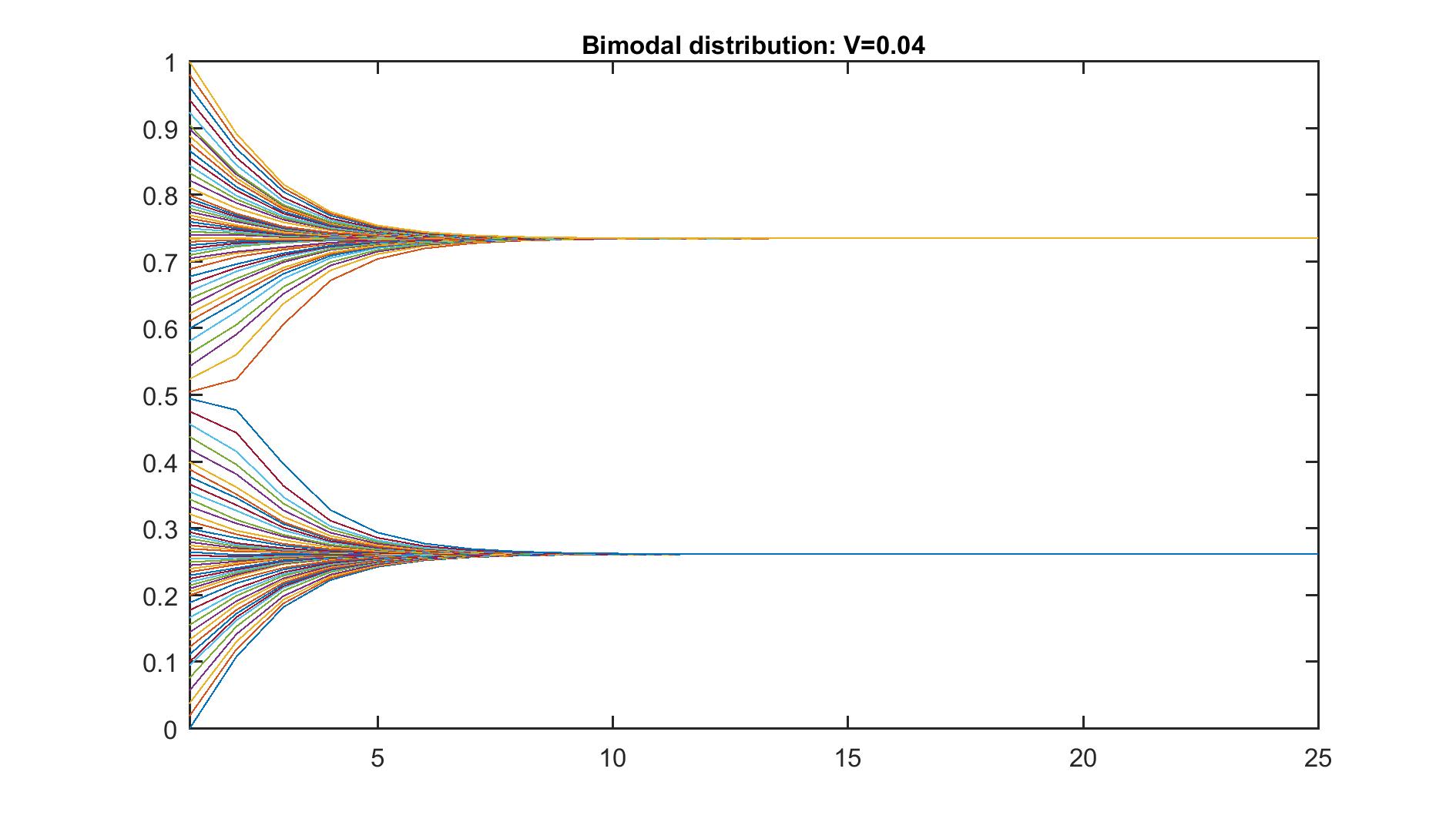}
				\caption{ $V=0.04$. Long-run disagreement arises, despite the relatively high value of $V$. }
			\end{subfigure}
			\begin{subfigure}[t]{0.48\textwidth}
				\centering
				\includegraphics[width=\textwidth]{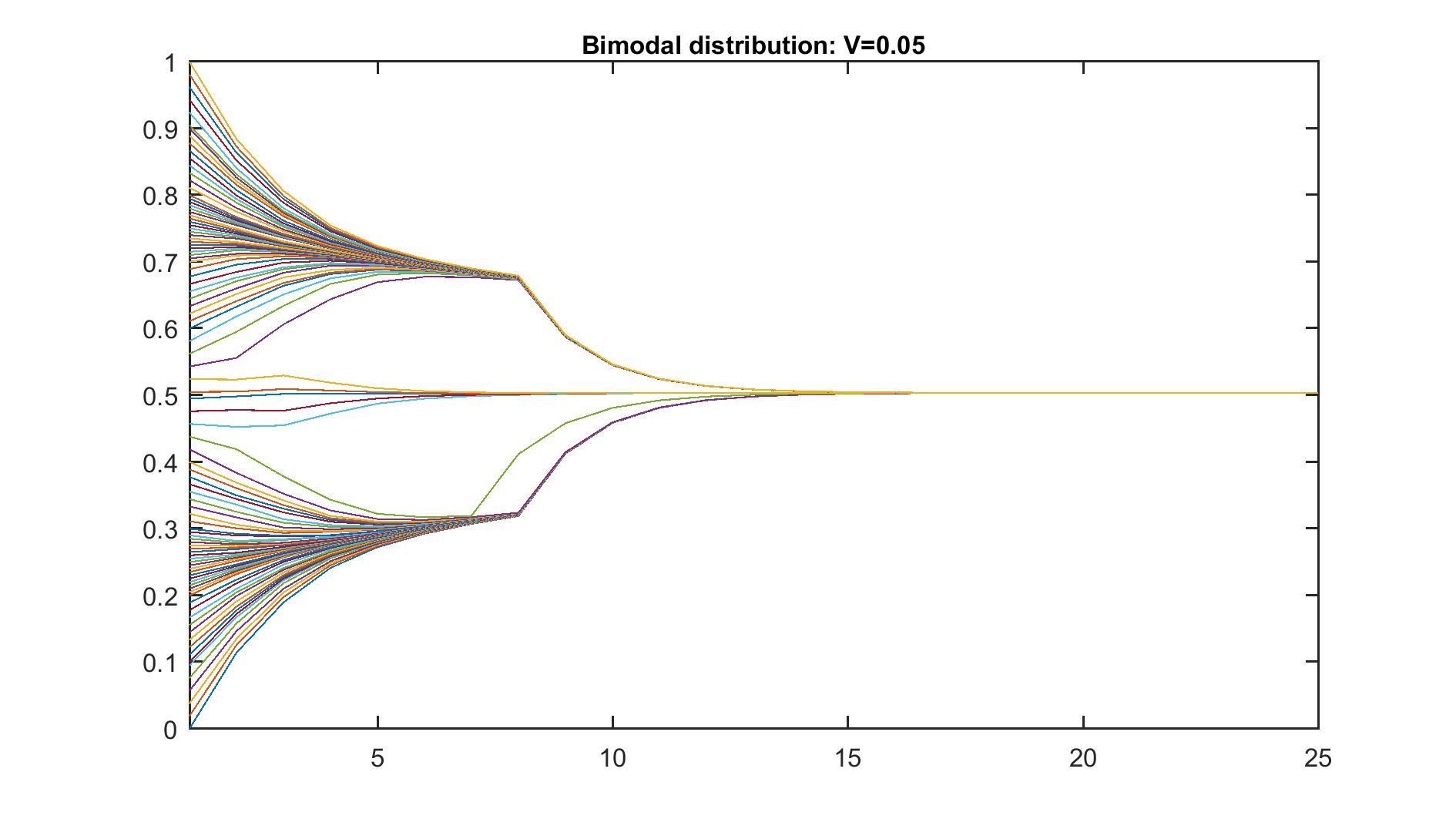}
				\caption{$V=0.05$. Increasing the value of $V$ we finally reach the threshold for consensus to arise in the long-run. }
			\end{subfigure}
			\caption{Analysis of the dynamics starting with a bimodal distribution.}
			\label{fig:bimodal}
		\end{figure}
	\end{example}
	
	The main message from last example is that it is possible to play with the initial distribution to induce disconnectedness between those nodes that are more sparsely distributed in the initial state.
	This happens in non--trivial ways, as we have discussed also for the uniform distribution in previous section.
	We have provided most of our formal proofs only for the case in which the initial distribution of opinions is uniform, because this greatly simplified the analysis.
	However, the uniform distribution is not \emph{per se} the ingredient that is driving us to rich and non easily predictable dynamics, because playing with the initial distribution can only increase the richness of possible dynamics of our model.

	\subsection{The evolution of the network}
	\label{app:networks}
	
	In this section we present the network configurations that emerge from the best response behavior of agents. Besides the purpose to further clarify the main mechanics of the model, we highlight the behavior of un-friending peers even in situations where the network is relatively dense during the initial periods. This phenomenon has been widely documented in real networks, especially during the U.S. Presidential campaign in 2016. Hence, deleting links is an engine for polarization in the model, but we believe this captures meaningful real-world examples which had and can have a significant impact on the evolution of opinions in a society.

	\subsection{Uniform distribution (Fig. \ref{fig:non-monotone})}
	
	First, we present the results of the simulations for the two cases with a uniform distribution, from Figure \ref{fig:non-monotone}.
	Figures \ref{fig:evolution_uni_disconnected1} display the first 6 steps of the evolution for the case that ends with two separated components. From period 6 on the network remains unchanged.\\
	Figures \ref{fig:evolution_uni_connected1} report steps 1-2-3-5-30-31 of the evolution for the case that ends with a single component: the first 5 steps are very similar to the first 5 of the previous case.
	However, the two \emph{groups} of nodes, in this case, remain also connected in all of the following steps  to a single central node that \emph{bridges} the two communities.
	At the end, between the steps 30 and 31 reported in the lower panel of Figure \ref{fig:evolution_uni_connected1}, from this almost bipartite configuration we reach in a single step the complete network.
	This happens because, from step 6 to step 30, the network remains the same but the opinions of nodes converge to the opinion of the bridging node which exerts a gravitational pull for the opinions of the rest of the society.
	
	\begin{figure}
		\centering
		\begin{subfigure}[b]{0.48\textwidth}
			\includegraphics[width=\textwidth]{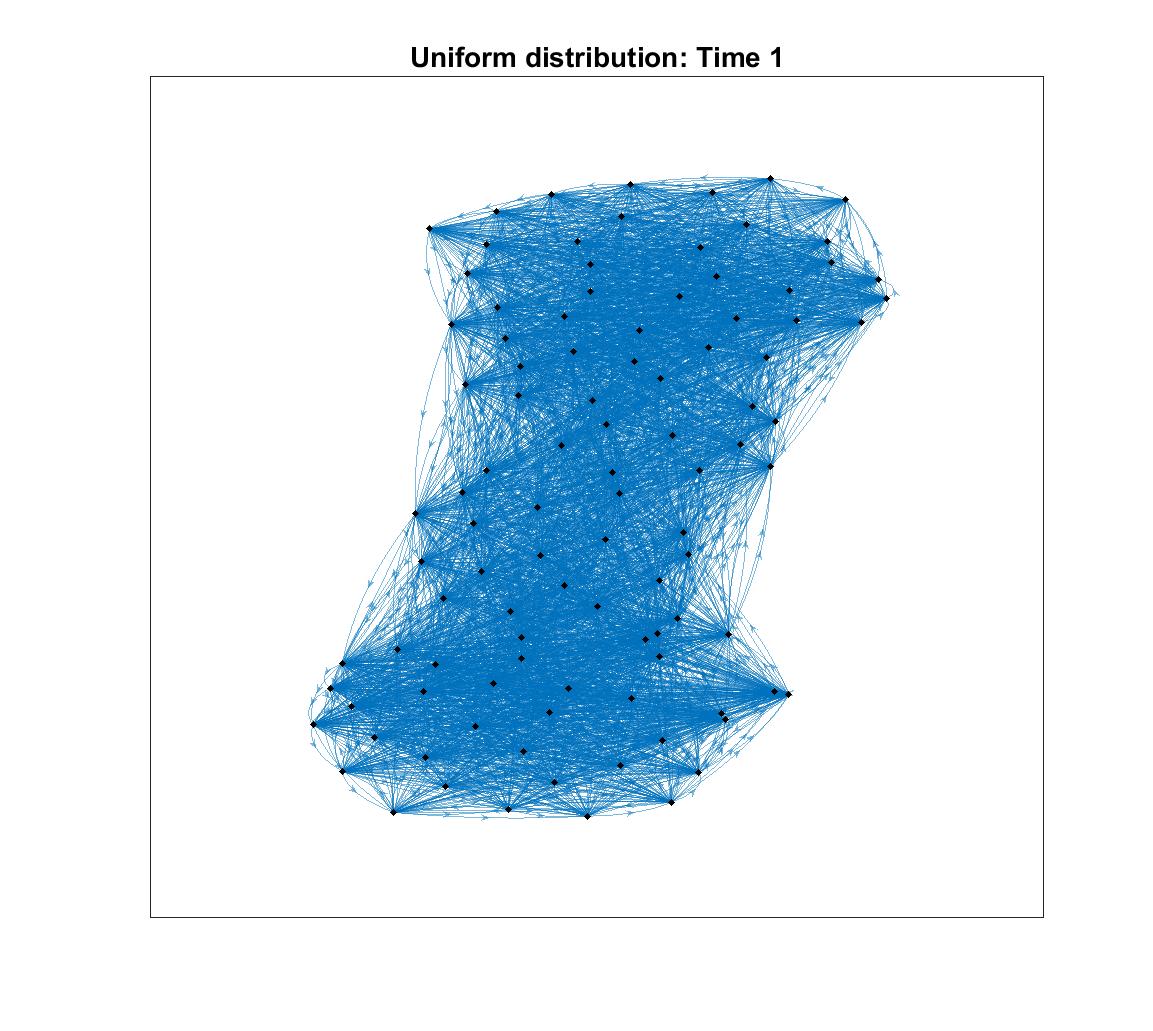}
		\end{subfigure}
		~
		\begin{subfigure}[b]{0.48\textwidth}
			\centering
			\includegraphics[width=\textwidth]{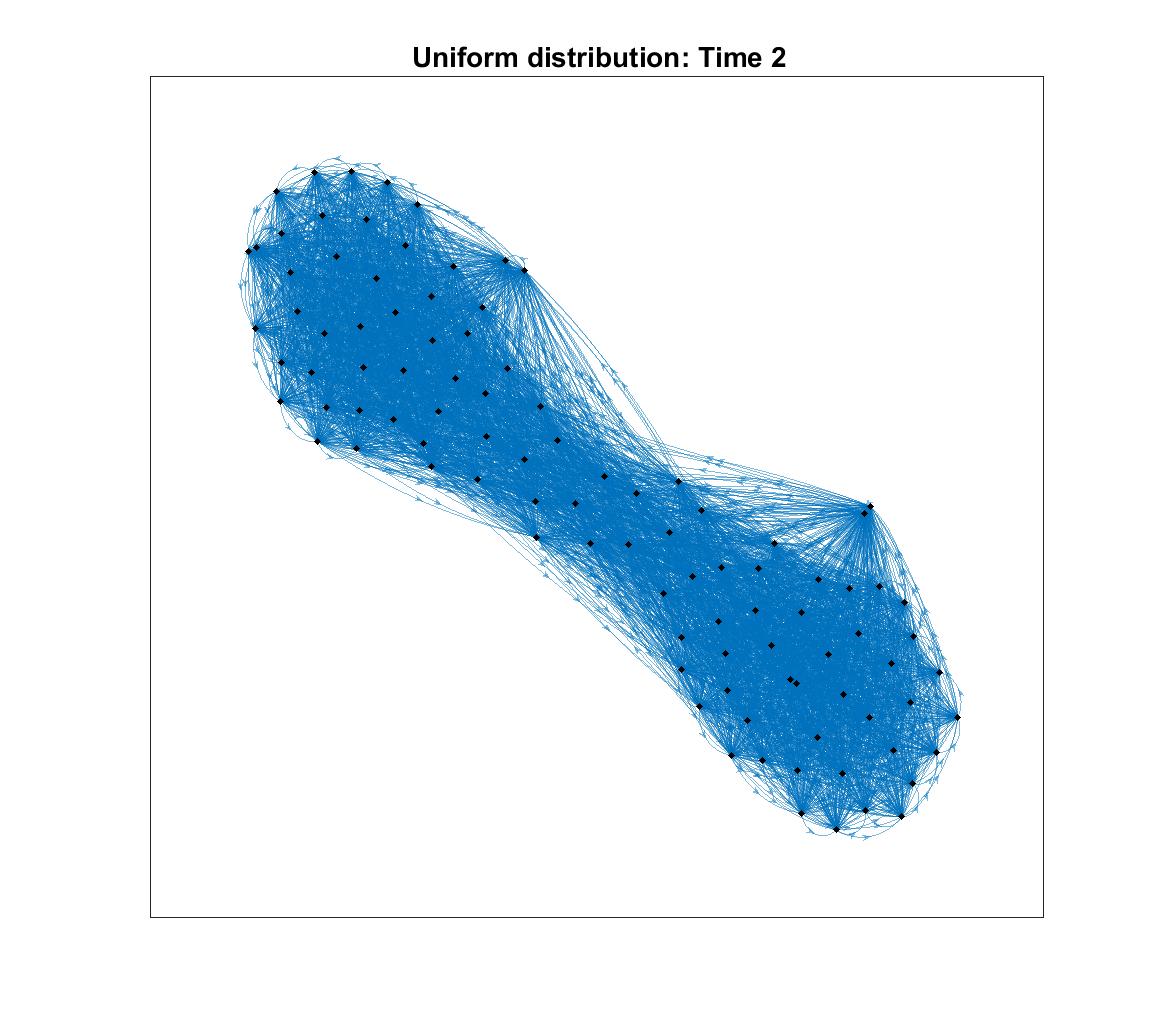}
		\end{subfigure}
		
		\begin{subfigure}[b]{0.48\textwidth}
			\includegraphics[width=\textwidth]{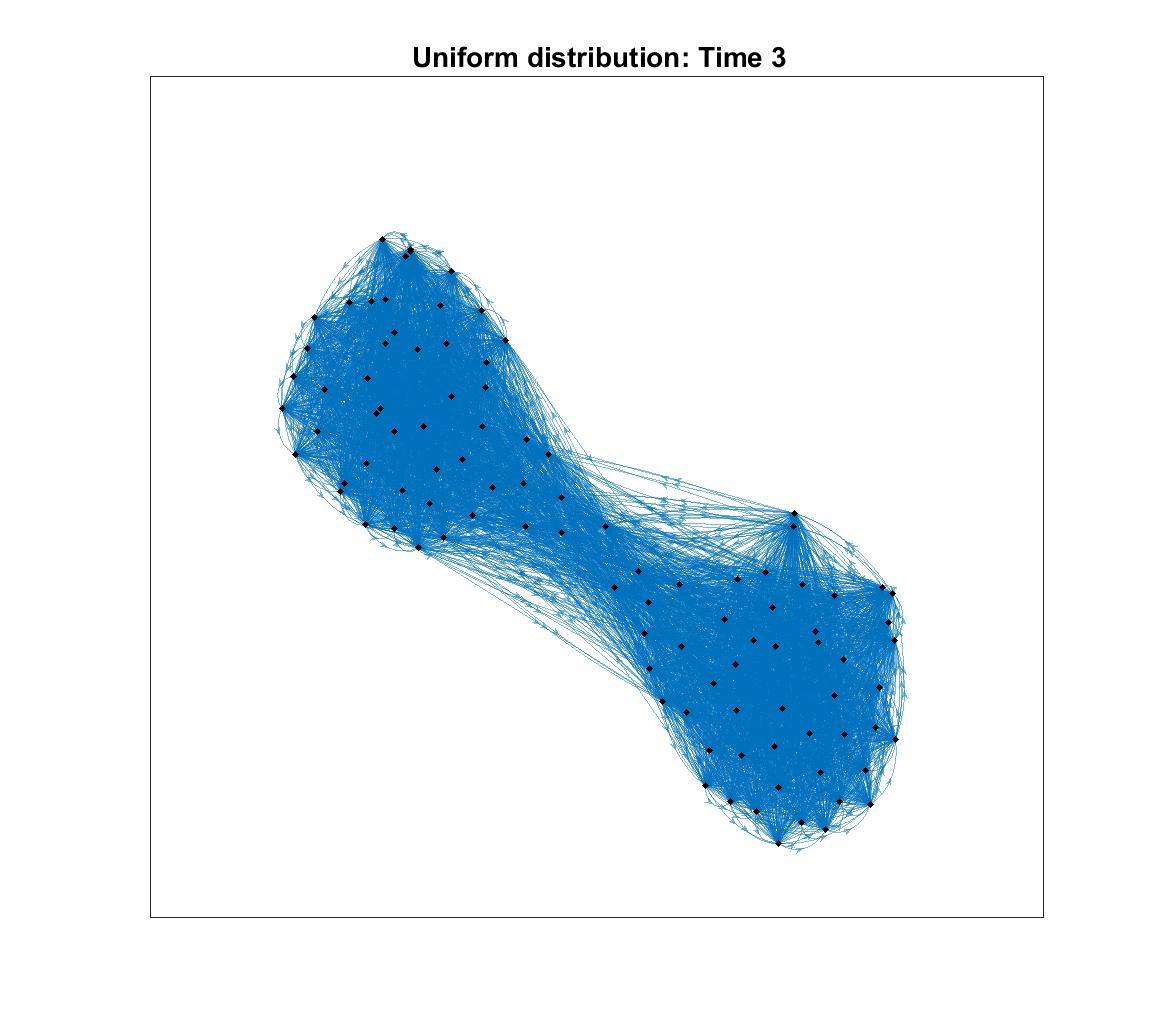}
		\end{subfigure}
		~
		\begin{subfigure}[b]{0.48\textwidth}
			\centering
			\includegraphics[width=\textwidth]{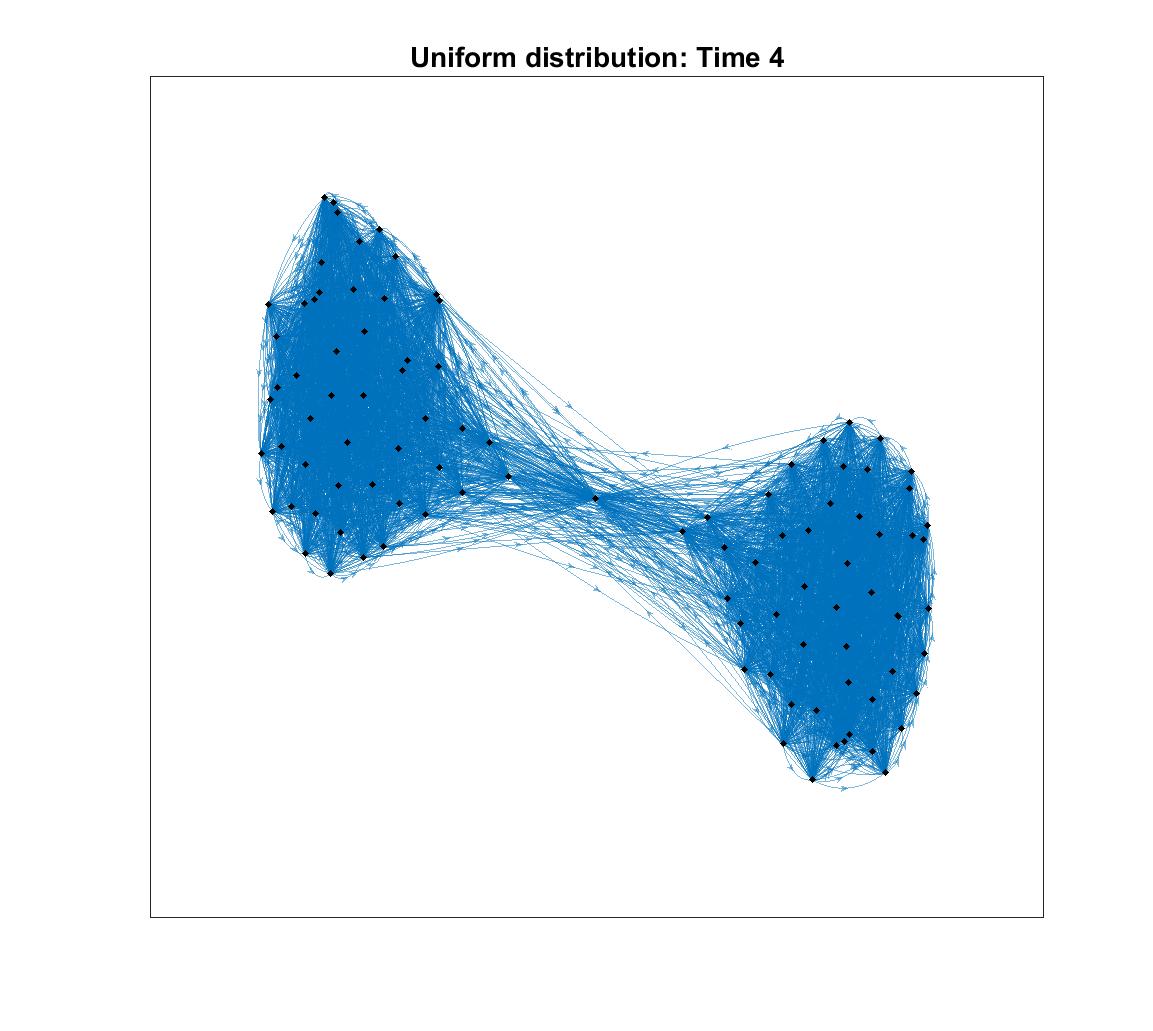}
		\end{subfigure}
		
		\begin{subfigure}[b]{0.48\textwidth}
			\includegraphics[width=\textwidth]{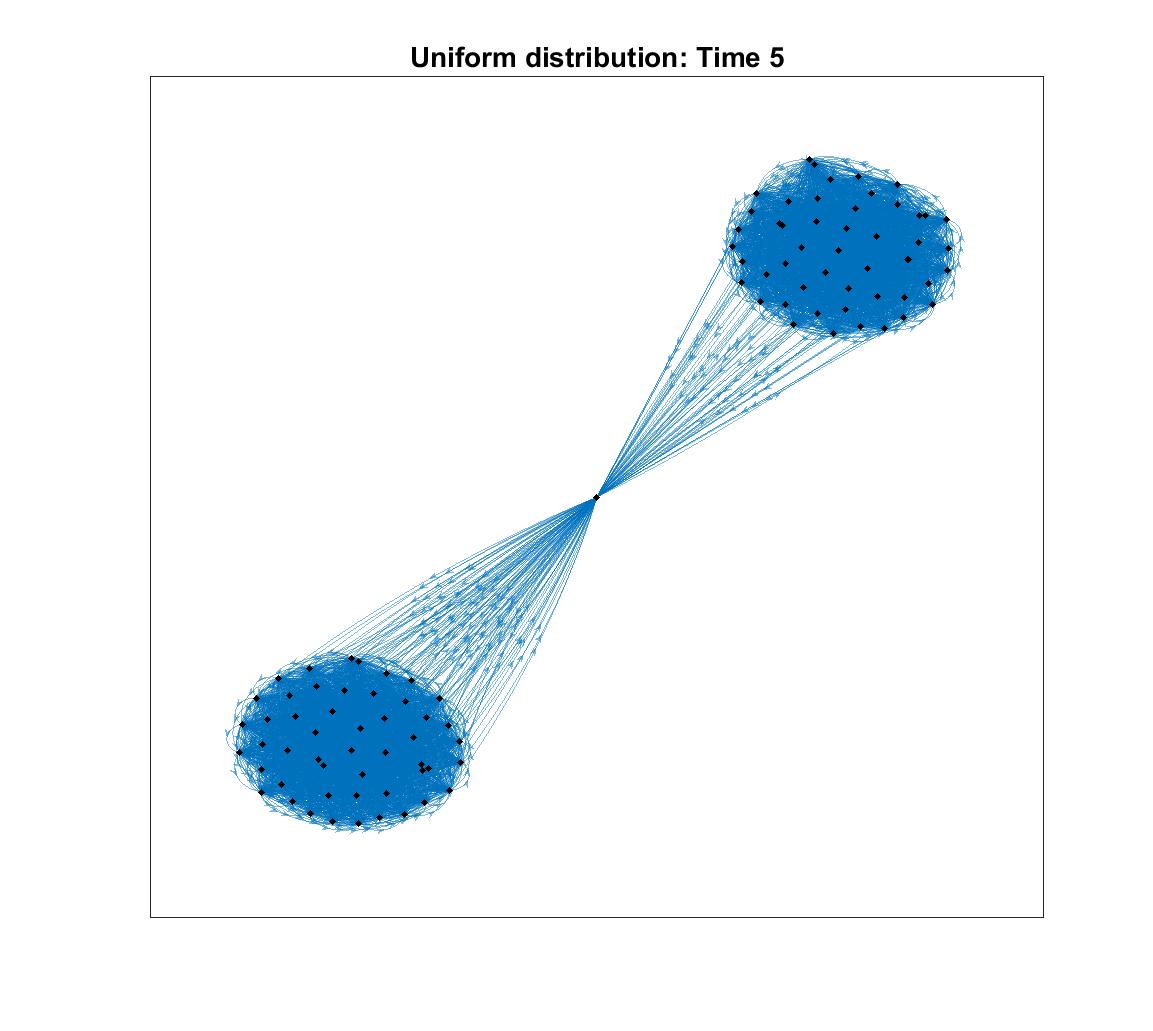}
		\end{subfigure}
		~
		\begin{subfigure}[b]{0.48\textwidth}
			\centering
			\includegraphics[width=\textwidth]{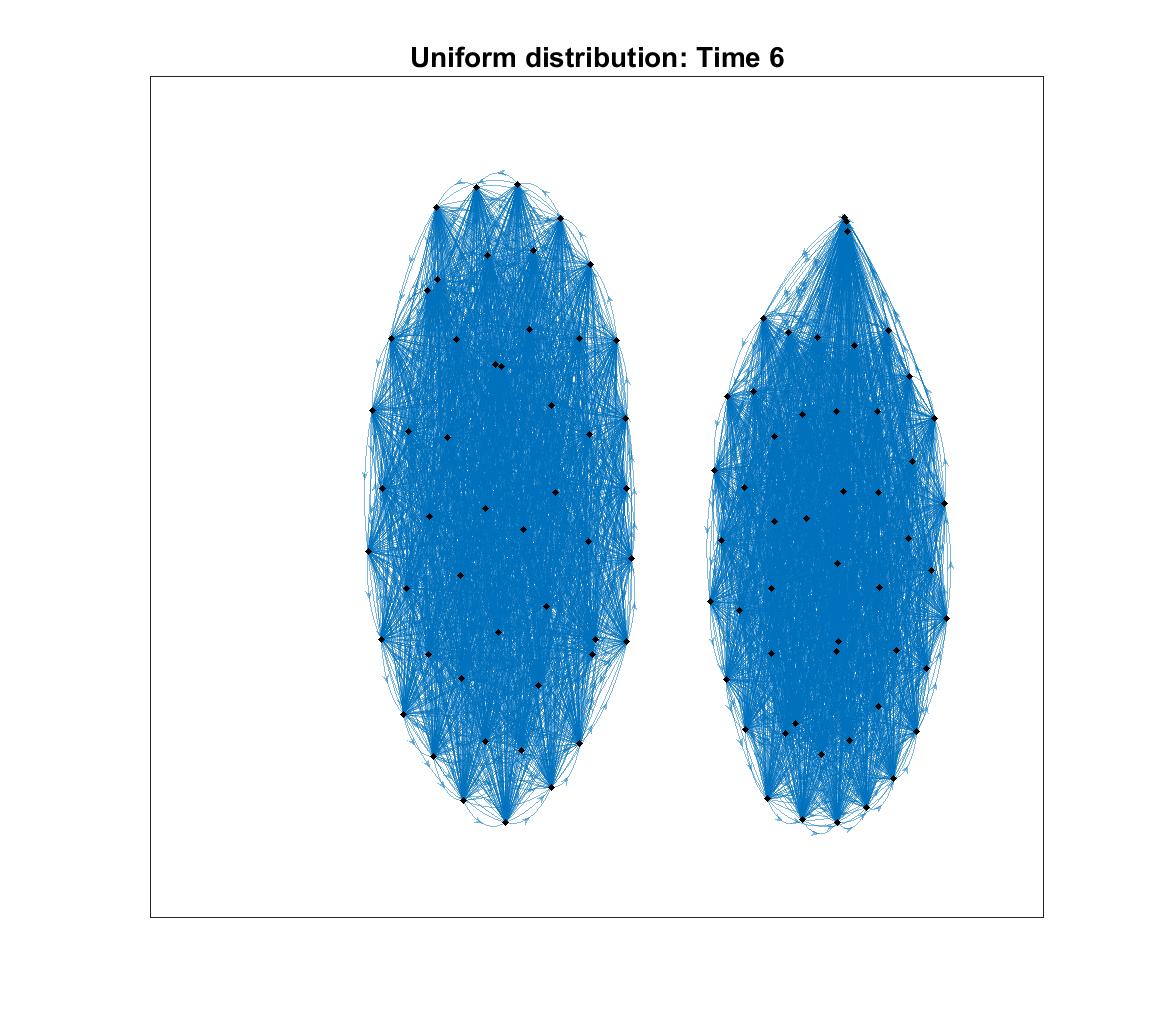}
		\end{subfigure}
		\caption{Evolution of the network: the case of a uniform distribution, $n=101$ agents, $f=0.5$ and $V=0.0350916$. These parameters are the same used in the left panel of Figure \ref{fig:non-monotone}. Periods 1-6. From period 6 on the network remains unchanged. }
		\label{fig:evolution_uni_disconnected1}
	\end{figure}

	\begin{figure}
		\begin{subfigure}[b]{0.48\textwidth}
			\centering
			\includegraphics[width=\textwidth]{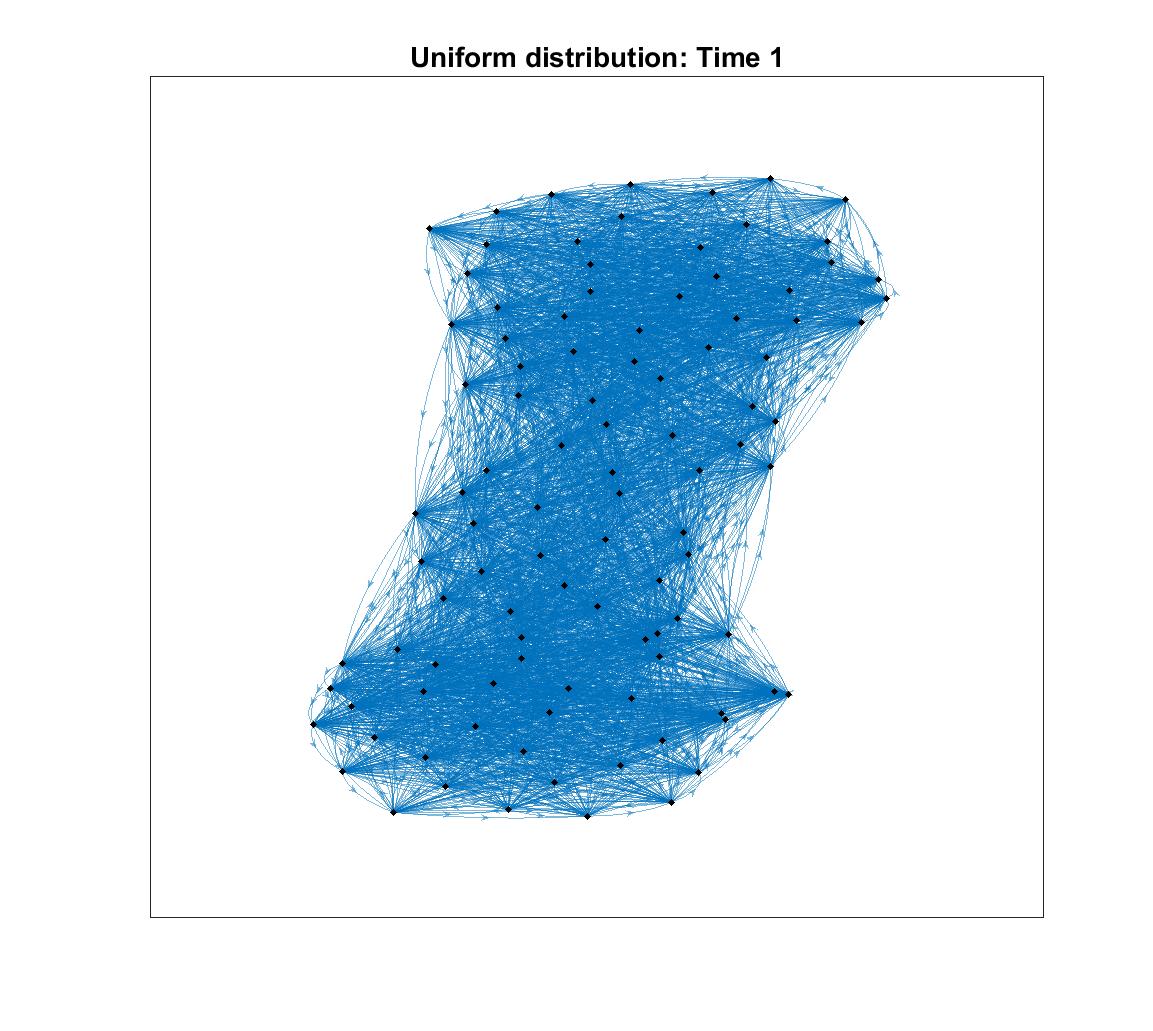}
		\end{subfigure}
		~
		\begin{subfigure}[b]{0.48\textwidth}
			\centering
			\includegraphics[width=\textwidth]{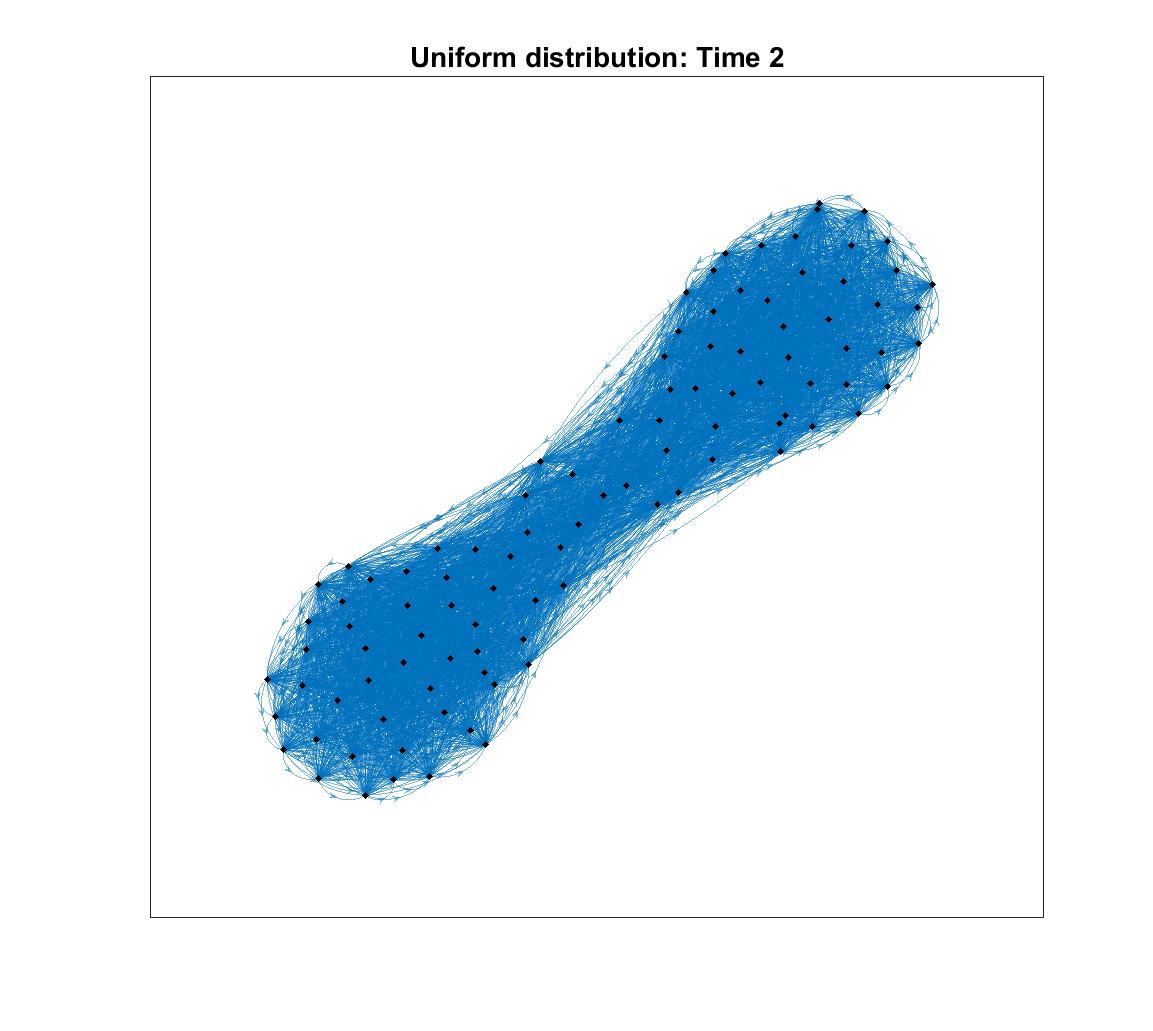}
		\end{subfigure}
		~
		\begin{subfigure}[b]{0.48\textwidth}
			\centering
			\includegraphics[width=\textwidth]{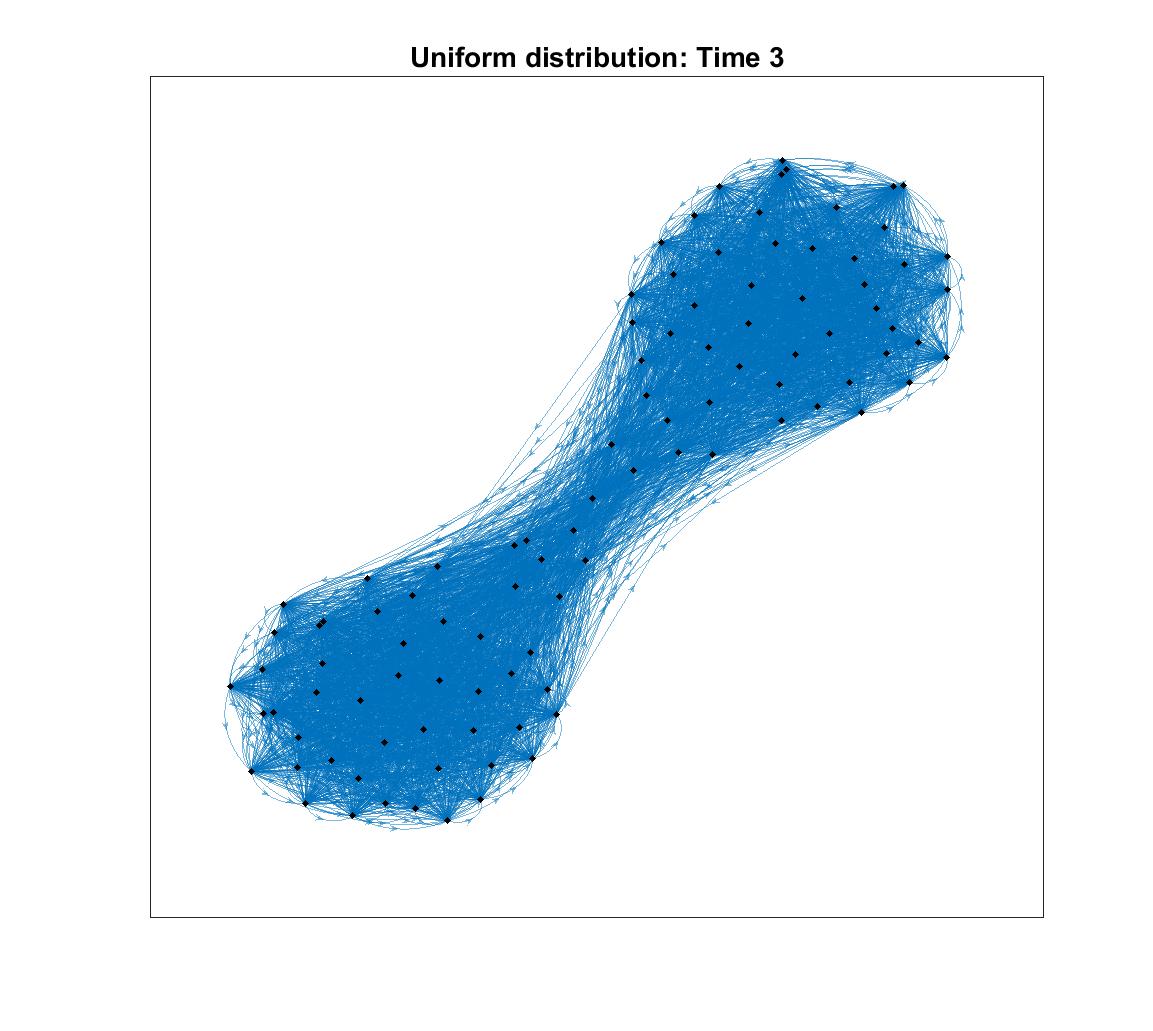}
		\end{subfigure}
		~
		\begin{subfigure}[b]{0.48\textwidth}
			\centering
			\includegraphics[width=\textwidth]{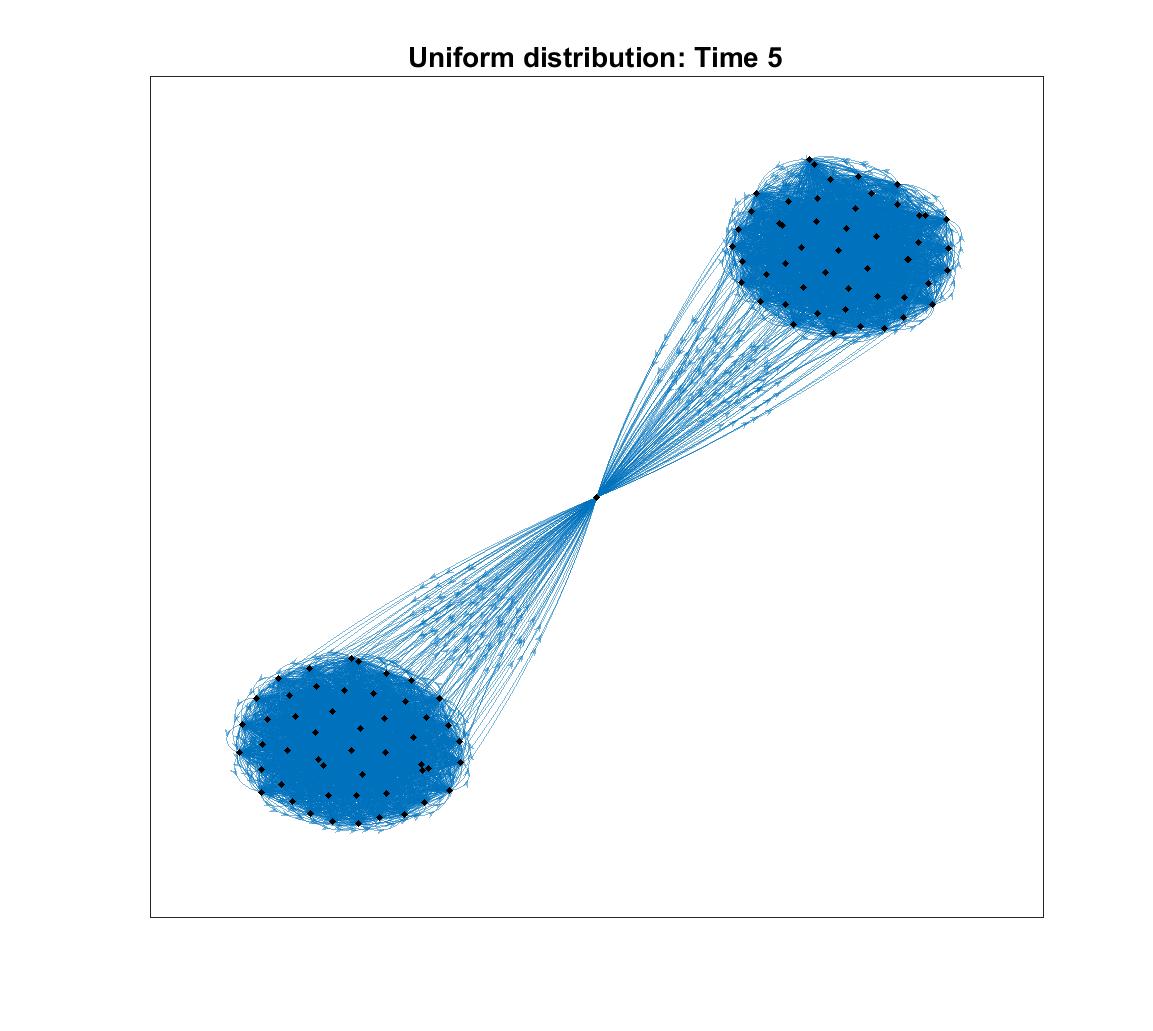}
		\end{subfigure}
		~
		\begin{subfigure}[b]{0.48\textwidth}
			\centering
			\includegraphics[width=\textwidth]{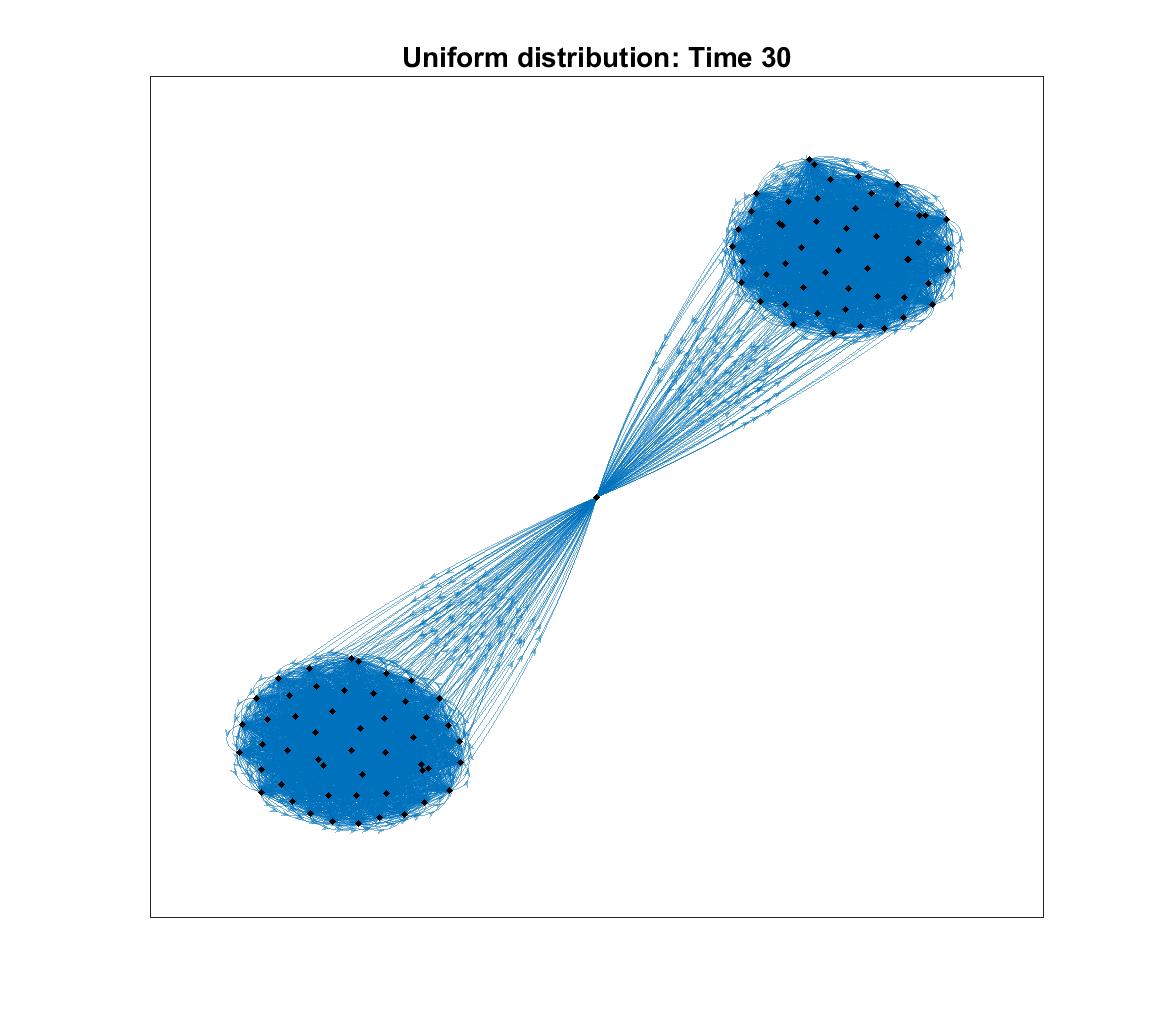}
		\end{subfigure}
		~
		\begin{subfigure}[b]{0.48\textwidth}
			\centering
			\includegraphics[width=\textwidth]{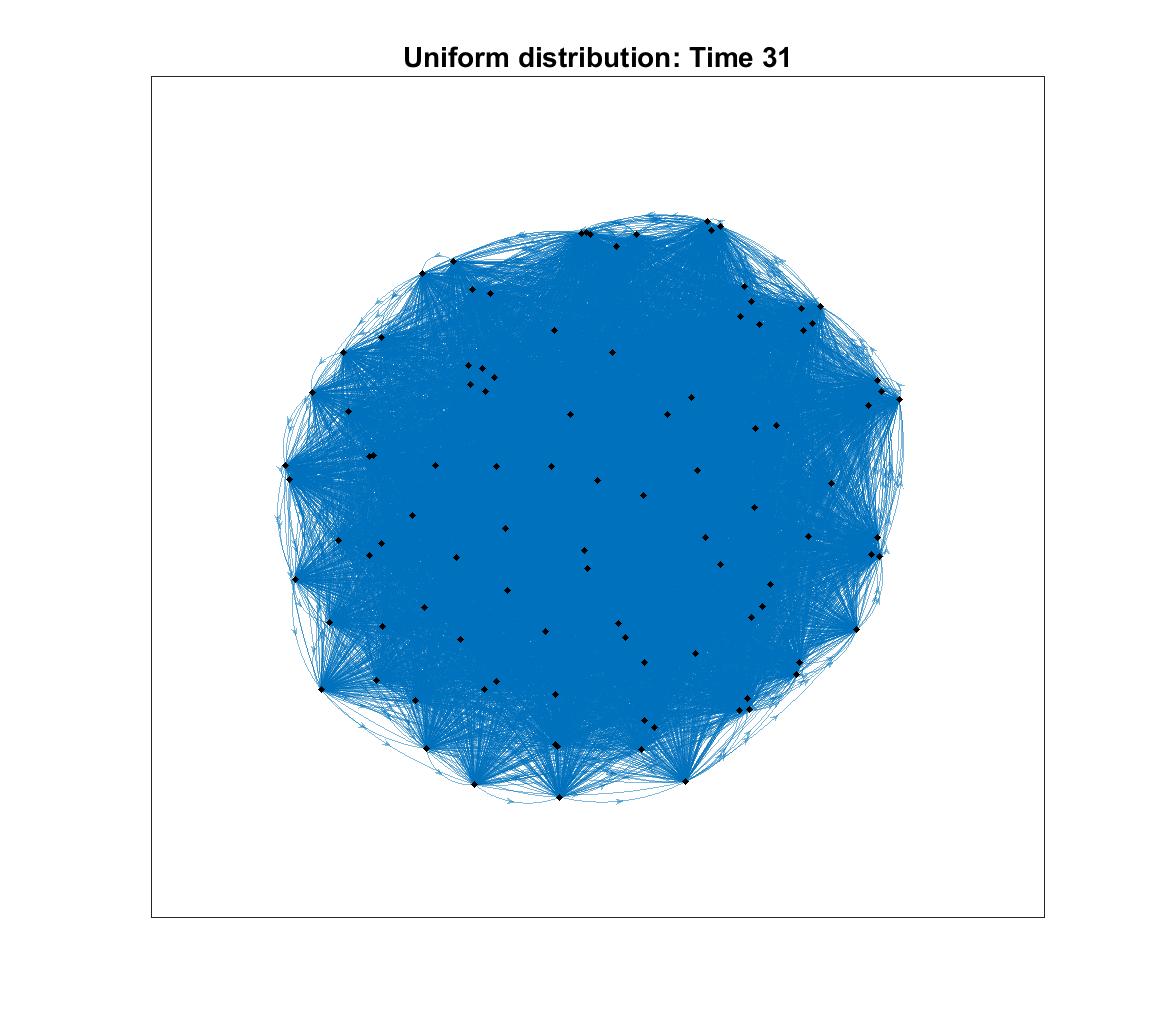}
		\end{subfigure}
		\caption{Evolution of the network: the case of a uniform distribution, $n=101$ agents, $f=0.5$ and $V=0.0350917$. These parameters are the same used in the right panel of Figure \ref{fig:non-monotone}. Periods 1-2-3-5-30-31. The network remains unchanged from period 5 to 30. In period 31 the number of links explodes and the complete network takes place.}
		\label{fig:evolution_uni_connected1}
	\end{figure}

	
	\subsection{Normal distribution (Fig. \ref{fig:gaussian})}
	Finally, we report the results of the simulations for the case with a normal distribution, from Figure \ref{fig:gaussian}. We highlight the initial high density of the network. However, over time, the central group become denser, as agents find abundance of profitable connections. On the other hand, the tails react less initially, and so this creates a discrepancy in their opinions. Eventually, The extreme groups will find themselves cutting the links with the moderate group and create a separate component. When the process stabilizes, in each component agents form the complete network. There is therefore partial consensus, as each component ends up having the same opinion. However in the overall society we have disagreement, which persists in the long-run.
 
	\begin{figure}[ht]
		\centering
		\begin{subfigure}[b]{0.48\textwidth}
			\includegraphics[width=\textwidth]{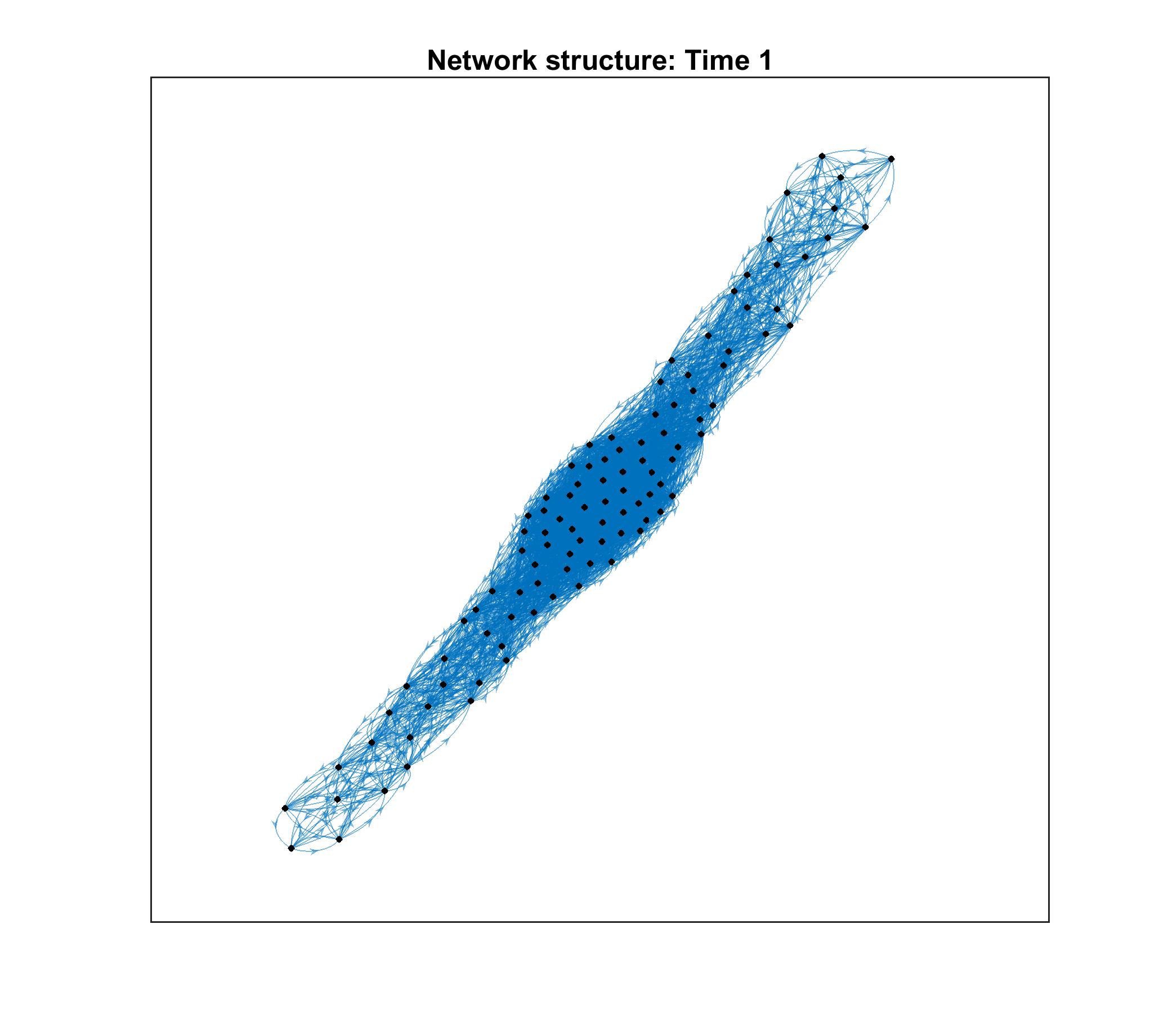}
		\end{subfigure}
		~
		\begin{subfigure}[b]{0.48\textwidth}
			\centering
			\includegraphics[width=\textwidth]{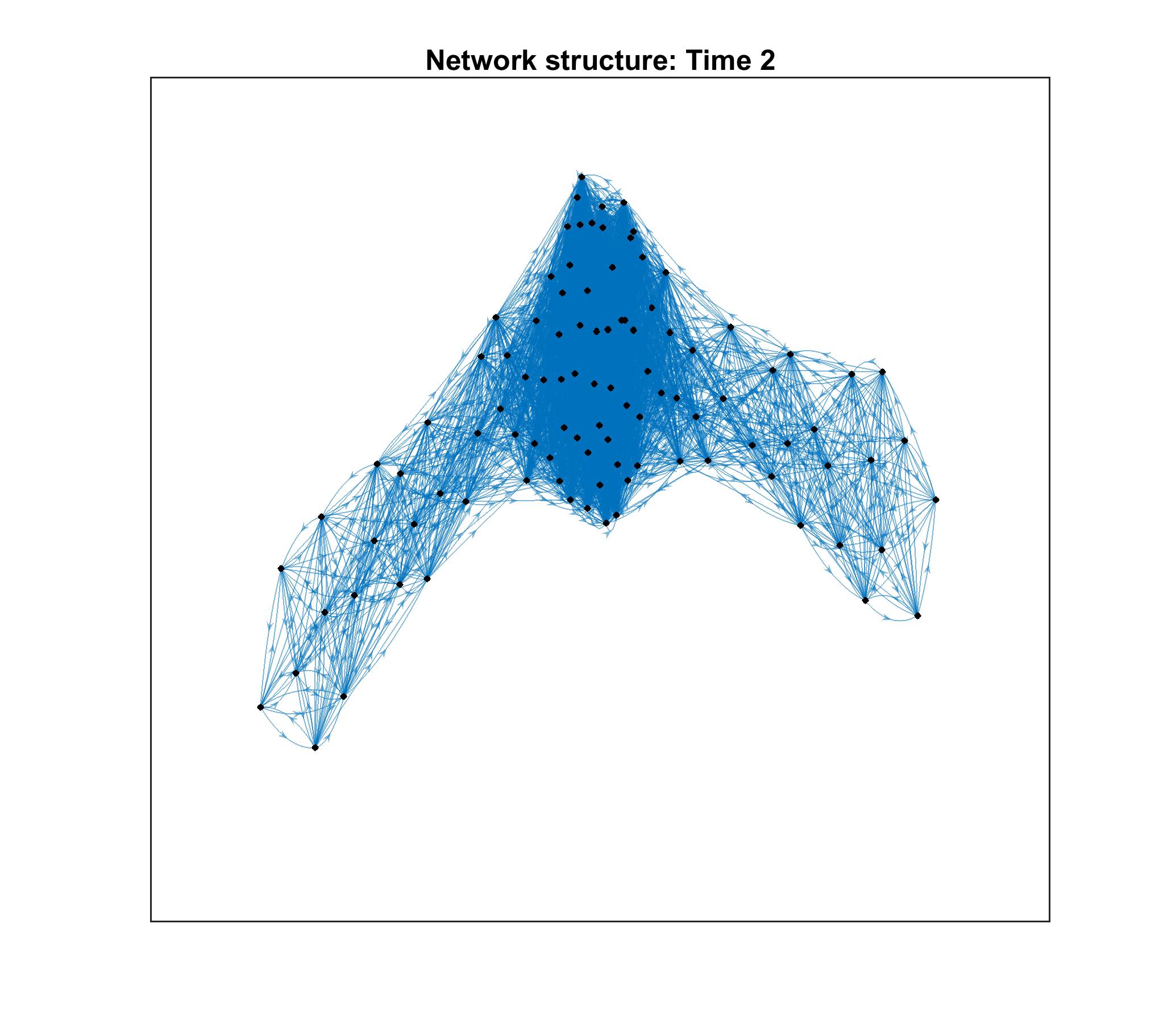}
		\end{subfigure}
		
		\begin{subfigure}[b]{0.48\textwidth}
			\includegraphics[width=\textwidth]{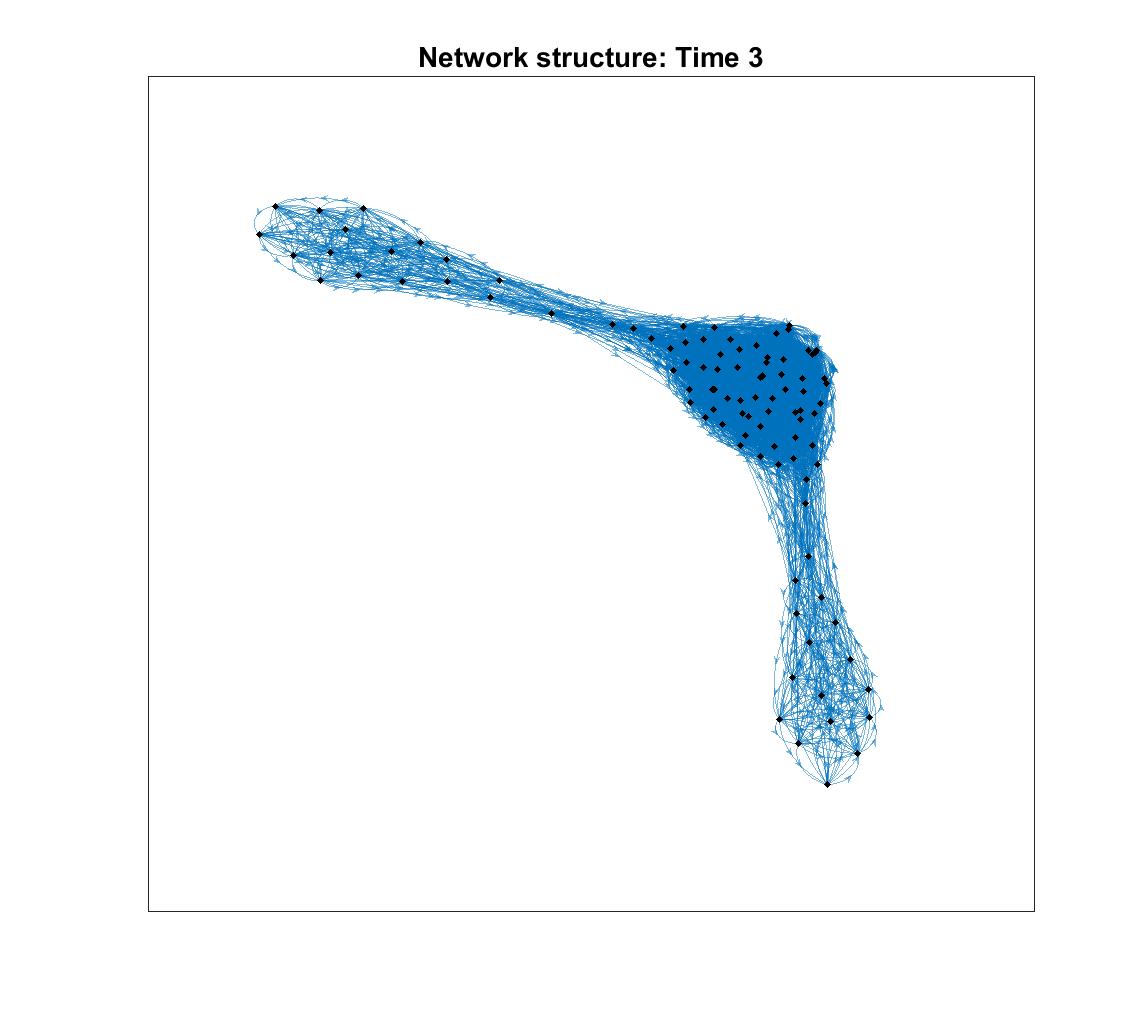}
		\end{subfigure}
		~
		\begin{subfigure}[b]{0.48\textwidth}
			\centering
			\includegraphics[width=\textwidth]{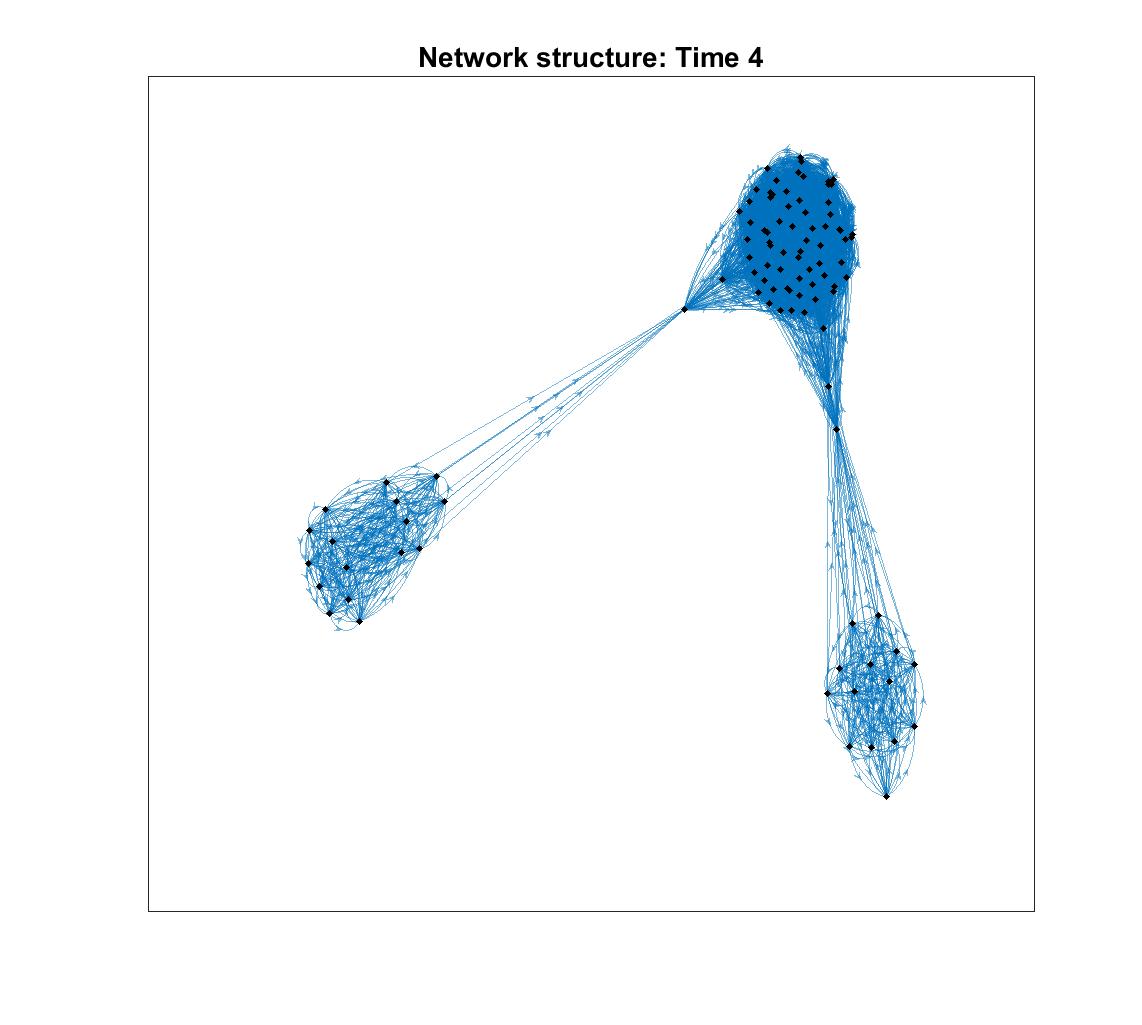}
		\end{subfigure}
		
		\begin{subfigure}[b]{0.48\textwidth}
			\includegraphics[width=\textwidth]{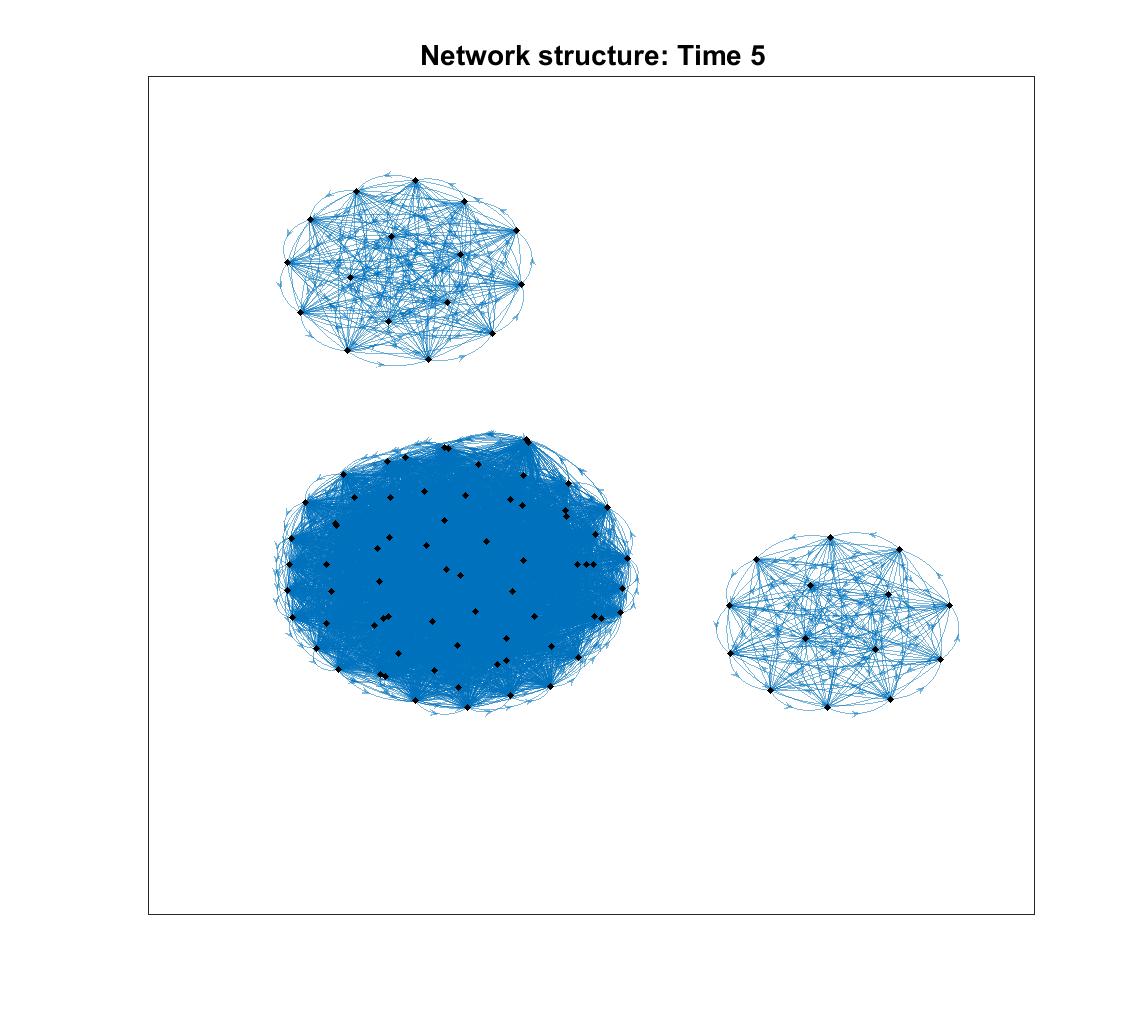}
		\end{subfigure}
		~
		\begin{subfigure}[b]{0.48\textwidth}
			\centering
			\includegraphics[width=\textwidth]{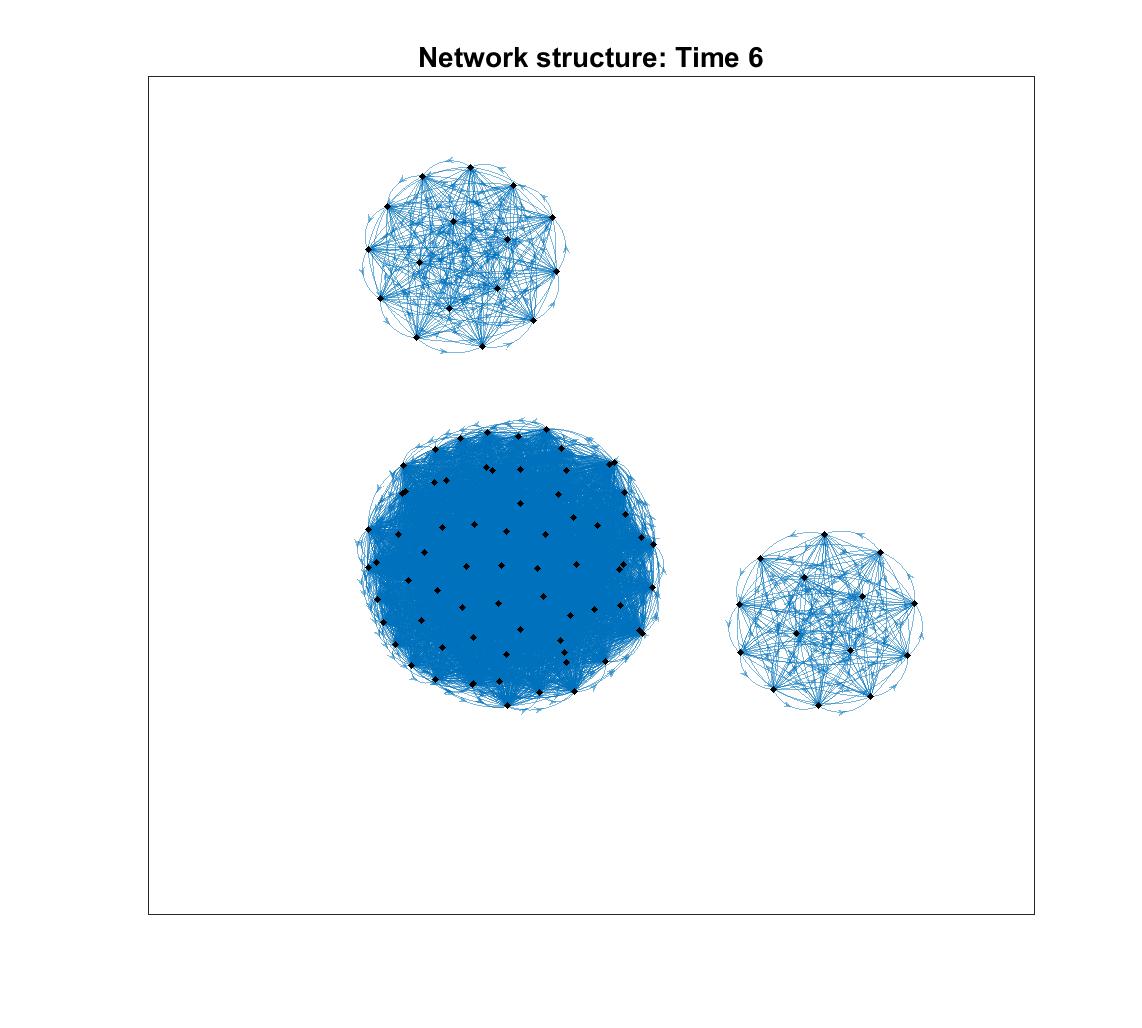}
		\end{subfigure}
		\caption{Evolution of the network: the case of a normal distribution, $n=101$ agents, $f=0.5$ and $V=0.01$. The parameters are the same of those used in the left panel of Figure \ref{fig:gaussian}.}
		\label{fig:evolution_normal}
	\end{figure}

\clearpage

 \section{About polarization}\label{app:polarization_normative}
\label{app:polarization}


To measure polarization in our framework, we partition the population into $s$ segments. Hence, we calculate the fraction of the population whose opinion lies within each segment, in a given point in time. 
Following \cite{esteban1994measurement} we measure polarization using the following formula:
	
\begin{equation}\label{eq:polarization}
	P(\rho,\textbf{x})=K\sum_{i=1}^{s}\sum_{j=1}^{s}\rho_i^{1+\alpha}\rho_j|\bar{x}_i-\bar{x}_j|
\end{equation}

where $\bar{x}_i$ with $i\in\{1,\ldots,s\}$ is the average opinion of the agents in the segment $i$ of the partition. Furthermore, the polarization measure requires two additional parameters, $K$ and $\alpha$. While $K$ is just a shifter and matters for re-scaling purposes only, the parameter $\alpha$ is what captures the essence of polarization. If $\alpha$ is equal to 0, the measure in equation \eqref{eq:polarization} coincides with the Gini index. Thus, as in \cite{esteban1994measurement}, we maintain the restriction $\alpha\in(0,1.6]$.\footnote{See Theorem 1 in \cite{esteban1994measurement}.} Using this measure of polarization, in Figure \ref{fig:non-monotone}, we show the numerical values associated with the simulations we displayed in section \ref{sec:nonmono}.

\begin{figure}[ht!]
		\centering

		\begin{subfigure}[t]{0.48\textwidth}
			\centering
			\includegraphics[width=\textwidth]{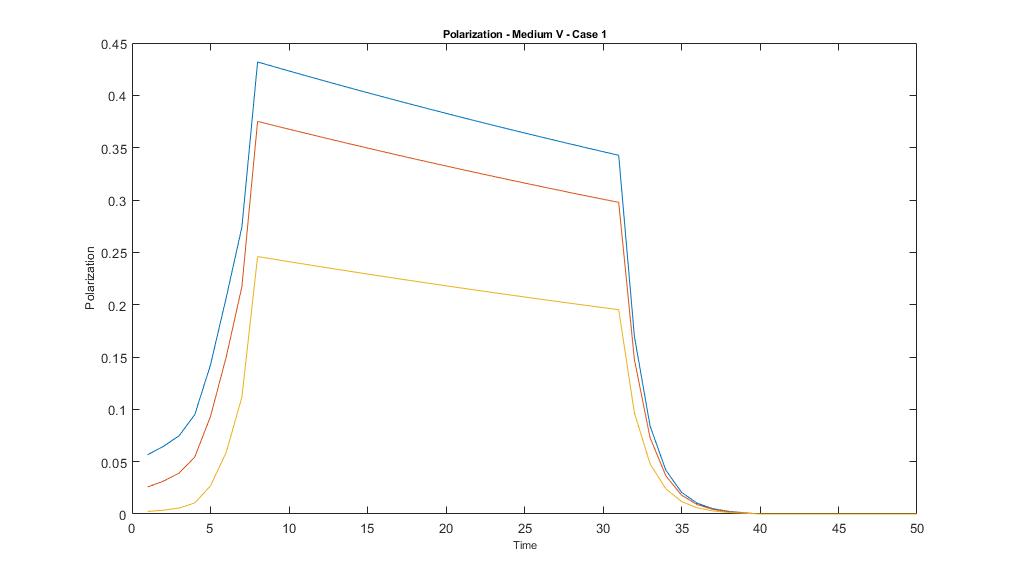}
			\caption{$f=0.5$ and $V=0.0350917$: Passed a threshold value, the opinions converge in the long run. Polarization suddenly increases, then slowly decreases until the network becomes fully connected and goes to 0, i.e. opinions have converged. Each curve is calculated with a different level of $\alpha$, respectively 0.8, 1, and 1.6.}
		\end{subfigure}
    		~
		\begin{subfigure}[t]{0.48\textwidth}
			\includegraphics[width=\textwidth]{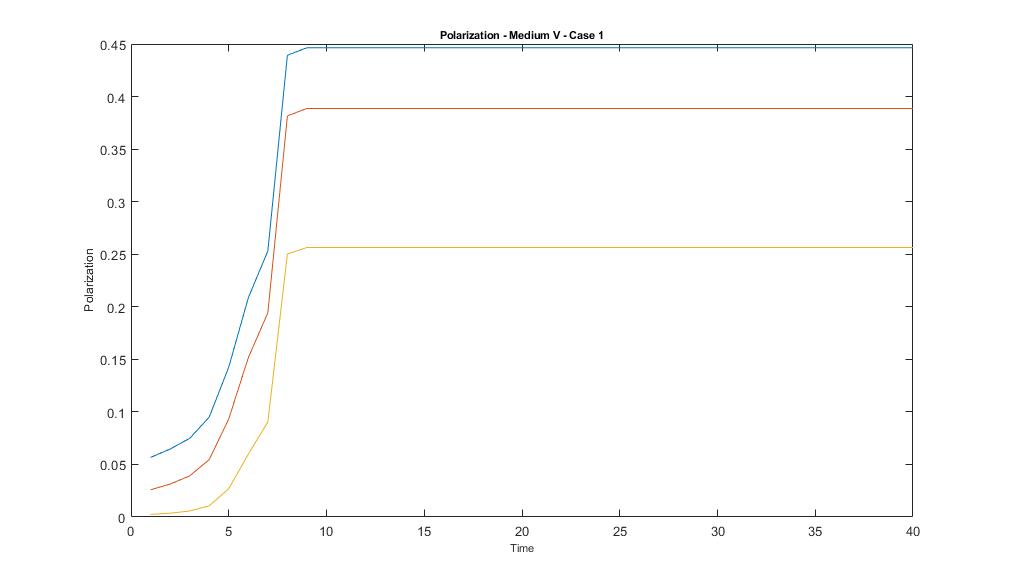}
			\caption{$f=0.5$ and $V=0.0350916$: The network disconnects. Polarization increases when agents condensate towards the extremes of the distribution of opinions. Then it reaches its maximum value. Each curve is calculated with a different level of $\alpha$, respectively 0.8, 1, and 1.6.}
		\end{subfigure}

		\caption{Plot of the dynamics starting of polarization with a uniform distribution, $n=101$ and intermediate values of $V$.}
	\end{figure}

	\bibliographystyle{chicago}
	\bibliography{biblio}

\end{document}